%% file: main.tex
\documentclass[11pt]{article}

\usepackage{amsthm}
\usepackage{graphicx} % support the \includegraphics command and options
\usepackage{array} % for better arrays (eg matrices) in maths

\usepackage{amsmath, amssymb, amsfonts, verbatim}
\usepackage{hyphenat,epsfig,subcaption,multirow}
\usepackage{nicefrac}
\usepackage{paralist}
\usepackage{thm-restate}
\usepackage{mleftright}

\usepackage[font=small, labelfont=bf]{caption}

\usepackage[dvipsnames]{xcolor}
\usepackage{pifont}   % for \ding symbols

\usepackage[ruled]{algorithm2e}

\DeclareFontFamily{U}{mathx}{\hyphenchar\font45}
\DeclareFontShape{U}{mathx}{m}{n}{
      <5> <6> <7> <8> <9> <10>
      <10.95> <12> <14.4> <17.28> <20.74> <24.88>
      mathx10
      }{}
\DeclareSymbolFont{mathx}{U}{mathx}{m}{n}
\DeclareMathSymbol{\bigtimes}{1}{mathx}{"91}

\usepackage{tcolorbox}
\tcbuselibrary{skins,breakable}
\tcbset{enhanced jigsaw}

\usepackage[normalem]{ulem}
\usepackage[compact]{titlesec}

\definecolor{DarkRed}{rgb}{0.5,0.1,0.1}
\definecolor{DarkBlue}{rgb}{0.1,0.1,0.5}

\definecolor{colorOne}{HTML}{E69F00}   % orange
\definecolor{colorTwo}{HTML}{56B4E9}   % sky blue
\definecolor{colorThree}{HTML}{009E73} % bluish green
\definecolor{colorFour}{HTML}{F0E442}  % yellow
\definecolor{colorFive}{HTML}{0072B2}  % blue
\definecolor{colorSix}{HTML}{D55E00}   % vermillion
\definecolor{colorSeven}{HTML}{CC79A7} % reddish purple
\definecolor{colorEight}{HTML}{000000} % black
\definecolor{colorNine}{HTML}{999999}  % gray

\usepackage{nameref}
\definecolor{ForestGreen}{rgb}{0.1333,0.5451,0.1333}
\definecolor{Red}{rgb}{0.9,0,0}
\usepackage[linktocpage=true,
	pagebackref=true,colorlinks,
		urlcolor=black,
	linkcolor=DarkRed,citecolor=ForestGreen,
	bookmarks,bookmarksopen,bookmarksnumbered]
	{hyperref}
\usepackage[noabbrev,nameinlink]{cleveref}
\crefname{property}{property}{Property}
\creflabelformat{property}{(#1)#2#3}
\crefname{equation}{eq}{Eq}
\creflabelformat{equation}{(#1)#2#3}

\usepackage{bm}
\usepackage{url}
\usepackage{xspace}
\usepackage[mathscr]{euscript}

\usepackage{environ}  % For creating new environments
\usepackage{chngcntr} % For changing counters

\usepackage{tikz}
\usetikzlibrary{arrows}
\usetikzlibrary{arrows.meta}
\usetikzlibrary{shapes}
\usetikzlibrary{backgrounds}
\usetikzlibrary{positioning}
\usetikzlibrary{decorations.markings}
\usetikzlibrary{decorations.pathmorphing}
\usetikzlibrary{patterns}
\usetikzlibrary{calc}
\usetikzlibrary{fit}
\tikzset{vertexA/.style={circle,fill=black,minimum size=5pt,inner sep=0pt, font=\small}}
\tikzstyle{vertexB}=[circle,fill=black,minimum size=3pt,inner sep=0pt, font=\small]
\tikzstyle{selected vertex} = [vertex, fill=red!24]
\tikzstyle{edge} = [draw, black!50, -]
\tikzstyle{vertexRed}=[circle,fill=red,draw,minimum size=10pt,inner sep=0pt, font=\small, line width=0.5pt]
\tikzstyle{vertexBlack}=[circle,fill=black,draw,minimum size=10pt,inner sep=0pt, font=\small, line width=0.5pt]

\usepackage{mdframed}

\usepackage[noend]{algpseudocode}
\makeatletter
\def\BState{\State\hskip-\ALG@thistlm}
\makeatother

\usepackage{cite}
\usepackage{enumitem}

\usepackage[margin=1in]{geometry}

\usepackage{float} % Add this in your preamble

\newtheorem{theorem}{Theorem}
\newtheorem{lemma}{Lemma}[section]

\newtheorem{proposition}[lemma]{Proposition}
\newtheorem{corollary}[lemma]{Corollary}
\newtheorem{claim}[lemma]{Claim}

\newtheorem*{claim*}{Claim}
\newtheorem*{proposition*}{Proposition}
\newtheorem*{lemma*}{Lemma}
\newtheorem*{problem*}{Problem}

\crefname{lemma}{Lemma}{Lemmas}
\crefname{claim}{Claim}{Claims}

\newtheorem{mdresult}{Result}

\newtheorem*{mdresult*}{Main Result}

\newtheorem{definition}[lemma]{Definition}
\theoremstyle{definition}
\newtheorem{remark}[lemma]{Remark}
\newtheorem{observation}[lemma]{Observation}

\newtheorem{mdinvariant}[lemma]{Lemma}

\theoremstyle{definition}
\newtheorem{mdalg}{Algorithm}
\newenvironment{Algorithm}{\begin{tbox}\begin{mdalg}}{\end{mdalg}\end{tbox}}

\allowdisplaybreaks

\renewcommand{\qed}{\nobreak \ifvmode \relax \else
      \ifdim\lastskip<1.5em \hskip-\lastskip
      \hskip1.5em plus0em minus0.5em \fi \nobreak
      \vrule height0.75em width0.5em depth0.25em\fi}

\renewcommand{\leq}{\leqslant}
\renewcommand{\geq}{\geqslant}

\input{macros}

\title{Sublinear-Time Lower Bounds for Approximating Matching Size using Non-Adaptive Queries}
\author{
Vihan Shah\footnote{(\href{mailto:vihanshah98@gmail.com}{\text{vihanshah98@gmail.com}}) 
Cheriton School of Computer Science, University of Waterloo. }  
}

\date{}

\begin{document}
\maketitle

%%\vspace{-0.55cm}
\thispagestyle{empty}
\pagenumbering{roman}

\input{abstract}
%
%
\clearpage
\bigskip
\setcounter{tocdepth}{3}
\tableofcontents

\clearpage

\pagenumbering{arabic}
\setcounter{page}{1}

\input{intro}

\input{prelim}

\input{log-apx-LB-AdjList}

\input{structure.tex}

\input{coupling.tex}

\input{algorithm}

\input{tree-query}

\section{Open Problems}

Our work provides the first non-adaptive bounds for a fundamental graph problem, establishing that adaptivity is necessary to achieve strong approximations for estimating the maximum matching size in sublinear time. However, several intriguing questions remain open.

\begin{enumerate}
	\item \textbf{Tight Bounds for Non-Adaptive Algorithms.}  
	Our bounds leave a polynomial gap between the achievable $n^{1/2}$-approximation and the $n^{1/3}$-approximation ruled out by our lower bound. Closing this gap, either by improving the lower bound or by designing better non-adaptive algorithms, remains a central open direction.
	
	\item \textbf{Intermediate Levels of Adaptivity.}  
	The literature to date has studied fully adaptive algorithms for matching size, and this work studies the fully non-adaptive setting. What happens in between? In particular, our lower bound does not extend to algorithms with two rounds of non-adaptive queries, leaving open what can be achieved in two rounds, or more generally with $r$ rounds of non-adaptive queries? Understanding this spectrum of partial adaptivity remains an important direction for future work. More concretely, can we achieve an $n^{\eps}$-approximation in $\Theta(1)$ rounds or a $\polylog{n}$-approximation in $\polylog{n}$ rounds?
	
	\item \textbf{Beyond Matching and Vertex Cover.}  
	Prior to this work, non-adaptive algorithms in the sublinear-time setting had not been studied. We initiate this study for matching and vertex cover size. It is natural to ask how other problems, such as maximum flow, minimum cut, and vertex connectivity, behave under non-adaptivity or low adaptivity.
	
\end{enumerate}

\subsection*{Acknowledgments}
The author is deeply grateful to Sepehr Assadi for insightful discussions and valuable guidance throughout the project. The author also thanks Janani Sundaresan and Eric Blais for their helpful feedback and conversations.

\bibliographystyle{alpha}
\bibliography{reference}

\end{document}

%% file: macros.tex
\newcommand{\Ot}{\ensuremath{\widetilde{O}}}

\newcommand{\eps}{\ensuremath{\varepsilon}}

\newcommand{\bracket}[1]{\left[#1\right]}
\newcommand{\paren}[1]{\ensuremath{\left(#1\right)}\xspace}
\newcommand{\card}[1]{\left\vert{#1}\right\vert}

\newcommand{\prob}[1]{\Pr\paren{#1}}
\newcommand{\expect}[1]{\Exp\bracket{#1}}

\newcommand{\set}[1]{\ensuremath{\left\{ #1 \right\}}}
\newcommand{\poly}{\mbox{\rm poly}}
\newcommand{\polylog}[1]{\textnormal{polylog}\,{#1}\xspace}

\newcommand{\Alg}{\ensuremath{\mathsf{Alg}}\xspace}

\DeclareMathOperator*{\Exp}{\ensuremath{{\mathbb{E}}}}
\DeclareMathOperator*{\Prob}{\ensuremath{\textnormal{Pr}}}
\renewcommand{\Pr}{\Prob}

\newenvironment{tbox}{\begin{tcolorbox}[
		enlarge top by=5pt,
		enlarge bottom by=5pt,
		 breakable,
		 boxsep=2pt,
                  left=5pt,
                  right=7pt,
                  top=10pt,
                  arc=0pt,
                  boxrule=1pt,toprule=1pt,
                  colback=white
                  ]%%
	}
{\end{tcolorbox}}

\newcommand{\II}{\ensuremath{\mathbb{I}}}

\newcommand{\mireal}[1][]{
  \ifx\relax#1\relax%
    \II(\mione \,; \mitwo)%
  \else%
    \II(\mione \,; \mitwo\mid #1)%
  \fi
}

\newcommand{\block}[1]{\ensuremath{{B}^{#1}}}

\renewcommand{\leq}{\leqslant}

\renewcommand{\geq}{\geqslant}
\renewcommand{\ge}{\geq}

\newcommand{\myqed}[1]{\let\qed\relax \hspace*{\fill} #1 \ensuremath{\square}}

\NewEnviron{constraint}{
	\begin{tcolorbox}
		\begin{equation}
			\BODY
		\end{equation}
	\end{tcolorbox}
}

\newcommand{\gyes}{\ensuremath{G_{\text{\tiny YES}}}}
\newcommand{\gno}{\ensuremath{G_{\text{\tiny NO}}}}

\newcommand{\dyes}{\ensuremath{\mathcal{D}_{\text{\tiny YES}}}}
\newcommand{\dno}{\ensuremath{\mathcal{D}_{\text{\tiny NO}}}}

\newcommand{\dd}{\ensuremath{\mathcal{D}^{D} }}

\newcommand{\dcno}{\ensuremath{\mathcal{D}^C_{\text{\tiny NO}}}}
\newcommand{\dcyes}{\ensuremath{\mathcal{D}^C_{\text{\tiny YES}}}}

\newcommand{\gcno}{\ensuremath{G^C_{\text{\tiny NO}}}}
\newcommand{\gcyes}{\ensuremath{G^C_{\text{\tiny YES}}}}

\newcommand{\gobs}{\ensuremath{G_{\text{obs}}}}
\newcommand{\eobs}{\ensuremath{E_{\text{obs}}}}

\newcommand{\hh}{\ensuremath{ \mathcal{H} }}
\newcommand{\xx}{\ensuremath{ \mathcal{X} }}
\newcommand{\yy}{\ensuremath{ \mathcal{Y} }}
\newcommand{\mm}{\ensuremath{ \mathcal{M} }}

\newcommand{\pyes}{\ensuremath{P_{\text{\tiny{YES}}}}}
\newcommand{\pno}{\ensuremath{P_{\text{\tiny{NO}}}}}
\newcommand{\ppair}{\ensuremath{P_{\text{pair}}}}
\newcommand{\snl}{\ensuremath{\mathsf{SNL}}\xspace}
\newcommand{\snlobs}{\ensuremath{\mathsf{SNL}_{u,u'}^{\text{obs}} }\xspace}
\newcommand{\snldup}{\ensuremath{\mathsf{SNL}_{u,u'}^{\text{dup}} }\xspace}
\newcommand{\snlcount}{\ensuremath{\mathsf{snl}^{\text{obs}} }\xspace}
\newcommand{\snldupct}{\ensuremath{\mathsf{snl}^{\text{dup}} }\xspace}
\newcommand{\snlsize}{\ensuremath{\mathsf{snl}_{u,u'} }\xspace}

\newcommand{\vhigh}{\ensuremath{V_{h}}\xspace}
\newcommand{\vlow}{\ensuremath{V_{\ell}}\xspace}

%% file: abstract.tex
\begin{abstract}

We study the problem of estimating the size of the maximum matching in the sublinear-time setting. This problem has been extensively studied, with several known upper and lower bounds. A notable result by Behnezhad (FOCS 2021) established a $2$-approximation in $\Ot(n)$ time. 
%On the lower bound side, \cite{Behnezhad2023Sublinear} showed that $\Omega(n^{1.2 - o(1)})$ queries are necessary to achieve better than a $2/3$-approximation.

However, all known upper and lower bounds are in the adaptive query model, where each query can depend on previous answers. In contrast, non-adaptive query models—where the distribution over all queries must be fixed in advance—are widely studied in property testing, often revealing fundamental gaps between adaptive and non-adaptive complexities. This raises the natural question: is adaptivity also necessary for approximating the maximum matching size in sublinear time? This motivates the goal of achieving a constant or even a polylogarithmic approximation using $\Ot(n)$ non-adaptive adjacency list queries, similar to what was done by Behnezhad using adaptive queries.

We show that this is not possible by proving that any randomized non-adaptive algorithm achieving an $n^{1/3 - \gamma}$-approximation, for any constant $\gamma > 0$, with probability at least $2/3$, must make $\Omega(n^{1 + \eps})$ adjacency list queries, for some constant $\eps > 0$ depending on $\gamma$. This result highlights the necessity of adaptivity in achieving strong approximations. However, non-trivial upper bounds are still achievable: we present a simple randomized algorithm that achieves an $n^{1/2}$-approximation in $O(n \log^2 n)$ queries.

Moreover, our lower bound also extends to the newly defined variant of the non-adaptive model, where queries are issued according to a fixed query tree, introduced by Azarmehr, Behnezhad, Ghafari, and Sudan (FOCS 2025) in the context of Local Computation Algorithms.

%More generally, we establish a tradeoff between the number of queries and the approximation factor in the lower bound: any $n^{\delta}$-approximation requires $\Omega(n^{1+\eps})$ queries for constants $\eps, \delta > 0$ satisfying $4\eps + 3\delta < 1$. On the upper bound side, we present a generalized algorithm that achieves an $n^{\delta}$-approximation using $O(n^{2 - 2\delta} \log n)$ non-adaptive queries for any $0 < \delta < 1$.
%
%Together, our results provide the first polynomial upper and lower bounds for estimating maximum matching size using non-adaptive queries, and demonstrate that non-adaptivity imposes strong limitations on the achievable approximation guarantees in this setting.

\end{abstract}

%We study the problem of estimating the size of the maximum matching in the sublinear-time setting. This problem has been extensively studied, with several known upper and lower bounds. A notable result by Behnezhad (FOCS 2021) established a 2-approximation in O~(n) time. 
%
%However, all known upper and lower bounds are in the adaptive query model, where each query can depend on previous answers. In contrast, non-adaptive query models—where the distribution over all queries must be fixed in advance—are widely studied in property testing, often revealing fundamental gaps between adaptive and non-adaptive complexities. This raises the natural question: is adaptivity also necessary for approximating the maximum matching size in sublinear time? This motivates the goal of achieving a constant or even a polylogarithmic approximation using O~(n) non-adaptive adjacency list queries, similar to what was done by Behnezhad using adaptive queries.
%
%We show that this is not possible by proving that any randomized non-adaptive algorithm achieving an n^{1/3 - gamma}-approximation, for any constant gamma > 0, with probability at least 2/3, must make Omega(n^{1 + eps}) adjacency list queries, for some constant eps > 0 depending on gamma. This result highlights the necessity of adaptivity in achieving strong approximations. However, non-trivial upper bounds are still achievable: we present a simple randomized algorithm that achieves an n^{1/2}-approximation in O(n \log n) queries.

%% file: intro.tex
\section{Introduction}
Maximum matching is among the most fundamental and well-studied problems in combinatorial optimization and theoretical computer science.
A \textbf{matching} in a graph is a set of edges no two of which share an endpoint, and a \emph{maximum matching} is a matching of the largest size.
Edmonds, in his seminal paper \cite{edmonds1965paths}, gave the first polynomial-time algorithm for the maximum matching problem, a breakthrough that influenced the formal development of computational complexity theory and inspired later formalizations of complexity classes, such as P and NP. This runtime was later improved to $O(m \sqrt{n})$ \cite{hopcroftKarp,micali1980v}, where $n$ is the number of vertices and $m$ is the number of edges. Further improving this runtime is a major open problem. The recent breakthrough on max flow \cite{chen2022maximum,van2023deterministic} implies an $m^{1+o(1)}$ time deterministic algorithm for bipartite matching. 
While $O(m \sqrt{n})$ remains the best-known bound for exact solutions in general graphs, allowing approximations makes the problem significantly easier. For instance, a greedy maximal matching gives a $2$-approximation in linear time.
A $(1+\eps)$-approximation to the maximum matching can be achieved in $O(m/\eps)$ time \cite{hopcroftKarp,duan2014linear}.

However, in the era of massive graphs, the traditional goal of linear-time computation becomes inadequate.
This motivates the study of \emph{sublinear-time} algorithms: algorithms that estimate the solution without reading the entire input. In the case of maximum matching, a natural question that arises is whether there exist sublinear-time algorithms for approximating the maximum matching.

\subsection{Prior Work}
The sublinear-time setting assumes restricted access to the graph, typically via the \emph{adjacency list model} (our focus) or the \emph{adjacency matrix model}. In the former, algorithms may query the degree of a vertex or its $i^\text{th}$ neighbor (the algorithm receives NULL if the vertex has fewer than $i$ neighbors). In the latter, the algorithm specifies a pair of vertices and learns whether they are adjacent.

Unfortunately, there are lower bounds that answer the aforementioned question on sublinear-time approximate matching in the negative. 
In the adjacency matrix model or the adjacency list model, any algorithm that outputs a constant approximate matching needs $\Omega(n^2)$ queries \cite{Assadi2019Coloring,bhattacharya2023sublinear}.
While this lower bound rules out the possibility of efficiently outputting the edges of a matching, it does not preclude meaningful progress. Rather than retrieving the matching itself, we can shift our focus to estimating its size, which turns out to be a much more feasible goal. 

The first wave of results \cite{parnas2007approximating,Nguyen2008Constant,Yoshida2009improved,Onak2012Near} focused mainly on bounded-degree graphs and were not sublinear-time for maximum degree $\Delta = \Theta(n)$.
The second wave of results run in truly sublinear-time even on graphs with a large maximum degree, starting with the works \cite{chen2020Sublinear,Kapralov2020Space}.
\cite{Behnezhad2021Optimal} marked a milestone by establishing a $2$-approximation in $\Ot(n)$ time\footnote{Throughout the paper, we use $\Ot(f) := O(f \cdot \polylog(n))$ to hide polylogarithmic factors in $n$.}.
These were followed by a series of results by \cite{Behnezhad2023Beating,Behnezhad2023Sublinear,bhattacharya2023sublinear,Behnezhad2023Sublinear} culminating in a $(1, \eps n)$-approximation\footnote{An $(x, y)$-approximation achieves multiplicative error $x$ and additive error $y$.} in $o(n^2)$ time  \cite{Bhattacharya2023Dynamic}.
Note that all these results hold in the adjacency list query model, except for the last one, which holds in the adjacency matrix model but is believed to extend to the adjacency list setting as well \cite{Behnezhad2024Approximating}. 

These sublinear-time upper bounds are complemented by lower bounds on the query complexity which bound the running time of any algorithm that approximates the maximum matching size.
These begin with the \(\Omega(n)\) lower bound for constant-factor approximations \cite{parnas2007approximating}. Subsequent results established super-linear lower bounds for achieving better approximations, including \(\Omega(n^{1.2 - o(1)})\) for a \((1.5 - \eps)\)-approximation \cite{Behnezhad2023Sublinear} and \(\Omega(n^{2 - f(\eps)})\) for a \((1, \eps n)\)-approximation \cite{Behnezhad2024Approximating}.

While these results have significantly advanced our understanding of sublinear-time algorithms for maximum matching, they all rely on adaptive querying strategies, where queries are chosen based on answers to previous queries. 
In many cases, non-adaptive queries offer distinct advantages. Since all queries come from a predetermined distribution, they can be executed independently, making them well-suited for parallel computation frameworks.
This motivates the study of algorithms for maximum matching in a fully non-adaptive setting, where the distribution over queries must be fixed in advance.

Non-adaptive query models have been widely studied in property testing algorithms \cite{Goldreich2001Three,Ben2003Some,FISCHER2004107,Goldreich2011Algorithmic,Blais2023Testing}, often revealing fundamental gaps between adaptive and non-adaptive complexities \cite{raskhodnikova2006note,Gonen2007Benefits,Canonne2017Adaptivity,Goldreich2022Non}.
For instance, \cite{raskhodnikova2006note} shows that in the bounded degree model for graph property testing, adaptivity is essential.
Non-adaptive queries have been studied in the cut query model \cite{addanki2022non,mordoch2025cut} and more recently for the maximum matching problem in the context of Local Computation Algorithms \cite{azarmehr2025lower}. 
While non-adaptive queries have not been explicitly studied in the context of sublinear-time algorithms, we note that the sublinear-time algorithms for vertex coloring in \cite{Assadi2019Coloring,assadi2025simple} are fully non-adaptive.
This motivates understanding whether adaptivity is necessary for approximating the maximum matching size in sublinear-time.

However, for approximating maximum matching size, no prior work has explored the power or limitations of non-adaptive queries. 
Investigating this model could provide new insights into the role of adaptivity in sublinear-time graph algorithms and reveal whether adaptivity is essential for achieving strong approximations. 
Thus, similar to the adaptive case \cite{Behnezhad2021Optimal}, we would ideally like to obtain a constant-factor approximation around $2$ using $\Ot(n)$ non-adaptive queries. However, since this may be challenging, we instead consider a more relaxed question:  
\begin{quote}
	\emph{Can we get an $O(1)$ or even a $\polylog{n}$-approximation to the maximum matching size using $\Ot(n)$ non-adaptive adjacency list queries?}
\end{quote}

\subsection{Our Contributions}
We answer this question in the negative and show that adaptivity is not just helpful but essential for strong approximations, as formalized in the following theorem.
\begin{theorem}\label{thm:lb-large-apx}
	For any constant $\gamma>0$ there exists a constant $\eps>0$ such that any randomized algorithm that is given access to an arbitrary graph using non-adaptive adjacency list queries and can get an $n^{1/3-\gamma}$-approximation to the maximum matching size with probability at least $2/3$, needs to use $\Omega(n^{1+\eps})$ queries.
\end{theorem}
Note that our hard instance is a multi-graph; we allow multiple edges to simplify the analysis, though we believe similar bounds hold without them.

This theorem rules out any $\polylog{n}$-approximations in $\Ot(n)$ queries.
In fact, it says something much stronger.
It rules out even a polynomial approximation better than $n^{1/3}$ in $\Ot(n)$ queries.
We prove a more general theorem which we state below:
\begin{restatable}{theorem}{lbthm} \label{thm:lb}
	Any randomized algorithm that is given access to an arbitrary graph using non-adaptive adjacency list queries and can get an $n^{\delta}$-approximation to the maximum matching size with probability at least $2/3$, needs to use $\Omega(n^{1+\eps})$ queries for any $\eps,\delta>0$ satisfying $n^{2\eps+3\delta} \log^3 n =o(n)$.
	In other words, this implies a lower bound of $\widetilde{\Omega}(n^{1.5 - 1.5\delta})$ queries for achieving an $n^{\delta}$-approximation.
\end{restatable}
This theorem is more general and gives a tradeoff between the approximation factor and the number of queries required. Moreover, since our hard instance is a \emph{bipartite} graph, the lower bound also applies to algorithms restricted to bipartite inputs.
If we want to stick to a $\polylog{n}$-approximation, then \Cref{thm:lb} shows that we get a strong lower bound of $\widetilde{\Omega}(n^{1.5})$ queries.
In particular, we have the following corollary.
\begin{corollary}\label{corr:lb-large-queries}
	Any randomized algorithm that is given access to an arbitrary graph using non-adaptive adjacency list queries and can get better than a $\polylog{n}$-approximation to the maximum matching size with probability at least $2/3$, needs to use $\widetilde{\Omega}(n^{1.5})$ queries.
\end{corollary}

These results highlight the necessity of adaptivity in existing algorithms. 
In the $\Ot(n)$ query regime, there is a drastic difference: adaptive algorithms achieve a $2$-approximation, while non-adaptive ones cannot get better than an $n^{1/3}$ approximation.
But is anything achievable with non-adaptive queries? At first glance, the model may seem too restrictive to yield meaningful approximations. We show that this is not the case by presenting an algorithm that achieves an $n^{\delta}$-approximation to the maximum matching size using $O(n^{2 - 2\delta} \log n)$ non-adaptive adjacency list queries. 
In particular, when $\delta=1/2$, the algorithm gives a $\sqrt{n}$-approximation in $O(n \log n)$ queries.
This shows that non-adaptive queries, despite their limitations, can still achieve non-trivial approximations.

\begin{restatable}{theorem}{ubthm} \label{thm:ub}
	There is a randomized algorithm that outputs an $n^{\delta}$-approximation to the maximum matching size with probability $1-n^{-1}$ using $O(n^{2-2\delta} \log^2 n)$ non-adaptive queries in the adjacency list model and runs in the same asymptotic time, for any constant $0<\delta < 1$. In particular, when $\delta=1/2$, the algorithm gives a $\sqrt{n}$-approximation to the matching size in $O(n \log^2 n)$ queries.
\end{restatable}
%It is worth noting that this theorem is even stronger, as it provides not only an approximation to the matching size but also returns the edges of the approximate matching.
If we fix the number of queries to be $\Ot(n)$, then our algorithm achieves an $n^{1/2}$-approximation.
On the other hand, \Cref{thm:lb-large-apx} shows that any algorithm achieving an $n^{1/3-\delta}$-approximation requires more than $\Ot(n)$ queries. 
More generally, we show that any $n^{\delta}$-approximation requires $\widetilde{\Omega}(n^{1.5 - 1.5\delta})$ queries, and can be achieved with $\Ot(n^{2 - 2\delta})$ queries.
Together, these results provide polynomial upper and lower bounds in the $\Ot(n)$ query regime.

\paragraph{Extension to the Tree-Based Non-Adaptive Model.}
We now turn to a different non-adaptive query model recently introduced by \cite{azarmehr2025lower} in the context of Local Computation Algorithms for the maximum matching problem. 
Informally, the algorithm selects a root vertex $r$ and commits to a fixed tree-shaped query pattern rooted at $r$. It first queries some neighbors of the root, then the neighbors of those vertices, and so on, following the predetermined query tree. We refer to this as the \emph{Non-Adaptive Tree Probe Model}. Randomness may be used in choosing the root and the query structure, but all queries are fixed in advance, before any answers are revealed.

At first glance, this model may appear adaptive, since it allows querying neighbors of previously queried vertices. However, because the query structure is determined entirely in advance, it remains non-adaptive in an \emph{oblivious} sense.

While \cite{azarmehr2025lower} analyze this model in the context of Local Computation Algorithms for the maximum matching problem, our focus is on approximating the maximum matching \emph{size} using non-adaptive queries. Specifically, their problem is to decide whether a given edge belongs to a constant-factor approximate maximum matching constructed by the LCA, whereas we study the global task of estimating the maximum matching size to within an approximation factor.

This tree-based model appears more powerful than the standard non-adaptive adjacency list query model, as it enables exploration of multi-hop neighborhoods and can simulate random walks. Nevertheless, we show that this added flexibility does not yield a fundamental advantage for approximating the matching size. Specifically, we prove that any algorithm operating in this seemingly stronger non-adaptive model still requires $\Omega(n^{1+\eps})$ queries to achieve an $n^{\delta}$-approximation, matching the lower bound in \Cref{thm:lb}.
\begin{restatable}{theorem}{treelbthm} \label{thm:tree:lb}
	Any randomized algorithm that is given access to an arbitrary graph using adjacency list queries in the non-adaptive tree probe model and can get an $n^{\delta}$-approximation to the maximum matching size with probability at least $2/3$, needs to use $\Omega(n^{1+\eps})$ queries for any $\eps,\delta>0$ satisfying $n^{2\eps+3\delta} \log^3 n =o(n)$.
\end{restatable}

Finally, we observe a simple corollary. Although our results are stated for estimating the maximum matching size, they also extend to the minimum vertex cover size. This is because the sizes of a maximum matching and a minimum vertex cover are within a factor of $2$ in general graphs, and exactly equal in bipartite graphs. Since our lower bound instance is bipartite, the same hardness results carry over directly to minimum vertex cover estimation. The upper bounds remain valid: for super-constant approximation factors, the same bounds apply, and for constant-factor approximations, our algorithm already uses $O(n^2)$ queries.
\begin{corollary}
	Any randomized algorithm that is given access to an arbitrary graph using adjacency list queries in the non-adaptive query model or the non-adaptive tree probe model and can get an $n^{\delta}$-approximation to the minimum vertex cover size with probability at least $2/3$, needs to use $\Omega(n^{1+\eps})$ queries for any $\eps,\delta>0$ satisfying $n^{2\eps+3\delta} \log^3 n =o(n)$.
\end{corollary}

\begin{corollary}
	There is a randomized algorithm that with probability $1-1/n$ can output an $n^{\delta}$-approximation to the minimum vertex cover size using $O(n^{2-2\delta} \log n)$ non-adaptive queries in the adjacency list model for any constant $0<\delta < 1$. 
\end{corollary}

\input{tech-overview}

\subsection{Related Work}

\paragraph{Comparison with \texorpdfstring{ \cite{Behnezhad2023Sublinear}}{[BRR23]}.}

Behnezhad, Roghani, and Rubinstein (henceforth BRR23) prove an $\Omega(n^{1.2 - o(1)})$ lower bound for a $(2/3 + \eps)$-approximation to the maximum matching size using adaptive queries in the adjacency list model. Our setting is non-adaptive, and our lower bound exhibits a significantly larger gap between the matching sizes in the yes and no distributions.

Both works use coupling-based arguments, but have different structural properties of the observed graph. BRR23 shows that their observed graph forms a rooted forest, while we prove that ours forms a disjoint union of stars. The main technical part of their proof is a correlation decay, but we do not need that in our proof because we get a \emph{disjoint} union of stars.
However, given the large approximation parameter, we need a fundamentally different coupling argument.

Finally, while the models differ, the $n^{1.2}$ barrier in their work does not arise in our setting, and we achieve a better exponent: we show that even polylogarithmic approximations require $\Omega(n^{1.5 - o(1)})$ queries.

\paragraph{Other Related Work.}

In addition to sublinear-time algorithms, maximum matching has been extensively studied in several related computational models.
Maximum matching is one of the most central problems in the \emph{dynamic graph} literature \cite{Onak2010Maintaining,Manoj2013Fully,Baswana2015Fully,Bernstein2016Faster,BehnezhadK2022New,Behnezhad2023Dynamic,BhattacharyaKS2023Dynamic,Assadi2023Regularity,Bhattacharya2023Dynamic,Liu2024Approximate,Behnezhad2024Fully,assadi2025improved}.
Some of these works have shown that sublinear-time matching size estimators can be used to design dynamic algorithms.

The problem has also been heavily explored in the \emph{streaming model}, including insertion-only streams under both single-pass \cite{Feigenbaum05graph,Assadi2023Regularity} and multi-pass settings \cite{Fischer2022Deterministic,Konrad2024Unconditional}, as well as dynamic streams, again in both single-pass \cite{AssadiKLY16,ChitnisCEHMMV16,DarkK20,Assadi2022Asymptotically} and multi-pass settings \cite{Ahn2015Access,Assadi2024Simple,Assadi2024Settling}.
Finally, the maximum matching problem has also been studied in \emph{distributed models} \cite{fischer2020improved,Fischer2022Deterministic,Salwa2023Local,Ghaffari2023Faster,Ghaffari2024Parallel,mitrovic2025framework} and \emph{local computation algorithms} \cite{behnezhad2023local,azarmehr2025lower}.

%% file: tech-overview.tex
\subsection{Technical Overview}
In this subsection, we present a simplified construction and analysis to build intuition for a weaker version of our lower bound.

\paragraph{Graph Construction.}
Consider the construction illustrated in \Cref{fig:yes-no-fig}. Note that this is not the final distribution. We have a bipartite graph $G = (V = L \cup R, E)$, where the vertex set is partitioned into three subsets: $A$, $B$, and $D$, of sizes $2n^{1-\delta}$, $2n$, and $2n^{1-\delta}$, respectively. Each of these sets is evenly split between the left and right sides of the bipartition.
The vertices in $D$ are \textbf{dummy} nodes whose sole purpose is to increase the degrees of vertices in $A \cup B$, which we refer to as the \textbf{core}.

\begin{figure}[H]
	\centering
	\scalebox{0.7}{
	\begin{minipage}{0.48\textwidth}
		\centering
		\input{hardInstance_yes_D}
		\caption*{(a) Yes Instance }
	\end{minipage}
	\hfill
	\begin{minipage}{0.48\textwidth}
		\centering
		\input{hardInstance_no_D}
		\caption*{(b) No Instance }
	\end{minipage}
}
\caption{Hard Distribution}
	\label{fig:yes-no-fig}
\end{figure}

In the \emph{yes} case, there is a perfect matching of size $n$ between the $B$ vertices, and $n^{\delta}$ random matchings between the $A$ vertices. In the \emph{no} case, the $B$ vertices are divided into groups of size $n^{1-\delta}$, and each group has a random perfect matching to $A$.
In both cases, every vertex in the core has an edge to a vertex in $D$ on the opposite side with probability $1/2$. This completes the construction.
The adjacency lists of all vertices are randomly permuted. Due to the edges to dummy vertices, the degree of each core vertex is $\Omega(n^{1-\delta})$.

It is easy to see that in the \emph{yes} case, the matching has size at least $n$, while in the \emph{no} case, the set $A \cup D$ forms a vertex cover of the graph, implying that the maximum matching has size at most $4n^{1-\delta}$.
We now fix a deterministic algorithm $\Alg$ that makes $q = n^{1+\eps}$ non-adaptive queries and distinguishes between the \emph{yes} and \emph{no} cases. To hide the structure of the edges within the core, we apply a random permutation to the names of all vertices, so that $\Alg$ has no information about which vertex belongs to which set.
The queries made by $\Alg$ are distributed arbitrarily across the vertex set, but the random permutation ensures that, on average, each vertex is queried approximately $n^{\eps}$ times. 
%For simplicity, we assume $\Alg$ makes exactly $n^{\eps}$ queries on all vertices.
%Since we are aiming to establish a weaker bound, we treat both $\eps$ and $\delta$ as small constants throughout the analysis.

\paragraph{Disjoint union of stars.}
We define a vertex as \textbf{happy} if it observes an edge within the core. For a single query, this happens with probability at most $n^{\delta} / \Omega(n^{1-\delta})$.
Leveraging the randomness introduced by the permutation that shuffles vertex names, and noting that the algorithm $\Alg$ makes an average of $n^{\eps}$ queries, the overall probability that a vertex is happy becomes $O(n^{2\delta + \eps - 1})$.
Therefore, the expected number of happy vertices is $O(n^{2\delta + \eps})$. Now consider an edge: the probability that both its endpoints are happy is $O(n^{4\delta + 2\eps - 2})$.
Since there are just $2n$ edges in the core, the number of edges with both happy endpoints will be $0$ with a large probability.
As a result, the algorithm $\Alg$ observes only a disjoint union of stars in the core—since the petals of the stars are not happy and cannot observe neighbors in the core.

Note that if a star has degree greater than $1$, then its center must be an $A$ vertex, since $B$ vertices participate in at most one matching in the \emph{yes} case. However, if a star has degree exactly $1$, its center could be either an $A$ or a $B$ vertex.
We can tell $\Alg$ the labels of the star centers—whether they belong to $A$ or $B$—but what remains crucially hidden is the labels of the petals.
In the \emph{yes} case, all $A$ star centers are connected to $A$ petals, and all $B$ star centers to $B$ petals. In the \emph{no} case, the pattern is reversed: $A$ star centers have $B$ petals, and $B$ star centers have $A$ petals.
To show that $\Alg$ cannot distinguish between the \emph{yes} and \emph{no} cases, it suffices to prove that $\Alg$ succeeds with probability at most $1/2 + o(1)$.

\paragraph{Coupling Argument.}
We formalize this by constructing a coupling between the \emph{yes} and \emph{no} distributions. The coupling is a one-to-one mapping from the \emph{no} distribution to the \emph{yes} distribution such that the observed graph (the edges seen by $\Alg$) remains identical under this mapping.
To construct this mapping, we partition the support of both distributions into blocks, where all elements within a block correspond to the same observed graph. We then define the coupling to map between corresponding blocks in the two distributions.
Finally, we show that this mapping covers a $(1 - o(1))$ fraction of the support in both the \emph{yes} and \emph{no} distributions. Thus, for most observed graphs, the probability that they arise from the \emph{yes} or \emph{no} distribution is roughly $1/2$.
For the remaining $o(1)$ fraction, the algorithm $\Alg$ can succeed with probability at most $1$, leading to an overall success probability of at most $1/2 + o(1)$.
This coupling argument is the main technical ingredient of our proof.

%% file: hardInstance_yes_D.tex
\begin{tikzpicture}[every node/.style={font=\small}, node distance=0.3cm and 0.5cm, decoration={snake, amplitude=0.8pt, segment length=5pt}]
	
	% Box dimensions
	\def\aboxheight{2.0}
	\def\bboxheight{3.2}
	\def\dheight{1.0}
	\def\boxwidth{1.2}
	\def\xsep{3.2} % reduced from 4.5
	
	% Left side blocks
	\node[draw, rounded corners, minimum width=\boxwidth cm, minimum height=\aboxheight cm] (AL) at (0,0) {};
	\node[draw, rounded corners, minimum width=\boxwidth cm, minimum height=\bboxheight cm, below=of AL, yshift=-1cm] (BL) {};
	\node[draw, rounded corners, minimum width=\boxwidth cm, minimum height=\dheight cm, below=of BL, yshift=-0.7cm] (DL) {$D^L$};
	
	% Right side blocks
	\node[draw, rounded corners, minimum width=\boxwidth cm, minimum height=\aboxheight cm, right=\xsep cm of AL] (AR) {};
	\node[draw, rounded corners, minimum width=\boxwidth cm, minimum height=\bboxheight cm, below=of AR, yshift=-1cm] (BR) {};
	\node[draw, rounded corners, minimum width=\boxwidth cm, minimum height=\dheight cm, below=of BR, yshift=-0.7cm] (DR) {$D^R$};
	
	% === Core Box ===
	%	\node[draw, dashed, fit=(AL)(BL)(AR)(BR), inner sep=8pt, label={[yshift=10pt]above:\textbf{Core}}] (CORE) {};
	
	\draw[dashed]
	([xshift=-0.4cm, yshift=-0.3cm]BL.south west) rectangle
	([xshift=0.4cm, yshift=0.8cm]AR.north east);

	% Labels
	\node at (AL.north) [above=0.2cm] {$A^L$};
	\node at (BL.north) [above=0.2cm] {$B^L$};
	\node at (AR.north) [above=0.2cm] {$A^R$};
	\node at (BR.north) [above=0.2cm] {$B^R$};
%	\node at ($(AL)!0.5!(BL) + (-1.4,0)$) {Core (L)};
%	\node at ($(AR)!0.5!(BR) + (1.4,0)$) {Core (R)};
	
	% Nodes in A^L and A^R
	\foreach \i in {1,...,4} {
		\node[circle, draw, inner sep=1pt, fill=white] (al\i) at ($(AL.north)+(0,-0.4*\i)$) {};
		\node[circle, draw, inner sep=1pt, fill=white] (ar\i) at ($(AR.north)+(0,-0.4*\i)$) {};
	}
	
	% Nodes in B^L and B^R
	\foreach \i in {1,...,8} {
		\node[circle, draw, inner sep=1pt, fill=white] (bl\i) at ($(BL.north)+(0,-0.35*\i)$) {};
		\node[circle, draw, inner sep=1pt, fill=white] (br\i) at ($(BR.north)+(0,-0.35*\i)$) {};
	}
	
	% Identity matching between B^L and B^R using colorFive (blue)
	\foreach \i in {1,...,8} {
		\draw[colorTwo, thick] (bl\i) -- (br\i);
	}
	
	% First random matching between A^L and A^R using colorTwo (sky blue)
	\draw[colorSeven, thick] (al1) -- (ar3);
	\draw[colorSeven, thick] (al2) -- (ar1);
	\draw[colorSeven, thick] (al3) -- (ar4);
	\draw[colorSeven, thick] (al4) -- (ar2);
	
	% Second random matching between A^L and A^R using colorThree (bluish green)
	\draw[colorThree] (al1) to[bend left=20] (ar2);
	\draw[colorThree] (al2) to[bend left=20] (ar3);
	\draw[colorThree ] (al3) to[bend left=20] (ar1);
	\draw[colorThree] (al4) to[bend left=20] (ar4);

	% === Squiggly connections from D blocks to Core ===
	
	% D^L to right side of Core
	\draw[decorate, thick, colorOne] (DL.east) -- ++(3.8, 1.2);
	
	% D^R to left side of Core
	\draw[decorate, thick, colorOne] (DR.west) -- ++(-3.8, 1.2);

	% Legend
%	\node[colorFive] at ($(br4)!0.5!(bl4) + (0,-0.6)$) {\small{Identity matching}};
%	\node[colorTwo] at ($(al2)!0.5!(ar2) + (0,1)$) {\small{Matching 1}};
%	\node[colorThree] at ($(al3)!0.5!(ar3) + (0,1.4)$) {\small{Matching 2}};
	
\end{tikzpicture}

%% file: hardInstance_no_D.tex
\begin{tikzpicture}[every node/.style={font=\small}, node distance=0.3cm and 0.5cm, decoration={snake, amplitude=0.8pt, segment length=5pt}]
	
	% Box dimensions
	\def\aboxheight{2.0}
	\def\bboxheight{3.2}
	\def\dheight{1.0}	
	\def\boxwidth{1.2}
	\def\xsep{3.2} % horizontal gap between L and R blocks
	
	% === Core Blocks ===
	% Left side
	\node[draw, rounded corners, minimum width=\boxwidth cm, minimum height=\aboxheight cm] (AL) at (0,0) {};
	\node[draw, rounded corners, minimum width=\boxwidth cm, minimum height=\bboxheight cm, below=of AL, yshift=-1cm] (BL) {};
	\node[draw, rounded corners, minimum width=\boxwidth cm, minimum height=\dheight cm, below=of BL, yshift=-0.7cm] (DL) {$D^L$};	
	
	% Right side
	\node[draw, rounded corners, minimum width=\boxwidth cm, minimum height=\aboxheight cm, right=\xsep cm of AL] (AR) {};
	\node[draw, rounded corners, minimum width=\boxwidth cm, minimum height=\bboxheight cm, below=of AR, yshift=-1cm] (BR) {};
	\node[draw, rounded corners, minimum width=\boxwidth cm, minimum height=\dheight cm, below=of BR, yshift=-0.7cm] (DR) {$D^R$};
	
	% === Core Box ===
	%	\node[draw, dashed, fit=(AL)(BL)(AR)(BR), inner sep=8pt, label={[yshift=10pt]above:\textbf{Core}}] (CORE) {};
	
	\draw[dashed]
	([xshift=-0.6cm, yshift=-0.3cm]BL.south west) rectangle
	([xshift=0.6cm, yshift=0.8cm]AR.north east);

	% === Labels for blocks ===
	\node at (AL.north) [above=0.2cm] {$A^L$};
	\node at (BL.north) [above=0.2cm] {$B^L$};
	\node at (AR.north) [above=0.2cm] {$A^R$};
	\node at (BR.north) [above=0.2cm] {$B^R$};
%	\node at ($(AL)!0.5!(BL) + (-1.4,0)$) {Core (L)};
%	\node at ($(AR)!0.5!(BR) + (1.4,0)$) {Core (R)};
	
	% === A^L and A^R Nodes ===
	\foreach \i in {1,...,4} {
		\node[circle, draw, inner sep=1pt, fill=white] (al\i) at ($(AL.north)+(0,-0.4*\i)$) {};
		\node[circle, draw, inner sep=1pt, fill=white] (ar\i) at ($(AR.north)+(0,-0.4*\i)$) {};
	}
	
	% === B^L and B^R Nodes ===
	\foreach \i in {1,...,8} {
		\node[circle, draw, inner sep=1pt, fill=white] (bl\i) at ($(BL.north)+(0,-0.35*\i)$) {};
		\node[circle, draw, inner sep=1pt, fill=white] (br\i) at ($(BR.north)+(0,-0.35*\i)$) {};
	}
	
	% === Group Boxes with Shifts for Easy Editing ===
	
	% B^L_1
	\begin{scope}[shift={(0,0)}]
		\node[
		draw, dashed, rounded corners,
		fit=(bl1)(bl2)(bl3)(bl4),
		inner sep=3pt,
		label={[xshift=-8pt, yshift=0pt]left:$B^L_1$}
		] {};
	\end{scope}
	
	% B^L_2
	\begin{scope}[shift={(0,0)}]
		\node[
		draw, dashed, rounded corners,
		fit=(bl5)(bl6)(bl7)(bl8),
		inner sep=3pt,
		label={[xshift=-8pt, yshift=0pt]left:$B^L_2$}
		] {};
	\end{scope}
	
	% B^R_1
	\begin{scope}[shift={(0,0)}]
		\node[
		draw, dashed,rounded corners,
		fit=(br1)(br2)(br3)(br4),
		inner sep=3pt,
		label={[xshift=8pt, yshift=0pt]right:$B^R_1$}
		] {};
	\end{scope}
	
	% B^R_2
	\begin{scope}[shift={(0,0)}]
		\node[
		draw, dashed,rounded corners,
		fit=(br5)(br6)(br7)(br8),
		inner sep=3pt,
		label={[xshift=8pt, yshift=0pt]right:$B^R_2$}
		] {};
	\end{scope}
	
	% === Edges for NO Instance ===
	
	% B^R_1 → A^L (identity, colorThree)
	\draw[colorTwo, thick] (br1) -- (al1);
	\draw[colorTwo, thick] (br2) -- (al2);
	\draw[colorTwo, thick] (br3) -- (al3);
	\draw[colorTwo, thick] (br4) -- (al4);
	
	% B^R_2 → A^L (identity, colorThree)
	\draw[colorTwo] (br5) -- (al1);
	\draw[colorTwo] (br6) -- (al2);
	\draw[colorTwo] (br7) -- (al3);
	\draw[colorTwo] (br8) -- (al4);

	% B^L_1 → A^R (colorTwo)
	% A^R → B^L_1 (same structure as previous A^L–A^R matching, now reversed)
	\draw[colorSeven, thick] (ar3) -- (bl1); % a^R_3 → b^L_{1,1}
	\draw[colorSeven, thick] (ar1) -- (bl2); % a^R_1 → b^L_{1,2}
	\draw[colorSeven, thick] (ar4) -- (bl3); % a^R_4 → b^L_{1,3}
	\draw[colorSeven, thick] (ar2) -- (bl4); % a^R_2 → b^L_{1,4}

	% B^L_2 → A^R (colorTwo)
	% A^R → B^L_2 (same as second A^L–A^R matching in yes instance)
	\draw[colorThree] (ar2) -- (bl5); % a^R_2 → b^L_{2,1}
	\draw[colorThree] (ar3) -- (bl6); % a^R_3 → b^L_{2,2}
	\draw[colorThree] (ar1) -- (bl7); % a^R_1 → b^L_{2,3}
	\draw[colorThree] (ar4) -- (bl8); % a^R_4 → b^L_{2,4}

	% === Squiggly connections from D blocks to Core ===
	
	% D^L to right side of Core
	\draw[decorate, thick, colorOne] (DL.east) -- ++(3.8, 1.2);
	
	% D^R to left side of Core
	\draw[decorate, thick, colorOne] (DR.west) -- ++(-3.8, 1.2);
	
	% === Legend (Optional) ===
%	\node[colorTwo] at ($(bl3)!0.5!(ar3) + (0,-0.8)$) {\small{$B^L_i \rightarrow A^R$}};
%	\node[colorThree] at ($(br3)!0.5!(al3) + (0,-1.2)$) {\small{$B^R_i \rightarrow A^L$}};
	
\end{tikzpicture}

%% file: prelim.tex
% !TeX root = main.tex 
%!TEX root = main.tex

\section{Preliminaries}\label{sec:prelim}

\paragraph{Graph Notation and Terminology.}
For a graph $G=(V,E)$, we use $n$ to represent a constant factor approximation to the number of vertices. $m$ represents the 
number of edges and $\mu(G)$ to represent the size of any \emph{maximum matching} of $G$.
We use $\deg(v)$ and $N(v)$ for each vertex $v \in V$ to denote the degree 
and neighborhood of $v$, respectively. 
Any vertex $v$ with $\deg(v)=0$ is called an \textbf{isolated} vertex.
A \textbf{star} graph is one where there is a central vertex, called the \textbf{star center}, connected to all other vertices, called \textbf{petals}, with no edges between the petals. Note that there are no multi-edges between the petals.

For a subset $F$ of edges in $E$, we use $V(F)$ to denote 
the vertices incident on $F$; similarly, for a set $U$ of vertices, $E(U)$ denotes the edges 
incident on $U$. We further use $G[U]$ for any set $U$ of vertices to denote the induced subgraph 
of $G$ on $U$. 
%For sets of vertices $S$ and $T$, $S$-to-$T$ edges are edges with one endpoint in $S$ and the other in $T$. 
We use $\card{T}$ to denote the number of elements in $T$ when it is a set, and the size of its support when $T$ is a distribution.

A \textbf{perfect matching} between two sets of vertices $L$ and $R$ is a matching in which each edge connects a unique vertex in $L$ to a unique vertex in $R$, and every vertex in $L \cup R$ is matched. This requires that $\card{L} = \card{R}$. An \textbf{identity matching} between $L$ and $R$ assumes that both sets are ordered, and matches the $i^{\text{th}}$ vertex in $L$ to the $i^{\text{th}}$ vertex in $R$ for all $i \in [\card{L}]$.

Let $M^*$ denote a maximum matching in the graph.
We define an $n^{\delta}$-\textbf{approximate matching} as a matching of size at least a $1/n^{\delta}$ fraction of the size of $M^*$.
We define an $n^{\delta}$-\textbf{approximation} to the \emph{matching size} as an estimate $s$ such that
\[
{n^{-\delta}} \cdot \card{M^*} \leq s \leq \card{M^*}.
\]

\paragraph{Probability Bounds.}
Throughout, we will use the term ``\textbf{with high probability}" to mean with probability at least $1-o(1)$.
In some places, we obtain a probability bound of $1 - 1/\poly(n)$. We do not refer to this as occurring ``with high probability''; instead, we state the bound explicitly.
We use the following standard form of Chernoff bounds:
\begin{proposition}[Chernoff bound; c.f.~\cite{dubhashi1996balls,dubhashi2009concentration}]\label{prop:chernoff}
	Suppose $X_1,\ldots,X_m$ are $m$ independent or negatively associated random variables with range $[0,1]$ each. Let $X := \sum_{i=1}^m X_i$ and $\mu_L \leq \expect{X} \leq \mu_H$. Then, for any $\eps > 0$, 
	\[
	\Pr\paren{X >  (1+\eps) \cdot \mu_H} \leq \exp\paren{-\frac{\eps^2 \cdot \mu_H}{3+\eps}} \quad \textnormal{and} \quad \Pr\paren{X <  (1-\eps) \cdot \mu_L} \leq \exp\paren{-\frac{\eps^2 \cdot \mu_L}{2+\eps}}.
	\]
\end{proposition}

%% file: log-apx-LB-AdjList.tex
\section{The Lower Bound}
\label{sec:const-apx-lb}

\begin{table}[t]
	\scriptsize
	\centering
	\renewcommand{\arraystretch}{1.2}
	
	\begin{minipage}[t]{0.44\textwidth}
		\centering
		\resizebox{\linewidth}{!}{
		\begin{tabular}{|c|>{\raggedright\arraybackslash}p{4.9cm}|}
			\hline
			\textbf{Symbol} & \textbf{\makebox[\linewidth][c]{Description}} \\
			\hline
			\multicolumn{2}{|c|}{\textbf{Graph Components}} \\
			\hline
			$G = (V, E)$ & Bipartite graph with $V := L \cup R$ \\
			$A, B, D$ & Sets of vertices $A^L \cup A^R$, $B^L \cup B^R$, $D^L \cup D^R$ of sizes $2n^{1-\delta}, 2n, 2n^{1-\delta}$ \\
			$C$ & $C := A \cup B$ is the core of $G$ \\
			$B_i^L, B_i^R$ & Groups within $B^L$ and $B^R$ ($i \in [n^{\delta}]$) \\
			$b_{i,j}^R$ & $j^{th}$ vertex in group $B^R_i$ ($j \in [n^{1-\delta}]$)\\			
			$a_i^R$ & $i^{th}$ vertex in $A^R$ ($i \in [n^{1-\delta}]$) \\			
			$b_i^R$ & $i^{th}$ vertex in $B^R$ ($i \in [n]$) \\						
			$\mathcal{U}, \mathcal{V}$ & Label sets for $L$ and $R$ \\			
			\hline
			\multicolumn{2}{|c|}{\textbf{Parameters}} \\
			\hline
			$n'$ & Size of $L$ and $R$, $n' = \Theta(n)$ \\
			$\delta$ & Parameter for approximation factor\\
			$\eps$ & Parameter for number of queries \\			
			$q:=n^{1+\eps}$ & Total number of non-adaptive queries \\
			$q_v$ & Number of queries for vertex $v$ \\
			$d^*$ & Degree of each vertex in $C$ \\
			\hline
		\end{tabular}
	}
		\caption{Notation for Graph Components and Parameters.}
	\end{minipage}
	\hfill
	\begin{minipage}[t]{0.52\textwidth}
		\centering
	\resizebox{\linewidth}{!}{
		\begin{tabular}{|c|>{\raggedright\arraybackslash}p{6.2cm}|}
			\hline
			\textbf{Symbol} & \textbf{\makebox[\linewidth][c]{Description}} \\
			\hline
			\multicolumn{2}{|c|}{\textbf{Distributions}} \\
			\hline
			$\dyes$, $\dno$ & ``Yes'' and ``No'' input distributions \\
			$\mathcal{D}$ & Mixture: $\frac{1}{2} \dyes + \frac{1}{2} \dno$ \\
			$\dd$ & Distribution of edges between Core and $D$ \\
			$\dcyes$, $\dcno$ & Distribution of edges in the Core in $\dyes$ and $\dno$ \\
			\hline
			\multicolumn{2}{|c|}{\textbf{Random Variables}} \\
			\hline
			$\pi$ & Random labeling permutation from $\Pi$ \\			
			$\ell(v)$ & Adjacency list of vertex $v$ \\
			$\psi(v)$ & Randomly chooses core indices in $\ell(v)$ \\
			$\sigma_C(v), \sigma_D(v)$ & Random permutation of core and dummy neighbors \\
			$\sigma(v)$ & Random permutation of all neighbors for $v \in D$ \\			
			\hline
			\multicolumn{2}{|c|}{\textbf{Observed Graph and Algorithm}} \\
			\hline
			$\Alg$ & Fixed deterministic algorithm making non-adaptive queries \\
			$\gobs = (V, \eobs)$ & Directed graph observed by $\Alg$ \\
			$\eobs^C$ & Observed edges in core ($C$) \\
			\hline
		\end{tabular}
	}
		\caption{Notation for the Distributions, Randomness, and Observed Graph.}
	\end{minipage}
\end{table}

In this section, we show that any algorithm that can get an $O(n^{\delta})$-approximation to the maximum matching 
size, needs to use $\Omega(n^{1+\eps})$ non-adaptive queries for any $\eps,\delta>0$ satisfying the equation $n^{2\eps+3\delta} \log^3 n =o(n)$.
\lbthm*

We will prove this theorem in this section. 
The goal is to establish a lower bound for randomized algorithms that achieve an $O(n^{\delta})$-approximation to the maximum matching size. 
To accomplish this, we construct a hard input distribution under which no deterministic algorithm, restricted to $n^{1+\eps}$ queries, can guarantee such an approximation. By Yao’s minimax principle, this immediately yields a worst-case lower bound for randomized algorithms. The construction is shown below.

\subsection{The Hard Instance Construction}

The hard instance is a bipartite graph $G=(V,E)$ (see \Cref{fig:yes-no-D-side-by-side}). 

\textbf{The vertex set:} The vertex set on the left $L$ is of size $n'=n+2n^{1-\delta}$ and is composed of three distinct 
subsets $A^L , B^L , D^L$ of size $n^{1-\delta}, n$ and $n^{1-\delta}$ respectively. 
Similarly, the vertex set on the right $R$ of size $n'$ is composed of $A^R , B^R , D^R$ of 
size $n^{1-\delta}, n$ and $n^{1-\delta}$ respectively. 
So the number of vertices in the graph is $2n' = \Theta(n)$.
We also divide $B^L$ (similar for $B^R$) into groups $B^L_i$ for $i\in [n^{\delta}]$ of size $n^{1-\delta}$.
Throughout, we may write $A, B, D$ to respectively refer to sets $A^L \cup A^R,B^L \cup B^R, D^L 
\cup D^R$.
We use $a^R_i$ to denote the $i^{th}$ vertex of $A^R$.
We do the same for the others.
Additionally, we use $b^R_{i,j}$ to denote the $j^{th}$ vertex in the group $B^R_i$.
Note that when considering vertices, capital letters denote sets of vertices and lowercase letters denote vertices.
Finally, we will refer to $C:=A \cup B$ as the \textbf{core} and call vertices in $D$ as \textbf{dummy} vertices (since they are added to just raise the degrees of all vertices in the core).

\textbf{The edge set:} For the edge set $E$, we give two distributions: in distribution $\dyes$ the 
maximum matching of $G$ is sufficiently large, and in distribution $\dno$ the maximum matching of 
$G$ is sufficiently small. 
The final input distribution $\mathcal{D} := \frac{1}{2} \dyes +\frac{1}{2} \dno$ draws its input from 
$\dyes$ with probability $1/2$ and from $\dno$ with probability $1/2$. 
The following edges are common in both distributions $\dyes$ and $\dno$ and come from distribution $\dd$:
\begin{itemize}
	\item Every vertex in $A^R$ has $n^{1-\delta}/2 - n^{\delta}$ edges to random vertices in $D^L$ sampled without replacement (same between $A^L$ and $D^R$).
	\item Every vertex in $B^R$ has $n^{1-\delta}/2 - 1$ edges to random vertices in $D^L$ sampled without replacement (same between $B^L$ and $D^R$).
\end{itemize}

The following edges are specific to $\dyes$ and $\dno$ respectively:
\begin{itemize}
	\item In $\dyes$, we additionally add an identity matching between $B^L$ and $B^R$ i.e. edges $(b^L_i,b^R_i)$ for all $i\in [n]$. 
	We also add $n^{\delta}$ random perfect matchings $M_1,M_2,\ldots ,M_{n^{\delta}}$ between $A^L$ and $A^R$. We call this entire distribution in the core unique to $\dyes$, $\dcyes$.
	\item In $\dno$, we add a random perfect matching between each group $B^L_i$ of $B^L$ and $A^R$ for $i \in [n^{\delta}]$.
	We also add an identity matching between each group $B^R_i$ of $B^R$ and $A^L$ for $i \in [n^{\delta}]$. This means that for every group $B^R_i$ we add edges $(b^R_{i,j} , a^L_j)$ for $j \in [n^{1-\delta}]$. We call this entire distribution in the core unique to $\dno$, $\dcno$.
\end{itemize}

\begin{figure}[t]
	\centering
	
	\begin{minipage}{0.45\textwidth}
		\centering
				\scalebox{0.9}{\input{hardInstance_yes_D_text}}
		\caption*{(a) Yes Instance ($\dyes$)}
	\end{minipage}
	\hfill
	\begin{minipage}{0.45\textwidth}
		\centering
				\scalebox{0.9}{\input{hardInstance_no_D_text}}
		\caption*{(b) No Instance ($\dno$)}
	\end{minipage}
	
	\caption{Illustration of a Yes and a No instance. \textcolor{colorTwo}{Blue} edges denote identity matchings: a matching of size $8$ in the Yes case, and two matchings of size $4$ each in the No case. \textcolor{colorSeven}{Pink} and \textcolor{colorThree}{Green} edges form matchings that are structurally identical across the instances but connect to different vertex sets on the left. \textcolor{colorOne}{Orange} squiggly lines represent edges from $D^L$ to $A^R \cup B^R$ (same for $D^R$), existing with probability roughly $1/2$.}
	\label{fig:yes-no-D-side-by-side}
\end{figure}

This concludes the construction of the yes and no distributions $\dyes$ and $\dno$. 
We note that this construction gives us a multi-graph.
The edges in the core will be called \textbf{core edges}. 
We will use the term $A$-to-$A$ edges to denote the edges between $A^L$ and $A^R$ (same for $B$-to-$B$ and $A$-to-$B$ edges).
For a vertex in the core, a neighbor in the core is called a \textbf{core neighbor} and its index in the adjacency list is called a \textbf{core index}.
We also emphasize that we pick a random \textbf{labeling permutation} $\pi$ over the vertices to give a random label to each vertex which will be seen by the algorithm in the adjacency list. The vertices have labels $\mathcal{U} := \set{u_1,u_2,\ldots, u_{n'}}$ for the $L$ side and $\mathcal{V} := \set{v_1,v_2,\ldots, v_{n'}}$ for the $R$ side. The set of all labeling permutations is $\Pi$.

We now describe the \textbf{adjacency lists} of these vertices. Note that the vertices will not have their original names in the structure but the labels given to them by the labeling permutation $\pi \in \Pi$.
We let $\ell(v)$ be the random variable denoting the adjacency list of vertex $v \in V$.
The adjacency list of every vertex should contain its neighbors in random order.
Thus, the randomness we need to define for the adjacency lists is for picking the locations of the core indices and random permutations over the elements that go in the core and the dummy indices.
%The non-core indices of the vertices in the core ($C$) have an edge to a uniformly random element of $D$ without replacement. 
%So adding a permutation over the non-core indices seems redundant, but we do so for convenience.

For $v \in C$, let $\psi(v)$ be the random variable which picks the core indices uniformly at random out of all the indices in the adjacency list of $v$. 
For $v \in C$, let $\sigma_C(v)$ be the random variable which picks a uniformly random permutation over the core neighbors of $v$.
For $v \in C$, let $\sigma_D(v)$ be the random variable which picks a uniformly random permutation over the dummy neighbors of $v$.
Note that for $v \in B$, $\sigma_C(v)$ is always the identity permutation because there is just one core index.
Therefore, for vertices $v \in C$, $\psi(v), \sigma_C(v)$ and $\sigma_D(v)$ together determine the positions and order of neighbors in $\ell(v)$.
For vertices $v\in D$, we just have a uniformly random permutation over all the neighbors in the adjacency list which is given by $\sigma(v)$.

Finally, we note that we give away the bipartition $L,R$ of the graph for free before the queries are 
made. 
An algorithm that makes $\Ot(n)$ queries can figure out the sets $D, C$ any vertex belongs to after the non-adaptive queries are made by looking at the approximate degree of every vertex. 
What is crucially hidden from the algorithm, however, is whether a vertex $v \in C$ belongs to 
$A$ or $B$.

We fix a deterministic algorithm $\Alg$ that makes a fixed number of non-adaptive adjacency list queries $q=n^{1+\eps}$ and differentiates between $\dyes$ and $\dno$.
For each vertex, the queries $\Alg$ makes for that vertex are fixed.
The number of queries $\Alg$ makes is $q_i^L$ for vertex $u_i$ and $q_j^R$ for vertex $v_j$.
If we consider some vertex $v \in V$ then we let $q_v$ be the number of queries $\Alg$ makes to the adjacency list of $v$.
When a query on a vertex say $u$ returns an edge, we think of the edge as directed going \emph{out of} $u$ for analysis purposes.
Any such edge returned in response to a query is referred to as an \textbf{observed} edge. 
We now formalize this notion by defining the \textbf{observed graph} $\gobs = (V, \eobs)$:
\begin{itemize}
	\item The graph is directed with vertex set $V$.
	\item For every vertex $v \in V$, all the edges observed by the queries on $v$ belong to the edge set $\eobs$ and are directed out of $v$.
	\item Each directed edge out of $v$ is labeled with its position in the adjacency list of $v$.
\end{itemize}
We also define $\eobs^C := \eobs(C)$ as the set of edges in the core and $\eobs^D := \eobs \setminus \eobs^C$ as the edges between the core and $D$.

\subsection{Properties of the Hard Instance}

In this subsection, we observe some simple properties of the hard instance.
We start with the following bounds on vertex degrees which follow immediately from the construction.
\begin{observation}
	For any graph $G$ drawn from $\dyes$ or $\dno$, the degree of any vertex in $C$ is $d^* := n^{1-\delta}/2$.
\end{observation}

Notice that the degrees of vertices in $A$ and $B$ are the same, and it is difficult for algorithm $\Alg$ to tell these two apart with a few number of queries. 
Vertices in $A$ have $n^{1-\delta}/2-n^{\delta}$ neighbors in $D$ and $n^{\delta}$ neighbors in the core.
Vertices in $B$ have $n^{1-\delta}/2-1$ neighbors in $D$ and $1$ neighbor in the core.
However, it is easy for $\Alg$ to identify $D$ vertices by just querying the degrees.
%%%%%%%%%%%%%%%%   Degrees of the vertices in D  %%%%%%%%%%%%%%%%
This is because the degrees of the vertices in $D$ are very concentrated around their expectation.
\begin{claim}\label{clm:deg-D-verts}
	With high probability, for all $v\in D$, $\expect{\deg(v)}=n/2 + n^{1-\delta}/2 - 2n^{\delta}$, $\deg(v) \in (1\pm 0.2) \cdot (n/2)$, and the number of edges $v$ has to $A$ is in $(1\pm 0.2) \cdot (n^{1-\delta}/2)$.
\end{claim}
\begin{proof}
	Let $v\in D$ be a vertex and let $X_i$ for $i\in[n^{1-\delta}]$ be $1$ if $a_i$ has an edge to $v$ and $0$ otherwise. 
	Let $Y_i$ for $i\in[n]$ be $1$ if $b_i$ has an edge to $v$ and $0$ otherwise.
	Note that if $v$ is on the left then it only has edges to vertices on the right and vice versa, so we drop the superscript of $L$ and $R$.
	
	We have $\expect{X_i} = \prob{X_i=1} = (d^*-n^{\delta})/n^{1-\delta} = 1/2 - n^{2\delta-1}$.
	Let $X = \sum_i X_i$ be the number of edges $v$ has to $A$.
	We have $\expect{X} = n^{1-\delta} \cdot (1/2- n^{2\delta-1}) = n^{1-\delta}/2 - n^{\delta}$.
	Similarly, we have $\expect{Y_i} = \prob{Y_i=1} = (d^*-1)/n^{1-\delta} = 1/2 - n^{\delta-1}$.
	Let $Y = \sum_i Y_i$ be the number of edges $v$ has to $B$.
	We have $\expect{Y} = n \cdot (1/2- n^{\delta-1}) = n/2  - n^{\delta}$.
	
	Then $Z=X+Y$ represents the degree of $v$. We have $\expect{Z} = \expect{X} +\expect{Y} = n/2 +n^{1-\delta}/2 -2 n^{\delta}$. 
	We now show concentration on $Z$ by showing concentration on $X$ and $Y$.
	$X$ and $Y$ are sums of independent random variables (since the choices across vertices are independent), so we can use the Chernoff bound (\Cref{prop:chernoff}) to prove concentration for them:	
	\begin{align*}
		\prob{\card{X - \expect{X}} > 0.1 \expect{X} } \leq 2 \exp(- (0.1)^2 \expect{X}/3) \ll 1/\poly(n),
	\end{align*}
	\begin{align*}
		\prob{\card{Y - \expect{Y}} > 0.1 \expect{Y} } \leq 2 \exp(- (0.1)^2 \expect{Y}/3) \ll 1/\poly(n).
	\end{align*}
	Thus, we get that $Z \in (1 \pm 0.1) \expect{Z}$ with probability $1-1/\poly(n)$.
	We union bound over all vertices and relax the bounds slightly to get the bounds in the claim.
\end{proof}

We now look at the graphs in $\dyes$ and $\dno$ and show that the matching sizes in both cases are very different.
\begin{lemma}\label{lem:matching-sizes}
	Let  $\gyes \sim \dyes$ and $\gno \sim \dno$. Then it holds with probability 1 that 
	$\card{\mu(\gyes)} \geq n$ and $\card{\mu(\gno)} \leq 4n^{1-\delta}$.
\end{lemma}
\begin{proof}
	All graphs in $\dyes$ have a matching of size $n$ between the $B$ vertices.
	The set $A \cup D$ forms a vertex cover for all graphs in $\dno$ implying that the matching size is at most $4n^{1-\delta}$ (the size of the maximum matching and minimum vertex cover are the same in bipartite graphs).
\end{proof}

Although the matching sizes in the yes and no instances differ significantly, we aim to show that since $\Alg$ makes a limited number of queries, it cannot distinguish between them with probability better than $\frac{1}{2} + o(1)$, which is only slightly better than random guessing. The central lemma we will establish is the following:
\begin{lemma}\label{lem:distinguish-yes-no}
The success probability of $\Alg$ in distinguishing between $\dno$ and $\dyes$ using $q = n^{1+\eps}$ queries is at most $1/2 + o(1)$. This holds under the constraint $n^{2\eps+ 3 \delta}\log^3 n = o(n)$.
\end{lemma}

With \Cref{lem:distinguish-yes-no} in hand, we can now immediately prove \Cref{thm:lb}.
\begin{proof}[Proof of \Cref{thm:lb}]
	Assume towards a contradiction that there is a randomized algorithm that uses $o(n^{1+\eps})$ non-adaptive adjacency list queries and gives better than a $n^{\delta}/4$ approximation to the maximum matching size with probability at least $2/3$.
	
	We know by Yao's minimax principle \cite{yao1977probabilistic} that distributional query complexity for any distribution is a lower bound on randomized query complexity.
	So consider the distribution $\mathcal{D} = \frac{1}{2} \dyes +\frac{1}{2} \dno$.
	The existence of such an algorithm implies that there is a deterministic algorithm $\Alg$ which when given an input from $\mathcal{D}$, gives better than a $n^{\delta}/4$ approximation to the maximum matching size with probability at least $2/3$.  
	
	We also know by \Cref{lem:matching-sizes} that every graph in $\dyes$ has a matching size of at least $n$ and the size of the maximum matching for every graph in $\dno$ is at most $4 n^{1-\delta}$.
	Thus, any algorithm that can approximate the maximum matching size with a factor better than $n^{\delta}/4$ with probability at least $2/3$ can differentiate between $\dno$ and $\dyes$ with probability at least $2/3$.
	But we know by \Cref{lem:distinguish-yes-no} that $\Alg$ makes $q$ queries and distinguishes between $\dno$ and $\dyes$ with success probability at most $1/2 +o(1)$.
	This gives us a contradiction because we got a deterministic algorithm with probability of success $2/3$.
	Thus, we get a lower bound of $\Omega(n^{1+\eps})$ queries for better than an $n^{\delta}/4$ approximation.
	We can re-scale to get the same lower bound for an $n^{\delta}$ approximation, proving the theorem.
	Finally, we have the constraint $n^{2\eps+ 3 \delta}\log^3 n = o(n)$.
\end{proof}

\paragraph{Near-optimality on our hard instance.}
Before diving into the complete proof we want to mention that the lower bounds we get are almost \textbf{optimal} for our hard instance.
For a $\polylog{n}$-approximation, $\widetilde{\Omega}(n^{1.5})$ bits is the best lower bound we can get. Similarly, for a lower bound of $\Omega(n^{1+\eps})$ for a constant $\eps>0$, the best approximation factor we can get is $n^{1/3}$. 

First consider the lower bound of $\widetilde{\Omega}(n^{1.5})$ bits.
We can distinguish between $\dno$ and $\dyes$ in $\Ot(n^{1.5})$ non-adaptive adjacency list queries with probability $1-1/\poly(n)$.
The algorithm just makes all $n$ queries on $O(\sqrt{n} \log n)$ randomly chosen vertices on the left and right. By the birthday paradox, with probability $1-1/\poly(n)$ there is an edge $e$ in the core such that both of its endpoints are sampled (there are $2n$ edges in the core). Since we make $n$ queries on each vertex, we know whether that vertex is in $A$ or $B$ implying that we know both endpoints of some edge $e$. Thus, we know whether we are in the yes or the no case with probability $1-1/\poly(n)$.

We now consider the $n^{1/3}$-approximation.
We can distinguish between $\dno$ and $\dyes$ in $\Ot(n)$ non-adaptive adjacency list queries with probability $1-1/\poly(n)$ when $\delta=1/3$.
The algorithm just makes $\Ot(n^{1/3})$ queries on $\Ot(n^{2/3})$ randomly chosen vertices on the left and right.
We will show that this is enough for the algorithm to detect an $A$-to-$A$ edge in the yes case implying that we can distinguish between the two cases.
There are $n^{2/3}$ vertices in $A$ so sampling $\Ot(n^{2/3})$ random vertices will sample $\Ot(n^{1/3})$ vertices in $A$. Each of these vertices has $n^{1/3}$ neighbors in $A$ and $O(n^{2/3})$ neighbors to $D$. Thus, sampling $\Ot(n^{1/3})$ random neighbors will pick at least $\log n$ neighbors in $A$. This also identifies that such a vertex is an $A$ vertex because we see more than one core edge (all vertices in $D$ are identified by degree queries). 

We fix the randomness for these $\Ot(n^{1/3})$ $A$ vertices on the left each of which have at least $\log n$ edges to $A$ vertices on the right. Let the set of these vertices on the right be $S$ of size $\Omega(n^{1/3})$.
We now look at the $\Ot(n^{1/3})$ $A$ vertices on the right. The probability that we sample no vertices from $S$ on the right is: 
\[
\left(1 - \Omega(n^{1/3})/n^{2/3} \right)^{\Ot(n^{1/3})} \leq \exp(- \polylog{n}) \ll 1/\poly(n).
\]
Therefore, we will sample at least one vertex from $S$ and observe an $A$-to-$A$ edge in the yes case.
Thus, we can distinguish between the two instances using $\Ot(n)$ non-adaptive adjacency list queries.
Finally, we note that we do not need access to random neighbor queries. We can simulate that with an $O(\log n)$ overhead (\Cref{sec:alg}).

\paragraph{Two-Round Limitation of the Construction.}
We also remark that our hard distribution admits a simple two-round distinguishing strategy.
In the first round, the algorithm queries the degrees of all vertices and samples
$O(\log n)$ random vertices on one side of the bipartition.
For each sampled vertex, it queries the entire adjacency list (i.e., makes $n$ adjacency list queries).
With high probability, at least one sampled vertex $u$ lies in $B$.
Since we reveal the full adjacency list of $u$ and know the degrees of all vertices,
we can identify that $u$ is a $B$ vertex and determine its unique non-dummy neighbor $v$.
Indeed, $v$ has degree $d^*$, whereas all dummy neighbors of $u$ have degree $\Theta(n)$.

In the second round, it remains only to determine whether $v$ belongs to $A$ or $B$,
which uniquely identifies whether the instance was drawn from $\dyes$ or $\dno$.
This can be done by querying the entire adjacency list of $v$.
If $v$ has $n^\delta$ non-dummy (core) neighbors, then $v \in A$,
which corresponds to the No case.
If instead $v$ has exactly one non-dummy neighbor, then $v \in B$,
which corresponds to the Yes case.
Thus, two adaptive rounds suffice to distinguish the two distributions with high probability.

%% file: hardInstance_yes_D_text.tex
\begin{tikzpicture}[every node/.style={font=\small}, node distance=0.3cm and 0.5cm, decoration={snake, amplitude=0.8pt, segment length=5pt}]
	
	% Box dimensions
	\def\aboxheight{2.0}
	\def\bboxheight{3.2}
	\def\dheight{1.0}
	\def\boxwidth{1.2}
	\def\xsep{3.2} % reduced from 4.5
	
	% Left side blocks
	\node[draw, rounded corners, minimum width=\boxwidth cm, minimum height=\aboxheight cm] (AL) at (0,0) {};
	\node[draw, rounded corners, minimum width=\boxwidth cm, minimum height=\bboxheight cm, below=of AL, yshift=-1cm] (BL) {};
	\node[draw, rounded corners, minimum width=\boxwidth cm, minimum height=\dheight cm, below=of BL, yshift=-0.7cm] (DL) {$D^L$};
	
	% Right side blocks
	\node[draw, rounded corners, minimum width=\boxwidth cm, minimum height=\aboxheight cm, right=\xsep cm of AL] (AR) {};
	\node[draw, rounded corners, minimum width=\boxwidth cm, minimum height=\bboxheight cm, below=of AR, yshift=-1cm] (BR) {};
	\node[draw, rounded corners, minimum width=\boxwidth cm, minimum height=\dheight cm, below=of BR, yshift=-0.7cm] (DR) {$D^R$};
	
	% === Core Box ===
	%	\node[draw, dashed, fit=(AL)(BL)(AR)(BR), inner sep=8pt, label={[yshift=10pt]above:\textbf{Core}}] (CORE) {};
	
	\draw[dashed]
	([xshift=-0.4cm, yshift=-0.3cm]BL.south west) rectangle
	([xshift=0.4cm, yshift=0.8cm]AR.north east);

	% Labels
	\node at (AL.north) [above=0.2cm] {$A^L$};
	\node at (BL.north) [above=0.2cm] {$B^L$};
	\node at (AR.north) [above=0.2cm] {$A^R$};
	\node at (BR.north) [above=0.2cm] {$B^R$};
%	\node at ($(AL)!0.5!(BL) + (-1.4,0)$) {Core (L)};
%	\node at ($(AR)!0.5!(BR) + (1.4,0)$) {Core (R)};
	
	% Nodes in A^L and A^R
	\foreach \i in {1,...,4} {
		\node[circle, draw, inner sep=1pt, fill=white] (al\i) at ($(AL.north)+(0,-0.4*\i)$) {};
		\node[circle, draw, inner sep=1pt, fill=white] (ar\i) at ($(AR.north)+(0,-0.4*\i)$) {};
	}
	
	% Nodes in B^L and B^R
	\foreach \i in {1,...,8} {
		\node[circle, draw, inner sep=1pt, fill=white] (bl\i) at ($(BL.north)+(0,-0.35*\i)$) {};
		\node[circle, draw, inner sep=1pt, fill=white] (br\i) at ($(BR.north)+(0,-0.35*\i)$) {};
	}
	
	% Identity matching between B^L and B^R using colorFive (blue)
	\foreach \i in {1,...,8} {
		\draw[colorTwo, thick] (bl\i) -- (br\i);
	}
	
	% First random matching between A^L and A^R using colorTwo (sky blue)
	\draw[colorSeven, thick] (al1) -- (ar3);
	\draw[colorSeven, thick] (al2) -- (ar1);
	\draw[colorSeven, thick] (al3) -- (ar4);
	\draw[colorSeven, thick] (al4) -- (ar2);
	
	% Second random matching between A^L and A^R using colorThree (bluish green)
	\draw[colorThree] (al1) to[bend left=20] (ar2);
	\draw[colorThree] (al2) to[bend left=20] (ar3);
	\draw[colorThree ] (al3) to[bend left=20] (ar1);
	\draw[colorThree] (al4) to[bend left=20] (ar4);

	% === Squiggly connections from D blocks to Core ===
	
	% D^L to right side of Core
	\draw[decorate, thick, colorOne] (DL.east) -- ++(3.8, 1.2);
	
	% D^R to left side of Core
	\draw[decorate, thick, colorOne] (DR.west) -- ++(-3.8, 1.2);

	% Legend
%	\node[colorFive] at ($(br4)!0.5!(bl4) + (0,-0.6)$) {\small{Identity matching}};
%	\node[colorTwo] at ($(al2)!0.5!(ar2) + (0,1)$) {\small{Matching 1}};
%	\node[colorThree] at ($(al3)!0.5!(ar3) + (0,1.4)$) {\small{Matching 2}};
	
	% Text Node
	\draw (5.8,0) node [anchor=north west][inner sep=0.75pt]   [align=left] {\large{$n^{1-\delta}$}};

	% Text Node
	\draw (5.8,-3.5) node [anchor=north west][inner sep=0.75pt]   [align=left] {\large{$n$}};
	
	% Text Node
	\draw (5.8,-6.7) node [anchor=north west][inner sep=0.75pt]   [align=left] {\large{$n^{1-\delta}$}};
	
\end{tikzpicture}

%% file: hardInstance_no_D_text.tex
\begin{tikzpicture}[every node/.style={font=\small}, node distance=0.3cm and 0.5cm, decoration={snake, amplitude=0.8pt, segment length=5pt}]
	
	% Box dimensions
	\def\aboxheight{2.0}
	\def\bboxheight{3.2}
	\def\dheight{1.0}	
	\def\boxwidth{1.2}
	\def\xsep{3.2} % horizontal gap between L and R blocks
	
	% === Core Blocks ===
	% Left side
	\node[draw, rounded corners, minimum width=\boxwidth cm, minimum height=\aboxheight cm] (AL) at (0,0) {};
	\node[draw, rounded corners, minimum width=\boxwidth cm, minimum height=\bboxheight cm, below=of AL, yshift=-1cm] (BL) {};
	\node[draw, rounded corners, minimum width=\boxwidth cm, minimum height=\dheight cm, below=of BL, yshift=-0.7cm] (DL) {$D^L$};	
	
	% Right side
	\node[draw, rounded corners, minimum width=\boxwidth cm, minimum height=\aboxheight cm, right=\xsep cm of AL] (AR) {};
	\node[draw, rounded corners, minimum width=\boxwidth cm, minimum height=\bboxheight cm, below=of AR, yshift=-1cm] (BR) {};
	\node[draw, rounded corners, minimum width=\boxwidth cm, minimum height=\dheight cm, below=of BR, yshift=-0.7cm] (DR) {$D^R$};
	
	% === Core Box ===
	%	\node[draw, dashed, fit=(AL)(BL)(AR)(BR), inner sep=8pt, label={[yshift=10pt]above:\textbf{Core}}] (CORE) {};
	
	\draw[dashed]
	([xshift=-0.6cm, yshift=-0.3cm]BL.south west) rectangle
	([xshift=0.6cm, yshift=0.8cm]AR.north east);

	% === Labels for blocks ===
	\node at (AL.north) [above=0.2cm] {$A^L$};
	\node at (BL.north) [above=0.2cm] {$B^L$};
	\node at (AR.north) [above=0.2cm] {$A^R$};
	\node at (BR.north) [above=0.2cm] {$B^R$};
%	\node at ($(AL)!0.5!(BL) + (-1.4,0)$) {Core (L)};
%	\node at ($(AR)!0.5!(BR) + (1.4,0)$) {Core (R)};
	
	% === A^L and A^R Nodes ===
	\foreach \i in {1,...,4} {
		\node[circle, draw, inner sep=1pt, fill=white] (al\i) at ($(AL.north)+(0,-0.4*\i)$) {};
		\node[circle, draw, inner sep=1pt, fill=white] (ar\i) at ($(AR.north)+(0,-0.4*\i)$) {};
	}
	
	% === B^L and B^R Nodes ===
	\foreach \i in {1,...,8} {
		\node[circle, draw, inner sep=1pt, fill=white] (bl\i) at ($(BL.north)+(0,-0.35*\i)$) {};
		\node[circle, draw, inner sep=1pt, fill=white] (br\i) at ($(BR.north)+(0,-0.35*\i)$) {};
	}
	
	% === Group Boxes with Shifts for Easy Editing ===
	
	% B^L_1
	\begin{scope}[shift={(0,0)}]
		\node[
		draw, dashed, rounded corners,
		fit=(bl1)(bl2)(bl3)(bl4),
		inner sep=3pt,
		label={[xshift=-8pt, yshift=0pt]left:$B^L_1$}
		] {};
	\end{scope}
	
	% B^L_2
	\begin{scope}[shift={(0,0)}]
		\node[
		draw, dashed, rounded corners,
		fit=(bl5)(bl6)(bl7)(bl8),
		inner sep=3pt,
		label={[xshift=-8pt, yshift=0pt]left:$B^L_2$}
		] {};
	\end{scope}
	
	% B^R_1
	\begin{scope}[shift={(0,0)}]
		\node[
		draw, dashed,rounded corners,
		fit=(br1)(br2)(br3)(br4),
		inner sep=3pt,
		label={[xshift=8pt, yshift=0pt]right:$B^R_1$}
		] {};
	\end{scope}
	
	% B^R_2
	\begin{scope}[shift={(0,0)}]
		\node[
		draw, dashed,rounded corners,
		fit=(br5)(br6)(br7)(br8),
		inner sep=3pt,
		label={[xshift=8pt, yshift=0pt]right:$B^R_2$}
		] {};
	\end{scope}
	
	% === Edges for NO Instance ===
	
	% B^R_1 → A^L (identity, colorThree)
	\draw[colorTwo, thick] (br1) -- (al1);
	\draw[colorTwo, thick] (br2) -- (al2);
	\draw[colorTwo, thick] (br3) -- (al3);
	\draw[colorTwo, thick] (br4) -- (al4);
	
	% B^R_2 → A^L (identity, colorThree)
	\draw[colorTwo] (br5) -- (al1);
	\draw[colorTwo] (br6) -- (al2);
	\draw[colorTwo] (br7) -- (al3);
	\draw[colorTwo] (br8) -- (al4);

	% B^L_1 → A^R (colorTwo)
	% A^R → B^L_1 (same structure as previous A^L–A^R matching, now reversed)
	\draw[colorSeven, thick] (ar3) -- (bl1); % a^R_3 → b^L_{1,1}
	\draw[colorSeven, thick] (ar1) -- (bl2); % a^R_1 → b^L_{1,2}
	\draw[colorSeven, thick] (ar4) -- (bl3); % a^R_4 → b^L_{1,3}
	\draw[colorSeven, thick] (ar2) -- (bl4); % a^R_2 → b^L_{1,4}

	% B^L_2 → A^R (colorTwo)
	% A^R → B^L_2 (same as second A^L–A^R matching in yes instance)
	\draw[colorThree] (ar2) -- (bl5); % a^R_2 → b^L_{2,1}
	\draw[colorThree] (ar3) -- (bl6); % a^R_3 → b^L_{2,2}
	\draw[colorThree] (ar1) -- (bl7); % a^R_1 → b^L_{2,3}
	\draw[colorThree] (ar4) -- (bl8); % a^R_4 → b^L_{2,4}

	% === Squiggly connections from D blocks to Core ===
	
	% D^L to right side of Core
	\draw[decorate, thick, colorOne] (DL.east) -- ++(3.8, 1.2);
	
	% D^R to left side of Core
	\draw[decorate, thick, colorOne] (DR.west) -- ++(-3.8, 1.2);
	
	% === Legend (Optional) ===
%	\node[colorTwo] at ($(bl3)!0.5!(ar3) + (0,-0.8)$) {\small{$B^L_i \rightarrow A^R$}};
%	\node[colorThree] at ($(br3)!0.5!(al3) + (0,-1.2)$) {\small{$B^R_i \rightarrow A^L$}};
	
		% Text Node
	\draw (5.8,0) node [anchor=north west][inner sep=0.75pt]   [align=left] {\large{$n^{1-\delta}$}};
	
	% Text Node
	\draw (5.8,-3.5) node [anchor=north west][inner sep=0.75pt]   [align=left] {\large{$n$}};
	
	% Text Node
	\draw (5.8,-6.7) node [anchor=north west][inner sep=0.75pt]   [align=left] {\large{$n^{1-\delta}$}};
	
\end{tikzpicture}

%% file: structure.tex
\subsection{Limitations of the Algorithm}\label{subsec:Limitations-of-Alg}

\paragraph{High-Level Proof Plan.}
In this section, we show the limitations of the algorithm $\Alg$. Since $\Alg$ is restricted to making only $q = n^{1+\eps}$ queries, it can uncover only a small portion of the graph. Furthermore, because every vertex in the core has many edges to dummy vertices ($D$), the observed core subgraph—denoted $\eobs^C$—is extremely sparse. To be more precise, the main lemma of this section states that, with high probability, the subgraph observed by $\Alg$ in the core ($\eobs^C$) forms a disjoint union of stars (see \Cref{fig:stars-structure}). This structural property is key to showing that $\Alg$ cannot distinguish between the yes and no instances based on limited query access.
We state the lemma formally below:
\begin{lemma}\label{lem:disj-union-stars}
	Let $G \sim \mathcal{D}$ and consider the algorithm $\Alg$ that makes $n^{1+\eps}$ 
	non-adaptive adjacency list queries. 
	Then, with high probability, the edges $\Alg$ finds inside the core ($\eobs^C$) forms a disjoint union of stars.
\end{lemma}

To establish this, we will prove that, with high probability, any observed endpoint $v$ of an observed edge $u \rightarrow v$ neither observes any additional edges in the core (\Cref{clm:no-happy-edge}), nor has any other observed edges incident on it (\Cref{clm:no-happy-pairs}). Moreover, no vertex observes multi-edges (\Cref{clm:no-multi-edges}).
We will also establish some additional properties in this section that hold with high probability and will be used in the analysis of \Cref{sec:coupling}:
\begin{itemize}
	\item The number of observed edges in the core is small—at most $\Ot(n^{\eps+\delta})$ (\Cref{clm:no-edges-core}).
	\item Any observed endpoint $v$ of an observed edge $u \rightarrow v$ receives only a small number of queries, specifically $q_v \leq \Ot(n^{2\eps+\delta})$ (\Cref{clm:discovered-heavy-edges}).
	This lemma also shows that the number of petal vertices with many queries (below the previous threshold of $\Ot(n^{2\eps+\delta})$) is small, revealing a tradeoff between the number of queries and the number of such vertices.
	\item Every vertex $v \in C$ has small in-degree from dummy vertices $D$, bounded by $O(n^{\eps})$ (\Cref{clm:indeg-from D}).
\end{itemize}

To prove these claims, we rely on the randomness of the labeling permutation $\pi \in \Pi$ and the adjacency list $\ell \in \mathcal{L}$. 
\Cref{clm:no-multi-edges} is the only one where we also use the randomness in the distribution $\dyes$.
The algorithm $\Alg$ specifies queries for each vertex label such that their numbers sum to $q=n^{1+\eps}$. 
While the query distribution over labels in $\mathcal{U} \cup \mathcal{V}$ may be highly non-uniform, the random permutation $\pi$ ensures that, in expectation, any fixed underlying vertex (e.g., $a_1$) receives $n^{\eps}$ queries. For intuition, one can think of the analysis as if each vertex label receives $n^{\eps}$ queries. The bounds we derive will be the same even under the actual, potentially non-uniform query distribution.

The randomness used in the adjacency lists of core vertices comes from the random placement of core indices, determined by $\psi$ and the random permutations $\sigma_C,\sigma_D$. 
For vertices in $D$, we use the randomness of the adjacency list being randomly permuted, as specified by $\sigma$.
%Importantly, we do not rely on the randomness of $\dcno$ or $\dcyes$.

Finally, to prove the claims we need, we will set the following constraint:
\begin{constraint}\label{constraint:no-happy-pairs}
	n^{2\eps+ 3 \delta}\log^2 n =o(n).
\end{constraint}

\begin{table}[t]
	\footnotesize
	\centering
	\renewcommand{\arraystretch}{1.2}
	\begin{tabular}{|c|>{\raggedright\arraybackslash}p{12.2cm}|}
		\hline
		\textbf{Symbol / Term} & \textbf{\makebox[\linewidth][c]{Description}} \\
		\hline
		Happy vertex & A vertex that reveals at least one edge in the core through its queried indices. \\		
		$S$ & Set of star centers which are also the set of happy vertices. \\
		$\pyes$, $\pno$ & Set of petals (neighbors of star centers) in the yes and no cases, respectively. \\
		$\tau$-heavy vertex & A vertex with at least $\tau$ queries made on it. $\tau$ is taken as powers of 2: $\tau = 2^i$ for $i \in [\alpha, \beta]$. \\				
		$\alpha$, $\beta$ & $\alpha = \log(n^\eps \log n)$ and $\beta = \log(16n^{2\eps+\delta} \log^2 n)$. \\
		\hline
	\end{tabular}
	\caption{Notation for \Cref{subsec:Limitations-of-Alg}.}
\end{table}

We now turn to the formal claims and their proofs, as outlined in the high-level plan.
We begin by showing that, with high probability, only a small number of core edges are observed by $\Alg$.
\begin{claim}\label{clm:no-edges-core}
	With high probability, the number of core edges seen by the algorithm $\Alg$ is at most $4 n^{\eps+\delta} \log n$ in $\dyes$ and $\dno$.
\end{claim}
\begin{proof}
	Consider any vertex $v$ in $C$ for which we make $q_v$ queries.
	Since we pick a random labeling permutation $\pi$ over the vertices, the probability that this vertex $v$ is in $A$ is $2n^{1-\delta}/(2n+4n^{1-\delta}) \leq n^{-\delta}$ and the probability $v$ is in $B$ is at most $1$.
	
	Let the vertex $v$ be in $A$. We want to calculate the expected number of core edges that are queried for $v$.
	Let $X^v_i =1$ if the $i^{th}$ core edge of $v$ is queried for $i \in [n^{\delta}]$.
	Since the core indices are chosen randomly using $\psi(v)$, this just means that the $i^{th}$ core edge lands in the queried region.
	We have that $\expect{X^v_i} = \prob{X^v_i=1} = q_v/d^*$ since exactly $q_v$ out of $d^*$ indices are queried.
	$X^v = \sum_i X^v_i$ represents the number of core edges of $v$ that land in the queried region.
	We have that $\expect{X^v} = \sum_i \expect{X^v_i} = n^\delta \cdot q_v/d^*$.
	
	Let the vertex $v$ be in $B$. We want to calculate the expected number of core edges that are queried for $v$. Since a $B$ vertex has just one core edge this is the same as the probability that the core edge lands in the queried indices.
	Let $Y^v =1$ if the core edge of $v$ is queried for and $0$ otherwise.
	We have that $\expect{Y^v} = \prob{Y^v=1} = q_v/d^*$ since exactly $q_v$ out of $d^*$ indices are queried.
	
	We now come back to the case where $v$ is an $A$ vertex or a $B$ vertex with some probability.
	Let $Z^v$ be the random variable that denotes the number of core edges of $v$ that land in the queried region.
	We have: 
	\begin{align*}
		\expect{Z^v} &= \prob{v \in A} \cdot \expect{Z \mid v \in A} + \prob{v \in B} \cdot \expect{Z \mid v \in B} \\
		&\leq n^{-\delta} \cdot n^{\delta} \cdot q_v/d^* + 1 \cdot q_v/d^* \\
		&= 4 n^{\delta} q_v/n.
	\end{align*}
	
	Let $Z = \sum_{v} Z^v$ be the total number of observed edges in the core.
	The expected number of observed edges is $\expect{Z} = \sum_v \expect{Z^v} \leq \sum_v 4 n^{\delta} q_v/n = n^{1+\eps} \cdot 4 n^{\delta}/n = 4n^{\eps+\delta}$.
	The second last equality is true because $\sum_v q_v = q =n^{1+\eps}$.
	By Markov's inequality the probability that we exceed $\log n$ times the expectation is $1/\log n$, implying the bound in the claim.
\end{proof}

To prove \Cref{lem:disj-union-stars}, we define a vertex as \textbf{happy} if the queries on it reveal at least one edge in the core.
We first calculate the probability that a vertex is happy.
\begin{claim}\label{clm:happy-vertex}
	Let $a \in A$ be an arbitrary vertex and $v \in \mathcal{U} \cup \mathcal{V}$ be an arbitrary label.
	The probability that vertex $a$ with label $v$ is happy is at most $n^\delta \cdot q_v/d^*$.
\end{claim}
\begin{proof}
	Since $v\in A$, it has $n^\delta$ core edges.
	Let $\mathcal{E}_i$ be the event that the $i^{th}$ core edge lands in the queried portion for $i \in [n^\delta]$.
	We have $\prob{\mathcal{E}_i} = q_v/d^*$ and by a union bound:
	\[
	\prob{v \text{ is happy}} = \prob{\mathcal{E}_1 \cup \mathcal{E}_2 \cup \ldots \cup \mathcal{E}_{n^\delta}} \leq n^{\delta} \cdot q_v/d^*. \qedhere
	\]
\end{proof}

\begin{figure}[t]
	\centering
	
	\begin{subfigure}[t]{0.58\columnwidth}
		\centering
		\scalebox{0.7}{\input{star1}}
		\caption{Observed core structure: a disjoint union of stars.}
		\label{fig:stars-good}
	\end{subfigure}
	\hfill
	\begin{subfigure}[t]{0.38\columnwidth}
		\centering
		\scalebox{0.7}{\input{star2}}
		\caption{Forbidden configurations.}
		\label{fig:stars-bad}
	\end{subfigure}
	
	\caption{
		Structure of the observed core subgraph under non-adaptive queries.
		With high probability, the induced subgraph on the core forms a
		disjoint union of stars, where edges are directed away from the queried
		vertex that reveals the edge. Star centers correspond to happy vertices,
		and petals are not happy. The configurations on the right occur with
		probability $o(1)$ and are ruled out by \Cref{clm:no-happy-edge,clm:no-happy-pairs,clm:no-multi-edges}.
	}
	
	\label{fig:stars-structure}
\end{figure}
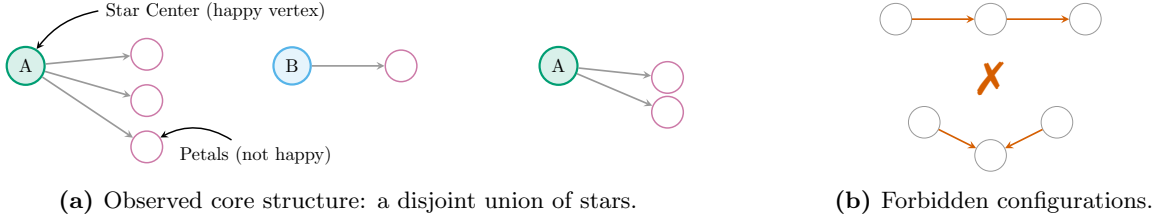

We now show that the incoming endpoint of every observed edge in the core is not happy.
\begin{claim}\label{clm:no-happy-edge}
	With high probability, the incoming endpoint of every observed edge in the core is not a happy vertex, in both $\dcyes$ and $\dcno$.
\end{claim}
\begin{proof}
	Consider the argument for $A$-to-$A$ edges. The cases of $A$-to-$B$ edges and $B$-to-$B$ edges can only have better bounds because the probability that a vertex in $B$ is happy is $1/d^*$ which is much smaller than $n^{\delta}/d^*$.
	Fix any $A$-to-$A$ edge $e=(a_i^L,a_j^R)$.
	Since $\pi$ is a random labeling permutation, $a_i^L$ and $a_j^R$ get different labels implying different queries to their adjacency lists depending on the label.
	
	Consider the case where $a_i^L$ gets label $u$ and $a_j^R$ gets label $v$.
	This happens with probability $1/{n'}^2$.
	We have $\prob{u \text{ observes } e} = q_u/d^*$ and $\prob{v \text{ is happy}} \leq n^{\delta} \cdot q_v/d^*$ (\Cref{clm:happy-vertex}).
	Also, note that these events are independent of each other because they only depend on the randomness of $\psi(u),\psi(v)$ in their adjacency lists.
	We now calculate the probability of event $\mathcal{E}$ that $a_i^L$ observes $e$ and $a_j^R$ is happy:
	\begin{align*}
		\prob{\mathcal{E}} &= \sum_{u \in \mathcal{U}} \; \sum_{v \in \mathcal{V}} (1/{n'}^2) \cdot \prob{u \text{ observes } e} \cdot \prob{\text{$v$ is happy}} \\
		&\leq \sum_{u \in \mathcal{U}} \; \sum_{v \in \mathcal{V}} (1/n^2) \cdot (q_u/d^*) \cdot (n^{\delta} q_v/d^*) \\
		&\leq 4n^{3\delta}/n^4 \sum_{u \in \mathcal{U}} q_u  \; \sum_{v \in \mathcal{V}}  q_v \\
		&\leq (4n^{3\delta}/n^4) \cdot n^{2+2\eps} \\
		&= 4 n^{2\eps+3\delta-2}.
	\end{align*}
	
	The same calculation also works for $A$-to-$B$ and $B$-to-$B$ edges with better bounds however we just use the worst-case bound obtained here.
	Since there are $2n$ edges in the core, the expected number of happy edges is at most $2n \cdot 4 n^{2\eps+3\delta-2} = 8 n^{2\eps+3\delta-1}=o(1)$ (by \Cref{constraint:no-happy-pairs}).
	Thus, by Markov's inequality, we conclude that no edge is happy with probability $1-o(1)$.
\end{proof}

We also show that no pair of vertices share an observed neighbor with high probability.
\begin{claim}\label{clm:no-happy-pairs}
	With high probability, no pair of vertices in the core observe an edge to the same neighbor for $\dcyes$ and $\dcno$.
\end{claim}
\begin{proof}
	A vertex in $A$ and a vertex in $B$ can never share a common neighbor in both cases by construction.
	In the yes case, only $A$-vertices can share neighbors, and in the no case, only $B$-vertices can share neighbors.
	Consider the argument for $A$ vertices in the yes case. The exact same argument works for $B$ vertices in the no case. 
	
	We fix any two vertices $a,a' \in A$ that share a neighbor $a''$.
	The number of possible choices for $a$ are $n^{1-\delta}$. After fixing $a$, $a'$ has at most $n^{2\delta}$ choices because $a$ has $n^\delta$ neighbors from which we pick $a''$ and $a''$ has $n^{\delta}$ different neighbors from which we choose $a'$.
	Thus, we fix a pair $a,a'$ that share neighbor $a''$, out of the $n^{1+\delta}$ choices and calculate the probability that both edges $e_1=(a,a'')$ and $e'_1=(a',a'')$ are observed.
	
	Consider the case where $a$ gets label $v$ and $a'$ gets label $v'$ ($v' \neq v$).
	This happens with probability $1/n' (n'-1)$.
	This is slightly different from the previous proof because both the vertices are on the same side.
	$\prob{(v,a'') \text{ is observed}} \leq q_v/d^*$ and $\prob{(v',a'') \text{ is observed}} = q_{v'}/d^*$.
	We now calculate the probability of the event $\mathcal{E}_1$ that both $(v,a'')$ and $(v',a'')$ are observed:
	\begin{align*}
		\prob{\mathcal{E}_1} &= \sum_{v \in \mathcal{V}} \; \sum_{v' \in \mathcal{V}-\set{v}} (1/n' (n'-1)) \cdot \prob{\text{$(v,a'')$ and $(v',a'')$ are observed}} \\
		&\leq \sum_{v \in \mathcal{V}} \; \sum_{v' \in \mathcal{V}} (1/n^2) \cdot (  q_v/d^*) \cdot (q_{v'}/d^*) \\
		&\leq 4n^{2\delta-4} \sum_{v \in \mathcal{V}} q_v \; \sum_{v' \in \mathcal{V}}  q_{v'} \\
		&\leq 4n^{2\delta-4} \cdot n^{2+2\eps} \\
		&= 4 n^{2\eps+2\delta-2}.
	\end{align*}
	
	Since there are $n^{1+\delta}$ choices for $a,a'$, the expected number of pairs of vertices in $A$ that share an observed neighbor is at most $n^{1+\delta} \cdot 4 n^{2\eps+2\delta-2} = 4 n^{2\eps+3\delta-1}=o(1)$ (by \Cref{constraint:no-happy-pairs}).
	Thus, by Markov's inequality, we conclude that no pair of happy vertices in $A$ share a neighbor with probability $1-o(1)$.
	
	We now consider $B$ vertices sharing an edge in the no case.
	We fix any two vertices $b,b' \in B$ that share a neighbor.
	The number of possible choices for $b$ are $n$. 
	After fixing $b$, $b'$ has $n^{\delta}$ choices because $b$ has $1$ neighbor $b''$ which has $n^{\delta}$ different neighbors from which we choose $b'$.
	Thus, we fix a pair $b,b'$ out of the $n^{1+\delta}$ choices calculate the probability that both edges $e_2=(b,b'')$ and $e'_2=(b',b'')$ are observed.
	
	Consider the case where $b$ gets label $v$ and $b'$ gets label $v'$ ($v' \neq v$).
	This happens with probability $1/n' (n'-1)$.
	$\prob{(v,b'') \text{ is observed}} \leq q_v/d^*$ and $\prob{(v',b'') \text{ is observed}} = q_{v'}/d^*$.
	We now calculate the probability of the event $\mathcal{E}_2$ that both $(v,b'')$ and $(v',b'')$ are observed:
	\begin{align*}
		\prob{\mathcal{E}_2} &= \sum_{v \in \mathcal{V}} \; \sum_{v' \in \mathcal{V}-\set{v}} (1/n' (n'-1)) \cdot \prob{\text{$v$ and $v'$ are happy}} \\
		&\leq \sum_{v \in \mathcal{V}} \; \sum_{v' \in \mathcal{V}} (1/n^2) \cdot (q_v/d^*) \cdot (q_{v'}/d^*) \\
		&\leq 4n^{2\delta-4} \sum_{v \in \mathcal{V}} q_v \; \sum_{v' \in \mathcal{V}}  q_{v'} \\
		&\leq 4n^{2\delta-4} \cdot n^{2+2\eps} \\
		&= 4 n^{2\eps+2\delta-2}.
	\end{align*}
	
	Since there are $n^{1+\delta}$ choices for $b,b'$, the expected number of pairs of happy vertices in $B$ that share a neighbor is at most $n^{1+\delta} \cdot 4 n^{2\eps+2\delta-2} = 4 n^{2\eps+3\delta-1}=o(1)$ (by \Cref{constraint:no-happy-pairs}).
	Thus, by Markov's inequality, we conclude that no pair of happy vertices in $B$ share a neighbor with probability $1-o(1)$.
	We can union bound over both cases and conclude the proof.
\end{proof}

We finally show that no vertex in the core observes a multi-edge with high probability.
\begin{claim}\label{clm:no-multi-edges}
	With high probability, no vertex in the core observes more than one edge to the same neighbor (multi-edges) for $\dcyes$.
\end{claim}
\begin{proof}
	Note that there can be many multi-edges in the graph between $A$ vertices. However, we will show that $\Alg$ cannot observe any multi-edge with high probability.
	
	Fix a vertex $a\in A^L$ and consider two distinct queries on $a$. 
	Because the adjacency list is uniformly permuted, each query chooses a uniformly random incident edge of $a$. 
	The probability that \emph{both} queries land on $A$ neighbors is at most $\left(n^{\delta}/d^*\right)^2 = 4n^{4\delta-2}$.

	Conditioned on both queries landing on $A$ vertices, their right endpoints in $A^R$ are independent and uniform over $\card{A^R}=n^{1-\delta}$, so the chance the endpoints coincide is $n^{\delta-1}$. 
	Hence, the probability that two fixed queries on $a$ are edges to the same vertex in $A^R$ is
	\[
	p_{\mathrm{same}} \leq  4n^{4\delta-2}\cdot n^{\delta-1} = 4n^{5\delta-3}.
	\]
	
	Fix a vertex label $v$ on the left. The probability $v$ lands in $A$ is at most $n^{-\delta}$. 
	Recall that $q_v$ is the number of queries made at $v$. Any pair of these queries could potentially observe a multi-edge.
	Thus, the probability that $v$ observes a multi-edge is:
	\[
	\prob{v \text{ observes a multi-edge}} \leq \prob{v \text{ is in } A} \cdot (q_v)^2 \cdot p_{\mathrm{same}} \leq n^{-\delta} \cdot (q_v)^2 \cdot 4n^{5\delta-3} = (q_v)^2 \cdot 4n^{4\delta-3}.
	\]
	
	Therefore, by a union bound, the probability we observe a multi-edge is:
	\[
	\prob{\text{Any multi-edge observed}} \leq  \sum_{v \in \mathcal{U}} (q_v)^2 \cdot 4n^{4\delta-3}
	\]
	
	To bound this probability expression, we need to bound the sum $\sum_{v \in \mathcal{U}} (q_v)^2$.
	Note that each $v$ has degree $d^*$, so $q_v\leq d^*$ (if $q_v$ is larger it sees NULL for the remaining queries) implies $q_v^2\leq d \cdot q_v$. 
	Thus, we have $\sum_{v \in \mathcal{U}} (q_v)^2 \leq d^* \cdot n^{1+\eps} = 0.5 n^{2+\eps-\delta}$.
	Substitute into the bound above:
	\begin{align*}
		\prob{\text{Any multi-edge observed}} &\leq  \sum_{v \in \mathcal{U}} (q_v)^2 \cdot 4n^{4\delta-3} \\
		&\leq 0.5 n^{2+\eps-\delta} \cdot 4n^{4\delta-3} \\
		&= 2n^{3\delta+\eps-1} =o(1).
	\end{align*}

	The last equality holds by \Cref{constraint:no-happy-pairs}.
	Thus, we conclude that no multi-edge is observed with probability $1-o(1)$.
	We can union bound over the cases for queries on each side and conclude the proof.
\end{proof}

Having established the necessary preliminaries, we now proceed to prove \Cref{lem:disj-union-stars}.
\begin{proof}[Proof of \Cref{lem:disj-union-stars}]
	We condition on the high probability events in \Cref{clm:no-happy-edge,clm:no-happy-pairs,clm:no-multi-edges}.
	Consider a happy vertex $v$ and its observed edge $e = v \rightarrow u$.
	We know from \Cref{clm:no-happy-edge} that the observed endpoint $u$ of an observed edge cannot be happy.
	Thus, there are no observed edges going out of $u$.
	Also, we know by \Cref{clm:no-happy-pairs} that no pair of vertices in the core observe an edge to the same vertex.
	Thus, there are no other observed edges incident on $u$.
	Finally, there are no multi-edges observed by \Cref{clm:no-multi-edges}.
	Therefore, the edges $\Alg$ finds inside the core ($\eobs^C$) form a disjoint union of stars with high probability.
\end{proof}

The answers to the queries made by $\Alg$ form a disjoint union of stars in the subgraph induced by the core. 
The center of each star can be either an $A$ or $B$ vertex.
We let $S$ denote the set of \textbf{star centers} and we call their observed neighbors \textbf{petals} denoted by $\pyes$ in the yes case and $\pno$ in the no case. 
Note that the center of a star is always a happy vertex i.e. the vertex that was queried, resulting in an edge within the core being observed.
If the center is $A$ then the star has degree at most $n^{\delta}$. 
Otherwise, if the center is $B$ then the star has degree $1$.
$\Alg$ can figure out the labels of the centers of the stars based on their degree.
The important part is that all the neighbor queries made on the petals of the stars are $D$ vertices, 
so we do not know for sure whether they are $A$ or $B$ vertices.
However, if there are many queries made on a petal vertex then $\Alg$ can figure out with some advantage whether the petal vertex comes from $A$ or $B$.
We now show that the number of queries on a petal vertex is actually small.

Let a vertex be $\tau$-\textbf{heavy} if the number of queries made on it are at least $\tau$.
Given that the total number of queries $\Alg$ makes is $q=n^{1+\eps}$, there are at most $n^{1+\eps}/\tau$ $\tau$-heavy vertices.
We now show that the number of queries made on petal vertices are less than $\tau$.
For this to be false an edge must be observed and its petal endpoint must be heavy.
We show the number of such edges is $0$ with high probability.
\begin{claim}\label{clm:discovered-tau-heavy}
	The number of edges observed with $\tau$-heavy petal vertices is at most $8n^{2\eps+\delta} \log^2 n/\tau$ with probability $1-2/\log^2 n$.
\end{claim}
\begin{proof}
	Consider the argument for $A$-to-$A$ edges. The proof is identical for other types of edges.
	Fix an $A$-to-$A$ edge $e=(a_i^L,a_j^R)$.
	Since $\pi$ is a random labeling permutation, $a_i^L$ and $a_j^R$ get random labels implying a different number of queries to their adjacency lists depending on the label.
	There are at most $n^{1+\eps}/\tau$ $\tau$-heavy vertex labels.
	Thus, the probability that $a_j^R$ gets a heavy label is $n^{1+\eps}/\tau \cdot n'$.
	
	Consider the case where $a_i^L$ gets label $u$.
	This happens with probability $1/n'$.
	$\prob{u \text{ observes } e} = q_u/d^*$.
	Note that the set of labels $\mathcal{U}$ and $\mathcal{V}$ are disjoint so the event that $u$ observes $e$ is independent of the event that $a_j^R$ picks a $\tau$-heavy label.
	We now calculate the probability that $e=(a_i^L,a_j^R)$ is observed-heavy:
	\begin{align*}
		\prob{e \text{ is observed-heavy}} &= \sum_{u \in \mathcal{U}} (1/n') \cdot \prob{u \text{ observes } e} \cdot \prob{\text{$v$ is $\tau$-heavy}} \\
		&\leq \sum_{u \in \mathcal{U}} (1/n) \cdot ( q_u/d^*) \cdot (n^{1+\eps}/\tau \cdot n) \\
		&=(2n^{\delta-2})\cdot (n^{\eps}/\tau) \cdot n^{1+\eps} \\
		&= 2n^{2\eps+\delta-1}/\tau.
	\end{align*}
	Since there are $2n$ edges in the core, the expected number of observed-heavy edges is at most $4n^{2\eps+\delta}/\tau$.
	Thus, by Markov's inequality, the probability that the number of observed-heavy edges is $\log^2 n$ times the expectation is $1/\log^2 n$.
	Note that the argument is the same for heavy-observed and observed-heavy edges, so we get at most $8n^{2\eps+\delta}\log^2 n/\tau$ such edges with probability $1-2/\log^2 n$ by a union bound.
\end{proof}

This claim implies that with high probability, there are no $\tau$-heavy petals for $\tau=16 n^{2\eps+\delta} \log^2 n$.
We have to be a bit careful with the other end. If we set $\tau=1$, we get that there are at most $8 n^{2\eps+\delta}$ $\tau$-heavy petals with high probability.
But we know by \Cref{clm:no-edges-core} that the number of observed edges which is same as the number of petals is at most $4n^{\eps+\delta}$.
The reason we get an extra factor of $n^{\eps}$ is that we count the number of $\tau$-heavy vertices as $n^{1+\eps}/\tau$. But when $\tau=1$ we have a tighter bound of $O(n)$ because there are only those many vertices in total.
Thus, at the lower end, we will set $\tau=n^\eps$, giving us at most $8n^{\eps+\delta} \log^2 n$ with high probability.
We now give the following lemma that captures these extreme cases along with others in the middle.
\begin{lemma}\label{clm:discovered-heavy-edges}
	With high probability, for all $i \in [\alpha,\beta]$, the number of edges observed with $2^i$-heavy petal vertices is at most $2^{-i} \cdot 8n^{2\eps+\delta} \log^2 n$ for $\alpha= \log (n^{\eps} \log n)$ and $\beta = \log (16 n^{2\eps+\delta} \log^2 n)$.
\end{lemma}
\begin{proof}
	We condition on the event in \Cref{clm:discovered-tau-heavy} for $\tau = 2^i$, for all $i \in [\alpha, \beta]$. 
	To account for all such values of $\tau$, we apply a union bound over the failure probabilities of all the cases. 
	Since the number of cases are at most $\beta \leq \log n$ (by \Cref{constraint:no-happy-pairs}), the total failure probability is at most $\log n \cdot (2/\log^2 n) = o(1)$.	
\end{proof}

Finally, we show that a vertex does not have many observed edges from $D$.
Any fixed vertex $v$ in the core has in-degree at most $4 n^{\eps}$ from $D$.
\begin{claim}\label{clm:indeg-from D}
	Every vertex $v \in A\cup B$ has in-degree at most $4n^{\eps}$ from $D$ in the observed graph with high probability.
\end{claim}
\begin{proof}
	Fix a vertex $v \in C$.
	We will calculate the number of observed edges from $D$ to $v$.
	$v$ has at most $d^*$ neighbors to uniformly random vertices in $D$.
	We let $X_i$ be a random variable that is $1$ if the $i^{th}$ edge of $v$ to $D$ is observed and $0$ otherwise for $i \in [d^*]$.
	We are interested in $X=\sum_{i \in [d^*]} X_i$ which is the total number of edges of $v$ to $D$ that are observed.
	
	We now focus on the random variable $X_i$.
	Let the $i^{th}$ edge of $v$ to $D$ fall on a vertex $d$.
	Since the labeling permutation is random, $d$ gets assigned a random label.
	Assume the label is $u$ which has $q_u$ queries.
	This happens with probability $1/n'$.
	The probability that the edge is observed when the label is $u$ is $q_u/\deg(u)$.
	Thus, the probability that $X_i=1$ is:
	\begin{align*}
		\prob{X_i=1} &= (1/n') \sum_{u \in \mathcal{U}} q_u/\deg(u) \\
		&\leq (4/n^2) \sum_{u \in \mathcal{U}} q_u \\
		&=4 n^{\eps-1}.
	\end{align*}
	This implies that $\expect{X} = \sum_{i\in [d^*]} \expect{X_i} = \sum_{i \in [d^*]} \prob{X_i=1} \leq d^* \cdot 4 n^{\eps-1}= 2 n^{\eps-\delta} \leq 2n^{\eps} =\mu_H$.
	Observe that the $X_i$'s are independent because each edge to $D$ is an independent sample with replacement.
	Also, each $X_i$ is in $[0,1]$ so we can apply the Chernoff bound (\Cref{prop:chernoff}) and get:
	\[
	\prob{X > 2 \mu_H} \leq \exp(-\mu_H/4) \ll 1/\poly(n).
	\] 
	Therefore, the in-degree of $v$ from $D$ is at most $4 n^{\eps}$ with probability $1-1/\poly(n)$.
	We can union bound this over vertices $v\in C$ and get that with high probability, all the vertices have in-degree at most $4n^{\eps}$ from $D$.
\end{proof}

For the rest of the proof we condition on the high probability events in \Cref{clm:deg-D-verts,clm:no-edges-core,lem:disj-union-stars,clm:discovered-heavy-edges,clm:indeg-from D}.

%% file: star1.tex
\begin{tikzpicture}[
	centerA/.style={circle, draw=colorThree, fill=colorThree!15, very thick, minimum size=7mm, font=\small},
	centerB/.style={circle, draw=colorTwo, fill=colorTwo!15, very thick, minimum size=7mm, font=\small},
	petal/.style={circle, draw=colorSeven, minimum size=6mm, thick},
	edge/.style={->, line width=0.9pt, draw=colorNine, >=stealth},
	annot/.style={->, thick, >=stealth, draw=colorEight},
	x=2.3cm, y=1.1cm
	]
	
	% --- Star 1 (A center, degree 3) ---
	\node[centerA] (A1) at (0,0) {A};
	\foreach \i in {1,...,3} {
		\node[petal] (N1\i) at (1,1-\i*0.8) {};
		\draw[edge] (A1) -- (N1\i);
	}
	
	% --- Star 2 (B center, degree 1) ---
	\node[centerB] (A2) at (2.2,0) {B};
	\node[petal] (N21) at (3.1,0) {};
	\draw[edge] (A2) -- (N21);
	
	% --- Star 3 (A center, degree 2) ---
	\node[centerA] (A3) at (4.4,0) {A};
	\foreach \i in {1,...,2} {
	\node[petal] (N3\i) at (5.3,0.4-\i*0.6) {};
		\draw[edge] (A3) -- (N3\i);
	}
	
%	% --- Global label ---
%	\node[draw=none, font=\small] at (3,-2) 
%	{Observed core graph forms a directed disjoint union of stars};
	
	% --- Annotations ---
%	\node[draw=none, font=\small] at (-1,0.6) {Star Center (happy vertex)};
%	\draw[annot, bend right=20] (-1,0.4) to (A1);
	
	\node[draw=none, font=\small] at (1.57,0.95) {Star Center (happy vertex)};
	\draw[annot, bend right=20] (0.6,0.95) to (A1);

	\node[draw=none, font=\small] at (1.9,-1.55) {Petals (not happy)};
	\draw[annot, bend right=30] (1.7,-1.28) to (N13);
	
\end{tikzpicture}

%% file: star2.tex
\begin{tikzpicture}[
	node/.style={circle, draw=colorNine, minimum size=6mm},
	edge/.style={->, line width=0.9pt, draw=colorSix, >=stealth},
	x=1.8cm, y=0.7cm
	]
	
	% =========================
	% TOP: your A2--N21 and N22--A2 snippet
	% =========================
	\node[node] (A2)  at (0, 1.6) {};
	\node[node] (N21) at (1, 1.6) {};
	\draw[edge] (A2) -- (N21);
	
	\node[node] (N22) at (-1, 1.6) {};
	\draw[edge] (N22) -- (A2);
	
	% =========================
	% MIDDLE: cross
	% =========================
%	\node[draw=none, font=\Huge, text=colorSix] at (0, 0) {$\times$};
	% If you prefer \ding{55} instead, use:
	 \node[draw=none, font=\Huge] at (0,0) {\textcolor{colorSix}{\ding{55}}};
	
	% =========================
	% BOTTOM: two vertices pointing to one below
	% =========================
	\node[node] (U1) at (-0.7, -1.2) {};
	\node[node] (U2) at (0.7,  -1.2) {};
	\node[node] (V1) at (0,    -2.1) {};
	
	\draw[edge] (U1) -- (V1);
	\draw[edge] (U2) -- (V1);
	
\end{tikzpicture}

%% file: coupling.tex
\section{The Coupling Argument}\label{sec:coupling}

\paragraph{High Level Idea.}
In this section, we will use coupling to show that after conditioning on some high probability events, the probability that an observed graph comes from $\dno$ is roughly $1/2$ and the probability it comes from $\dyes$ is roughly $1/2$.
Note that the support of $\dno$ and $\dyes$ are completely disjoint. However, the algorithm $\Alg$ only makes a few queries on this graph revealing a small portion of the graph. For most choices of the revealed part of the graph called the observed graph, it looks like it could have come from $\dno$ or $\dyes$ with roughly equal probability.
 
Ideally we want the coupling to be a perfect matching between elements of $\dno$ and $\dyes$ such that when we uniformly sample an element of $\mathcal{D} = \frac{1}{2} \dyes +\frac{1}{2} \dno$ we can do the sampling by picking a uniformly random edge of the matching between $\dno$ and $\dyes$ and then randomly choosing an endpoint (we will have $\card{\dno} = \card{\dyes}$).
The property the matching has is that for any edge in the matching the observed graph is exactly the same for both the endpoints of the edge. 
This means that any algorithm cannot tell apart the endpoints of an edge with probability better than $1/2$.
This sampling process is the same as sampling a uniformly random element from the no case and the yes case independently and then picking one of them with probability $1/2$.
The problem with this however is that the samples in the no and yes case could give different observed graphs.
We crucially need the observed graphs to be the same, so we can argue that the algorithm $\Alg$ does not know whether the instance that the observed graph came from is a yes instance or a no instance.
This is the reason we need a coupling
%The matching represents how the distributions are coupled together which means that if we uniformly sample an element from $\dno$ then its matched element is sampled from $\dyes$.

Unfortunately, we will not get a perfect matching between $\dno$ and $\dyes$ such that the observed graphs for the endpoints of each matching edge are identical.
However, if we ignore a $o(1)$ fraction of elements, then we indeed get a perfect matching between the two distributions (call this the successful matching).
We will add an arbitrary perfect matching between the ignored $o(1)$ fraction of elements and call this the failed matching.
When $\mathcal{D}$ samples an edge of the failed matching the algorithm $\Alg$ succeeds with probability $1$ in distinguishing between the two cases.
When $\mathcal{D}$ samples an edge of the successful matching the algorithm $\Alg$ succeeds with probability $1/2$ in distinguishing between the two cases.
Since the fraction of edges in the failed matching is $o(1)$ and in the successful matching is $(1-o(1))$, $\Alg$ is correct only with probability $1/2+o(1) < 2/3$.
Finally, we will need the following constraint in this section:
\begin{constraint}\label{constraint:matching-size}
	n^{2\eps+ 3 \delta}\log^3 n = o(n).
\end{constraint}

\begin{table}[t]
	\scriptsize
	\centering
	\renewcommand{\arraystretch}{1.2}
	\begin{tabular}{|c|>{\raggedright\arraybackslash}p{12.8cm}|}
		\hline
		\textbf{Symbol / Term} & \textbf{\makebox[\linewidth][c]{Description}} \\
		\hline
		$\hh = (\xx \cup \yy, \mm)$ & Bipartite graph used to define the coupling between $\dno$ and $\dyes$. \\
		$\xx, \yy$ & Sets of nodes in $\hh$, where each node is a tuple consisting of a graph instance from $\dno$ or $\dyes$, a labeling permutation $\pi \in \Pi$, and a set of adjacency lists $\ell \in \mathcal{L}$. \\
		$\mm$ & A matching between $\xx$ and $\yy$ of size $(1 - o(1)) \cdot \card{\xx}$, where each edge connects nodes whose observed graphs are identical. \\
		\hline
		{block} & A subset of nodes in $\xx$ or $\yy$ defined by fixing part of the underlying randomness (graph, permutation, or adjacency lists). \\
		{bad node} & A node that violates one or more of the high-probability events we condition on. \\
		{block-compliant} & A property of matching $\mm$ such that, when blocks are treated as supernodes, $\mm$ remains a valid matching on the resulting supernode graph. \\
		{confined matching} & The portion of matching $\mm$ restricted to operate within the current partition of blocks, preserving block-compliance. \\
		\hline
	\end{tabular}
	\caption{Notation for \Cref{sec:coupling}.}
\end{table}

We will fix elements in $\dno$ and then show which element in $\dyes$ they are coupled with.
When constructing the coupling we first couple the edges in the core which is a coupling between distributions $\dcno$ with $\dcyes$.
We then couple the labeling permutations which come from $\Pi$ which is the set of all labeling permutations.
Then we couple the edges to $D$ which come from the distribution $\dd$ which is the same in the no and the yes case.
We finally couple the adjacency lists between the two cases which come from $\mathcal{L}$.
This is just a rough idea of how the coupling works on a high level (see \Cref{fig:coupling-levels})—the final picture is more intricate.
To be very clear about the coupling, we define the following bipartite graph $\hh = (\xx \cup \yy, \mm)$.

\begin{figure}[t]
	\centering
	\scalebox{0.5}{\input{coupling_fig2}}
\caption{High-level illustration of the coupling. We progressively fix the randomness in the two distributions in parallel. We omit several steps and only give a flavor of the construction. In particular, we fix the core edges, the labeling permutation, and parts of the adjacency lists. At each stage the supports are restricted until both induce the same observed graph $\gobs$.}
	\label{fig:coupling-levels}
\end{figure}
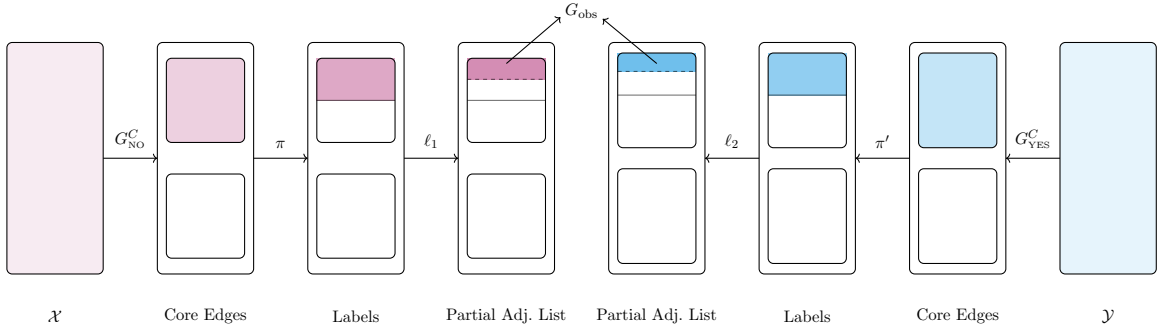

Each node of the graph in the no case (set $\xx$) corresponds to a graph $\gno$ from $\dno$, a labeling permutation $\pi$ from $\Pi$ and the adjacency lists of all the vertices denoted by $\ell \in \mathcal{L}$. 
The nodes of the graph in the yes case (set $\yy$) correspond to a graph $\gyes$ from $\dyes$, a permutation $\pi$ from $\Pi$ and permutations of the adjacency lists of all the vertices denoted by $\ell \in \mathcal{L}$.
Recall that the random variable $\ell(v)$ for adjacency lists of vertices $v \in C$ corresponds to picking the locations of the core indices uniformly at random using $\psi(v)$, randomly permuting the core neighbors using $\sigma_C(v)$ and the dummy neighbors using $\sigma_D(v)$.
The random variable $\ell(v)$ for adjacency lists of vertices $v \in D$ corresponds to $v$'s edges to the core and randomly permuting the entire adjacency list using $\sigma(v)$.
Note that we will use the term nodes for elements in $\xx$ and $\yy$, and reserve the term vertices for the graphs in $\dno$ and $\dyes$.

We will add edges to this graph and ensure that the edges of the graph $\mm$ forms a matching. We will also have the property that for any edge $(x,y) \in \mm$, the observed graph on $x$ and $y$ will be identical.
We first want to show that the sizes of $\xx$ and $\yy$ are the same.
We start by showing that the number of possible core graphs in the no case ($\dcno$) and the yes case ($\dcyes$) are the same.
\begin{claim}\label{clm:size-core}
$\card{\dcno} = \card{\dcyes}$ and $\card{\dno} = \card{\dyes}$.
\end{claim}
\begin{proof}
	In $\dcno$, the randomness comes from the matchings between $A^R$ and $B^L$. There is a random perfect matching between $A^R$ and each group of $B^L$. There are $n^{\delta}$ groups each of size $n/n^{\delta}$.
	In $\dcyes$, the randomness comes from the matchings between $A^L$ and $A^R$. There are $n^{\delta}$ random perfect matchings each of size $n/n^{\delta}$.
	
	Thus, both cases have $n^{\delta}$ random perfect matchings of size $n/n^{\delta}$ implying that the number of possible core graphs are the same in both cases ($\card{\dcno} = \card{\dcyes}$).
	This also implies that $\card{\dno} = \card{\dyes}$ because $\dno$ is a product distribution of $\dcno$ and $\dd$, and $\dyes$ is a product distribution of $\dcyes$ and $\dd$.
\end{proof}

We use this to prove that the support sizes of $\xx$ and $\yy$ are the same.
\begin{claim}\label{clm:size-XY}
	$\card{\xx} = \card{\yy}$.
\end{claim}
\begin{proof}
	\Cref{clm:size-core} proves that $\card{\dcno} = \card{\dcyes}$.
	The distribution $\dd$ is used in both $\dno$ and $\dyes$ and is independent of $\dcno$ and $\dcyes$, so we pick an arbitrary fixing of $\dd$.
	We now consider the set of labeling permutations $\Pi$ which is the same for both cases and independent of $\dno,\dyes$. Thus, we pick an arbitrary fixing of $\Pi$.
	 
	We finally look at the different possibilities for the adjacency lists in $\mathcal{L}$.
	For every vertex in the core $v\in C$, $\psi(v),\sigma_C(v)$ and $\sigma_D(v)$ have the same support size in the no and the yes case because we just have to randomly choose the core indices from the $d^*$ total indices and then randomly permute both parts.
	For every dummy vertex $v \in D$, the number of adjacency lists is the same in both cases since they are randomly permuted ($\deg(v)!$ possibilities in both cases).
	
	For each fixing of $\dd$ and $\Pi$ we have that $\card{\mathcal{L}}$ is the same in both cases and $\card{\dcno} = \card{\dcyes}$.
	Therefore, we can conclude that $\card{\xx} = \card{\yy}$.
\end{proof}

We now want to show that the matching $\mm$ is large.
We will prove that $\mm$ matches $(1-o(1))$ fraction of the vertices in $\xx$.
\begin{lemma}\label{lem:large-matching-coupling}
	$\card{\mm} \geq (1-o(1)) \card{\xx}$.
\end{lemma}

\begin{figure}[t]
	\centering
	\scalebox{0.85}{\input{distribution_coupling}}
	\caption{
		The perfect matching $\mm^*=\mm\cup\mm'$ pairs instances from $\xx$ and $\yy$.
		For edges in $\mm$, the endpoints look identical, so any algorithm succeeds with probability $1/2$.
		Edges in $\mm'$ form only a $o(1)$ fraction, and their endpoints may be distinguished with probability at most $1$.
		Hence the overall success probability is at most $1/2+o(1)$.
	}
	\label{fig:full-coupling-large-yes}
\end{figure}
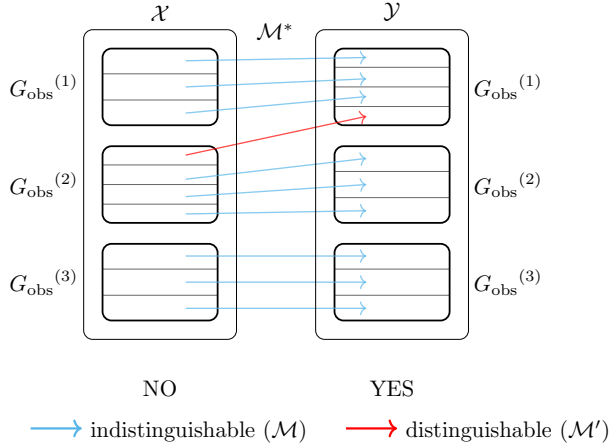

We will use this to prove \Cref{lem:distinguish-yes-no} (see \Cref{fig:full-coupling-large-yes}).
\begin{proof}[Proof of  \Cref{lem:distinguish-yes-no}]
	Consider algorithm $\Alg$ that makes $q=n^{1+\eps}$ queries and distinguishes between $\dno$ and $\dyes$.
	For any edge $e \in \mm$, the observed graph for both the endpoints is the same (property of the edges of $\mm$).
	So any algorithm cannot do better than guessing when asked to distinguish between the endpoints of $\mm$ and will be wrong with probability $1/2$.
	
	We will add an arbitrary perfect matching between the ignored $o(1)$ fraction of elements and call this the \textbf{failed} matching $\mm'$.
	For any edge $e \in \mm'$, $\Alg$ can be correct with probability at most $1$.
	 Also, note that $\mm^* = \mm \cup \mm'$ is a perfect matching between $\xx$ and $\yy$.
	 
	We now look at the distribution $\mathcal{D} = \frac{1}{2} \dyes +\frac{1}{2} \dno$ that samples a random edge of $\mm^*$ and then samples a random endpoint of that edge which is a random sample from $\mathcal{D}$.
	When $\mathcal{D}$ samples an edge of the successful matching $\mm$, the algorithm $\Alg$ succeeds with probability $1/2$ in distinguishing between the two cases.
	When $\mathcal{D}$ samples an edge of the failed matching $\mm'$, the algorithm $\Alg$ succeeds with probability at most $1$ in distinguishing between the two cases.
	Finally, the number of edges in the failed matching is $o(1) \cdot \card{\xx}$ and in the successful matching is $(1-o(1)) \cdot \card{\xx}$. 
	Thus, $\Alg$ is correct only with probability at most $(1-o(1)) \cdot (1/2) + o(1) \cdot 1 \leq 1/2+ o(1)$.
	
	We finally note that we have constraints \Cref{constraint:no-happy-pairs,constraint:matching-size} but  \Cref{constraint:no-happy-pairs} is weaker giving us the final constraint: $n^{2\eps+ 3 \delta}\log^3 n = o(n)$.
\end{proof}

We will now prove \Cref{lem:large-matching-coupling} in the rest of the section.
The idea is to divide $\xx$ and $\yy$ into blocks such that each block corresponds to a potentially different observed graph.
However, all the nodes in the block crucially have the same observed graph.
We will create these blocks by \textbf{fixing} some part of the randomness of the distributions.
Recall that $\xx$ and $\yy$ are the supports of distributions over tuples consisting of a graph, a labeling permutation, and a set of adjacency lists. By fixing parts of the randomness within these distributions, we can partition $\xx$ and $\yy$ into blocks, where each block corresponds to a particular fixing of that randomness.
For instance, if we fix the first element of the tuple to take on different graphs from $\dno$, then $\xx$ is partitioned into $\card{\dno}$ blocks, each corresponding to a particular graph from $\dno$.

We will have the property that if we think of the blocks as supernodes, $\mm$ still forms a matching on the supernode graph.
This introduces the notion of the matching being \textbf{block-compliant}.
The matching $\mm$ is block-compliant with respect to some fixing if for every edge $e \in \mm$ the following holds. 
Say edge $e$ is between two blocks $\xx_1$ and $\yy_1$ then any edge incident on a vertex in $\xx_1$ goes to $\yy_1$ and vice versa.

We will then show that if we match two blocks (on the supernode level) then their sizes are within a $(1 \pm o(1))$ factor of each other so $\mm$ matches all the nodes in one block and most nodes in the other.
We ignore a $o(1)$ fraction of nodes overall because there is a $o(1)$ fraction of nodes that does not satisfy the high probability events we condition on.
Also, when we match two blocks, their sizes could be a $(1 \pm o(1))$ factor apart.
Therefore, we get that matching $\mm$ has size $(1-o(1)) \cdot (1-o(1)) \cdot \card{X}=(1-o(1)) \cdot \card{X}$.

We call a node \textbf{bad} if it does not satisfy the high probability events we condition on and \textbf{good} otherwise. 
We will call a block is \textbf{bad} if all the nodes in it are bad.
We will show that every bad node will be part of a bad block, and we will ignore bad blocks.
These bad blocks we ignore will only be a $o(1)$ fraction of all the nodes because they are filled only with bad nodes and the bad nodes only form $o(1)$ fraction of all the nodes.
We will also say that a block is \textbf{good} if it is not bad.
Since all bad nodes belong to bad blocks, a good block will only contain good nodes.

We will not fix the randomness in one shot and get the blocks we want.
We will start with blocks of a large size and then divide them into finer and finer blocks.
Each time we fix some part of the randomness we will get sub-blocks of a block. We will find bad blocks at almost every level of this partitioning, and we will ignore such blocks.
We will consider arbitrary but fixed fixings on both sides that give good blocks and will keep $\mm$ \textbf{confined} within those blocks.
This means that the matching $\mm$ will be block-compliant at each level of the partitioning into finer blocks and will have edges incident on only good blocks (see \Cref{fig:dist}).

\subsection{Coupling the Core}\label{subsec:Coupling-the-Core}

%We sample $\gno \sim \dno$ and let $\pino$ be the permutation used on the vertices. We will now construct $\gyes \sim \dyes$ using coupling with $\gno$.
%Note that $\gno$ and $\pino$ are fixed along with the queries $q_i^L,q_j^R$ the algorithm makes.
%We condition on the event that the output is a disjoint union of stars inside the core.
%Let $S$ be the set of all star centers.

Here we describe the coupling for the core of the graph.
We consider an arbitrary core $\gcno$ in the no case $\dcno$. 
First consider the edges $(b^L_{i,j}, a_{j'}^R)$ and $(b^R_{i,j}, a_{j}^L)$ we have in $\gcno$ for all $i \in [n^{\delta}]$ and $j \in [n/n^{\delta}]$.
We will remove these edges and add edges $(b_{i,j}^L,b_{i,j}^R)$ and $(a_{j}^L,a_{j'}^R)$ which belong to $\gcyes$.
These four edges together form a \textbf{quad} in the coupling.

Using this we can also define the notion of \textbf{corresponding} edges.
For any edge $(u,v)$ in $\gcno$ directed out of $u$, its corresponding edge in $\gcyes$ is the only other edge in the coupling \textbf{quad} incident on $u$ and directed out of $u$.
We can similarly define this for $(u,v)$ in $\gcyes$ as well.
For instance, consider any edge $(b_{i,j}^L, a_{j'}^R)$ in the core in $\gcno$ and consider the direction $b_{i,j}^L \rightarrow a_{j'}^R$.
Its corresponding edge in the yes case will be the edge $b_{i,j}^L \rightarrow b_{i,j}^R$.
On the other hand, if we had considered the direction $a_{j'}^R \rightarrow b_{i,j}^L$ then the corresponding edge in the yes case would have been the edge $a_{j'}^R \rightarrow a_{j}^L$.

This procedure describes the coupling for all the edges within the core.
The instances shown in \Cref{fig:yes-no-D-side-by-side} illustrate examples of such coupled cores.
We now show that this coupling, denoted by $\mm^C$, always gives us an instance of $\dcyes$.

\begin{claim}\label{clm:core-coupling-valid}
	If we start with any $\gcno \in \dcno$ and apply the coupling procedure we get a core graph $\gcyes \in \dcyes$.
\end{claim}
\begin{proof}
This is because firstly we get edges $(b_i^L,b_i^R)$ for all $i \in [n]$.
Secondly, every group $B^L_i$ for $i \in [n^{\delta}]$ gives a uniformly random matching between $A^L$ and $A^R$.
This is because the matching between $B^L$ and $B^R$ is an identity matching in the yes case and the matching between $B^R$ and $A^L$ is an identity matching in the no case.
So if there is an edge between $b^L_{i,j}$ and $a^R_{j'}$ in the no case then there is an edge between $a^L_{j}$ and $a^R_{j'}$ in the yes case.
Therefore, the matching between $B^L_i$ and $A^R$ in the no case is the same as the matching between $A^L$ and $A^R$ in the yes case if we relabel $b^L_{i,j}$ with $a^L_j$.
Thus, we get $n^{\delta}$ uniformly random matchings between $A^L$ and $A^R$ in the yes case concluding the proof.
\end{proof}

\begin{table}[t]
	\footnotesize
	\centering
	\renewcommand{\arraystretch}{1.2}
	\begin{tabular}{|c|>{\raggedright\arraybackslash}p{12.8cm}|}
		\hline
		\textbf{Symbol / Term} & \textbf{\makebox[\linewidth][c]{Description}} \\
		\hline
		$\gcno, \gcyes$ & Core subgraphs of $\gno$ and $\gyes$, sampled from $\dcno$ and $\dcyes$, respectively. \\
		$\mm^C$ & Coupling (perfect matching) between $\dcno$ and $\dcyes$ that ensures the observed core subgraphs are identical for matched endpoints. \\
		{quad} & A group of four edges (two from $\gcno$, two from $\gcyes$) that define how individual edges are coupled. \\
		{corresponding edge} & For an edge $u \rightarrow v$ in $\gcno$, the coupled edge $u \rightarrow v'$ in $\gcyes$ belonging to the same quad (or vice versa). \\
		$\xx_1, \yy_1$ & Subsets of $\xx, \yy$ corresponding to a fixed core structure $\gcno$ and its coupled core structure $\gcyes$ \\
		\hline
	\end{tabular}
	\caption{Notation introduced in \Cref{subsec:Coupling-the-Core}.}
\end{table}

We now show that this coupling indeed gives a perfect matching.
\begin{claim}
	The coupling given by $\mm^C$ is a perfect matching between $\dcno$ and $\dcyes$. 
\end{claim}
\begin{proof}
	To show that this coupling is a perfect matching we need to show that the sizes of both sides are equal and that every node has degree exactly $1$ in this coupling.
	Note that \Cref{clm:size-core} showed that $\card{\dno^C} = \card{\dcyes}$. 
	Also, note that the degree of a node in $\dcno$ is exactly $1$ because the coupling gives exactly one instance of $\dcyes$ (\Cref{clm:core-coupling-valid}).
	All we need to show now is that every node in $\dcyes$ has degree $1$ in this coupling.
	
	Consider any $\gcyes \in \dcyes$ which is defined by matchings $M_1,M_2,\ldots, M_{n^{\delta}}$. 
	We will now invert the coupling to show exactly one pre-image.
	Consider any matching $M_i$ for $i \in [n^{\delta}]$ and groups $B^L_i, B^R_i$ in $B$. 
	We know that $B^R_i$ and $A^L$ have an identity matching between them in the no case and $B^L_i$ and $B^R_i$ have an identity matching between them in the yes case. 
		
	Consider an arbitrary edge between $(a^L_j , a^R_{j'})$ in the yes case.
	We look at the quad this edge is a part of in the coupling.
	This includes edges $(b^L_{i,j},b^R_{i,j})$ and $(a^L_{j},b^R_{i,j})$.
	To get the edge $(a^L_j , a^R_{j'})$ using the coupling, the last edge in the quad has to be $(b^L_{i,j} , a^R_{j'})$. If we had some other edge $(b^L_{i,j} , a^R_{j''})$ in the no case, then we would get the edge $(a^L_{j} , a^R_{j''})$ in the yes case.
	
	Thus, every node in the yes case has a unique pre-image in the coupling implying that the coupling is a perfect matching between $\dcno$ and $\dcyes$.
	The matching between $B^L_i$ and $A^R$ in the no case is essentially the same matching as the one between $A^L$ and $A^R$ in the yes case (up to changing the names of the vertices in a fixed order).
\end{proof}

Therefore, we can fix any arbitrary structure for the core $\gcno$ in the no case and its coupled structure $\gcyes$ in the yes case. 
This fixing creates blocks, and we need to ensure that $\mm$ is block-compliant.
We define the set of nodes $\xx_1 \subseteq \xx$ as $\xx$ restricted to $\gcno$ i.e. the nodes where the structure of the core is $\gcno$.
Similarly, define $\yy_1 \subseteq \yy$ as $\yy$ restricted to $\gcyes$.
Finally, we note that the sizes of these sets of nodes will be the same under any fixing because everything else (labeling permutation, edges to $D$, permutation of the adjacency lists) is independent of it.
%We say that the matching $\mm$ is \textbf{confined} to $(\xx',\yy')$ for any $\xx' \subseteq \xx$ and $\yy' \subseteq \yy$ if any edge of $\mm$ incident on a node in $\xx'$ will have its other endpoint incident on $\yy'$ and any edge of $\mm$ incident on a node in $\yy'$ will have its other endpoint incident on $\xx'$.
For any fixing of the core $(\xx_1,\yy_1)$, we will ensure that the matching $\mm$ is confined to it (see \Cref{fig:dist}).
This also makes sense because if we go to a block with a different structure then the observed graph could be different but for every edge in $\mm$ the endpoints should have the same observed graph.
For the rest of the proof, we will only look at the induced subgraph on $\xx_1$ and $\yy_1$.

\begin{center}
	\begin{tcolorbox}[colback=gray!5!white, colframe=black!50, boxrule=0.5pt, arc=4pt, width=0.95\linewidth, left=6pt, right=6pt, top=6pt, bottom=6pt]
		\textbf{Fixings So Far}
		\vspace{4pt}
		
		\begin{itemize}
			\item \Cref{subsec:Coupling-the-Core}: Fixed structure for the core: $\gcno$ in the no case and its coupled structure $\gcyes$ in the yes case.
		\end{itemize}
	\end{tcolorbox}
\end{center}

\subsection{Fixing the Permutation and the Star Centers}\label{subsec:coupling-perm}

We now fix an arbitrary labeling permutation $\pi \in \Pi$ in both $\xx_1$ and $\yy_1$.
We also fix the set $S$ of all the star centers (these are the set of vertices in the core that are happy).
Lastly, we also fix $\psi(S),\sigma_C(S)$ which fixes the locations of the core indices in the adjacency lists of vertices in $S$ and also the permutation of these core edges.
Doing this fixes the entire observed graph in the core $\eobs^C$ because the queries are deterministically fixed.
We also note that there can be a bad choice of the fixing which disagrees with the high probability event in \Cref{lem:disj-union-stars} ($\eobs^C$ is not a disjoint union of stars) in the no or the yes case. The fixing is also bad if it disagrees with the high probability event in \Cref{clm:no-edges-core} (the number of observed edges are bounded) or \Cref{clm:discovered-heavy-edges} (the number of observed-heavy edges are not bounded, which will happen if petals have many queries).
In these cases, the entire block on at least one side will be bad, and so we can ignore this fixing.

Consider a star center $v$ and its observed neighbor $u$ in the no case and its observed neighbor $u'$ in the yes case. 
$\Alg$ will see the edge $(v,u)$ in the no case and the edge $(v,u')$ in the yes case. Thus, the observed graphs will be different in both cases.
This implies that $\mm$ will not be confined to this particular fixing if we try to match the nodes in these blocks.
To solve this problem, we construct a new permutation $\pi'$ which is called the \textbf{dual permutation} of $\pi$.
A \textbf{dual} permutation $\pi'$ of a labeling permutation $\pi$ is one which starts with $\pi$ but swaps the label of every petal observed in the no case ($\pno$) with the label of the \textbf{corresponding petal} observed in the yes case ($\pyes$).
Consider a star center $v$ and its petal $u \in \pno$ in the no case. The corresponding edge for $v \rightarrow u$ is the edge $v \rightarrow u'$. Thus, $u' \in \pyes$ is the corresponding petal of $u$.
We can similarly go from the yes case to the no case.

Note that $\pi'$ is the same as $\pi$ on all vertices in $D$ (dummy vertices), all vertices in $S$ (star centers), and all isolated vertices in $\eobs^C$ (the observed graph inside the core). 
It differs only for vertices in $P:= \pno \cup \pyes$ under permutation $\pi$.
Note that the sets $\pyes$ and $\pno$ are different under $\pi$ but $\pno$ under $\pi$ will be the same as $\pyes$ under $\pi'$.
%Consider $v\in S$ and let $\nno(v)$ represent $v$'s neighbors in the core $\gcno$ and $\nyes$ represent $v$'s neighbors in the core $\gcyes$.
We define $\ppair$ to be the set of \textbf{petal pairs} where each petal pair is a set ($(u,u') \in \ppair$) containing a petal observed in the no case ($u \in \pno$) and its corresponding petal observed in the yes case ($u' \in \pyes$).
In $\pi'$ we swap the labels within each petal pair in $\ppair$.
Note that this is possible because $\eobs^C$ is a disjoint union of stars, so observed vertices in the no case and the yes case do not overlap for different vertices in $S$.

$\pi'$ is called a dual permutation of $\pi$ because if we start with permutation $\pi'$ then its dual permutation will be $\pi$ since we will swap back the labels.
\begin{claim}
	If we start with a labeling permutation $\pi$ and get $\pi'$ as the dual permutation then if we start with $\pi'$ we will get $\pi$ as the dual permutation.
\end{claim}
\begin{proof}
	Note that we start with the labeling permutation $\pi$. The set $S$ of star centers is fixed along with their core indices ($\psi(S)$).
	We then swap the names of the observed neighbors (petals) of vertices in $S$ in the no and the yes case.
	Notice however that when we start with the labeling permutation $\pi'$, the queries on vertices in $S$ will be the same and $\psi(S)$ is the same, so we will find exactly the same neighbors. This means that the underlying vertices whose names will be swapped will be identical under $\pi$ and $\pi'$. Hence, swapping twice gives us back $\pi$.
\end{proof}

This fixes the problem we had because we can use $\pi$ in the no case and $\pi'$ in the yes case. This way in both cases the edge the algorithm sees will be $(v,u)$.
Note that the dual permutation is a function of the observed graph. So if we are in another fixing then the dual permutation for $\pi$ may not be $\pi'$.

We are now ready to define blocks such that $\mm$ is block-compliant with respect to it.
Let $\xx_2$ be the subset of $\xx_1$ where the set of star centers is $S$, the core indices in their adjacency lists are $\psi(S)$ with permutations $\sigma_C(S)$ and the labeling permutation is either $\pi$ or $\pi'$.
Define $\yy_2$ similarly.
We will ensure that $\mm$ is confined to $(\xx_2,\yy_2)$.
This makes sense because if the star centers, their adjacency lists or the labeling permutations are different, the observed graph may be different.
However, edges of $\mm$ will go from nodes that have permutation $\pi$ in the no case to nodes that have permutation $\pi'$ in the yes case and vice versa.

%\begin{claim}
%	A block with such a fixing is completely good or completely bad.
%\end{claim}
%
%\begin{claim}
%	The sizes of these blocks are the same.
%\end{claim}
%
%\vihan{dont need to prove the above two}

\begin{claim}\label{clm:same-obs-graph-core}
	For any node in $\xx_2$ with labeling permutation $\pi$ and any node in $\yy$ with labeling permutation $\pi'$, the observed graph in the core $\gobs^C$ is the same.
\end{claim}
\begin{proof}
	Since $\xx_2$ is a finer division of $\xx_1$, any node in it has the core graph $\gcno$.
	Similarly, any node in $\yy_2$ has core graph $\gcyes$ which is the coupled graph of $\gcno$.
	The set of star centers $S$ is the same for all nodes in $\xx_2$ and $\yy_2$.
	Finally, the set of petals for a node in $\xx_2$ under $\pi$ has the same labels as the set of petals for a node in $\yy_2$ under $\pi'$ by definition of $\pi'$.
	Therefore, the observed graph in the core $\gobs^C$ is the same for all these nodes.
\end{proof}

\begin{table}[t]
	\footnotesize
	\centering
	\renewcommand{\arraystretch}{1.2}
	\begin{tabular}{|c|>{\raggedright\arraybackslash}p{12.8cm}|}
		\hline
		\textbf{Symbol / Term} & \textbf{\makebox[\linewidth][c]{Description}} \\
		\hline
		$P := \pno \cup \pyes$ & Set of petal vertices observed in either the no or yes case under a fixed permutation $\pi$. \\
		$\ppair$ & Set of petal pairs $(u, u')$ where $u \in \pno$ and $u'$ is its \emph{corresponding petal} in $\pyes$. \\
		$\pi, \pi'$ & Labeling permutation $\pi$ and its \emph{dual permutation} $\pi'$; $\pi'$ is obtained by swapping labels within petal pairs under labeling permutation $\pi$. \\
		$\xx_2, \yy_2$ & Subsets of $\xx_1, \yy_1$ with the following fixing: labeling permutations are $\pi$ or $\pi'$, star centers are $S$ and the randomness of $\psi(S),\sigma_C(S)$ in the adjacency lists is fixed. \\
		$\xx_3, \yy_3$ & Subsets of $\xx_2, \yy_2$ with the following fixing: $\ell(C \setminus P)$ (which include edges to $D$), edges to $D$ for queried locations for all $v \in P$, and edge to $D$ for pairs $(u,u') \in \ppair$.\\
		$\snl_{u,u'}$ & Shared neighbor list of petal pair $(u, u')$—unordered list of edges to $D$ assigned to the pair $(u,u')$ to be split later between $u$ and $u'$. \\
		$\snlobs$ & Discovered portion of the shared neighbor list via queries to $D$. \\
		\hline
	\end{tabular}
	\caption{Notation for \Cref{subsec:coupling-perm,subsec:coupling-edges-to-D}.}
\end{table}

\begin{center}
	\begin{tcolorbox}[colback=gray!5!white, colframe=black!50, boxrule=0.5pt, arc=4pt, width=0.95\linewidth, left=6pt, right=6pt, top=6pt, bottom=6pt]
		\textbf{Fixings So Far}
		\vspace{4pt}
		
		\begin{itemize}
			\item \Cref{subsec:Coupling-the-Core}: Fixed structure of core: $\gcno$ in the no case and $\gcyes$ in the yes case.
			\item \Cref{subsec:coupling-perm}: Fixed set of star centers $S$ along with $\psi(S)$ and $\sigma_C(S)$. Also, fixed the labeling permutations $\pi$ and its dual $\pi'$.
		\end{itemize}
	\end{tcolorbox}
\end{center}

\subsection{Fixing edges to D and Adjacency Lists of Non-Petals}\label{subsec:coupling-edges-to-D}

\begin{figure}[t]
	\centering

	\begin{minipage}{0.44\textwidth}
		\centering
		\input{distribution}
		\caption*{(a) This figure shows coupling of the core for graph $\hh$ and how the matching stays block-compliant for the blocks and the final sub-blocks.}
	\end{minipage}
	\hfill
	\begin{minipage}{0.44\textwidth}
		\centering
		\input{distribution2}
		\caption*{(b) This figure shows that the matching is block-compliant for the final blocks and how we pick the final edges of $\mm$.}
	\end{minipage}

	\caption{Showing the block-compliant matching $\mm$.}
	\label{fig:dist}
\end{figure}

We are currently in a block where the structure in the core $C$ is fixed, the set of star centers $S$ is fixed along with $\psi(S),\sigma_C(S)$ and the labeling permutation is either $\pi$ or $\pi'$.

We now describe the fixing we have for the edges to $D$.
We fix the edges to $D$ for all vertices in $C - P$ (the non-petal vertices in the core).
For vertices in $P$ we also fix edges to $D$ for each petal pair in $\ppair$ i.e. if $(u, u') \in \ppair$ is a pair then we fix edges to $D$ for $u \cup u'$ but do not specify the exact endpoint. 
The number of neighbors to $D$ we have to pick for $u \cup u'$ is $2 d^* - n^{\delta} -1$.
This is because one petal vertex is in $A$ and the other is in $B$, so they have $n^{\delta}$ and $1$ core indices respectively.
To do this we take the union of a set of $d^*-1$ distinct neighbors in $D$ with a set of $d^*-n^{\delta}$ distinct neighbors in $D$.
If the intersection size of these two sets is outside the range $n^{1-\delta}/4 \pm 3 \sqrt{\log n} \cdot n^{(1-\delta)/2}$ for any petal pair then the corresponding block is bad. Using \Cref{clm:sample-overlap} and union bounding over all petal pairs we can say that this happens with probability $1/n$.

We call this unordered list the \textbf{shared neighbor list} of $u$ and $u'$ ($\snl_{u,u'}$).
We have that $\card{\snl_{u,u'}} = 2 d^* - n^{\delta} -1$.
When we later partition those edges between $u$ and $u'$, we will add these edges to their adjacency lists.
Also, note that the maximum frequency of an element of $D$ in $\snl_{u,u'}$ is $2$ because we have a union of two lists each of which contains distinct elements of $D$. If some element appears twice in $\snl_{u,u'}$ then it will be a neighbor of both $u$ and $u'$.

\begin{claim}\label{clm:sample-overlap}
	Let $S_1$ be sample of $D^R$ of size $d^*-1$ without replacement and let $S_2$ be sample of $D^R$ of size $d^*-n^{\delta}$ without replacement.
	The intersection size of $S_1$ and $S_2$ is concentrated in $n^{1-\delta}/4 \pm 3 \sqrt{\log n} \cdot n^{(1-\delta)/2}$ with high probability.
\end{claim}
\begin{proof}
	To show bounds on the intersection size we will fix the randomness of $S_1$ and then use the randomness of $S_2$ to get the bounds on the size.
	Let $X_i$ be the random variable which is $1$ if the $i^{th}$ sample of $S_2$ lands in $S_1$ and $0$ otherwise for $i \in [d^*-n^{\delta}]$.

	Even though the draws are \emph{without} replacement, every element of the ground set is equally likely to appear in position $i$.  Hence, we have $\expect{X_i} = \prob{X_i=1} = \frac{d^*-1}{n^{1-\delta}}$.
	Let $X=\sum_i X_i$ be the random variable denoting the intersection size.
	Using linearity of expectation we get:
	\[
	\expect{X}
	\;=\;
	\sum_{i=1}^{d^*-n^{\delta}} \expect{X_i}
	\;=\;
	(d^*-n^{\delta}) \cdot \frac{d^*-1}{n^{1-\delta}}
	\;=\;
	n^{1-\delta}/4 - n^{\delta}/2 -1/2 +n^{2\delta-1}.
	\]
	Since elements in $S_2$ are sampled uniformly without replacement, the variables $\{X_i\}$ are negatively associated. Since each $X_i$ is in $[0,1]$, we can apply the Chernoff bound in \Cref{prop:chernoff} with $\eps = 10 \sqrt{\log n/n^{1-\delta}}$:
	\begin{align*}
		\prob{\card{X - \expect{X}} > \eps \expect{X} } &\leq 2 \exp \left( - \eps^2 \cdot \expect{X}/4 \right) \\
		&\leq 2 \exp \left( - \eps^2 \cdot n^{1-\delta}/20 \right) \\
		&\leq 2 \exp \left( - 100 \log n /20 \right) \\
		&\leq n^{-4}.
	\end{align*}
	Therefore, with high probability, $X$ is in $n^{1-\delta}/4 \pm 3 \sqrt{\log n} \cdot n^{(1-\delta)/2}$.
\end{proof}

We also fix the adjacency list $\ell(v)$ for all $v \in V-P$.
$\psi(v),\sigma_C(v)$ have already been fixed for $v\in S$ so we only fix $\sigma_D(v)$ for $v\in S$. We now also fix this for other vertices in $C-P-S$.
The only vertices in $C-P-S$ are the isolated vertices in $\eobs^C$ (all queries on them go to $D$).
We fix their entire adjacency lists.
Note that we can fix $\ell(D)$ because for all $v \in D$, their edges have been fixed (up to switching within petal pairs).
So their adjacency lists are not completely fixed because they may have an edge to petal pair $(u,u')$ which implies that the endpoint of the edge is not fixed.
However, if it has two edges to a petal pair $(u,u')$ then we arbitrarily fix one of neighbors as $u$ and the other one as $u'$.
Note that any edge to a petal pair $(u,u')$ corresponds to a fixed element in $\snl_{u,u'}$. So it is a function of whether that particular neighbor is assigned to $u$ or $u'$.
If this fixing violates \Cref{clm:deg-D-verts} then the corresponding block is bad, and we ignore it. We can do this because the degrees of all vertices in $D$ have been fixed.

Fixing $\ell(D)$ fixes the observed edges from $D$ to $P$ but the endpoints of these edges are not fixed within the petal pairs.
These observed edges are present in the shared neighbor list $\snl$.
For $(u,u') \in \ppair$, $\snlobs$ represents the \emph{observed} part of the shared neighbor list of $(u,u')$.
The only queried edges that are not fixed are the ones in $P$.
Thus, for vertices in $P$, we now fix edges to $D$ for all the queried locations in their adjacency lists.
We will do this by picking neighbors from $\snl$. These can be the observed elements ($\snlobs$), the duplicate elements, or remaining ones which are not observed and have just one copy.

Consider vertex $u \in P$ in the no case where labeling permutation $\pi$ is used. The vertex $u$ has $q_u$ queried locations out of $d^*$ locations.
The vertex $u \in P$ in the yes case where labeling permutation $\pi'$ is used, is a different vertex in terms of its name in the underlying structure, but it has the same set of $q_u$ out of $d^*$ queried indices.
Thus, we fix the edges to $D$ in these locations for all petals in $P$.
After this the size of $\snl_{u,u'}$ is $\snlsize:=2d^* -1 - n^{\delta} -q_u - q_{u'}$.
Note that when fixing these edges to $D$ we also have to fix their \emph{order}. 
This fixes some randomness of $\sigma_D(P)$.

We are now ready to define blocks such that $\mm$ is block-compliant with respect to it.
Let $\xx_3$ be the subset of $\xx_2$ with all the above fixings and define $\yy_3$ similarly.
We will ensure that $\mm$ is confined to $(\xx_3,\yy_3)$.
All the edges of the graph have been fixed except for endpoints of some edges that can switch within petal pairs.
Note that since the queries are fixed, this also fixes the observed edges from $D$.
The observed edges from $D$ to $C - P$ are fixed and the observed edges to $P$ are fixed in the shared neighbor lists. 
The edges in $\snlobs$ will be part of the observed graph irrespective of whether they are assigned to $u$ or $u'$. 
We now show that the observed graph is almost fixed within the blocks.
\begin{claim}\label{clm:same-obs-graph-except-petals}
	For any node in $\xx_3$ with permutation $\pi$ and any node in $\yy_3$ with permutation $\pi'$, the observed graph $\gobs$ is the same up to edges from $D$ to $P$ switching within petal pairs. Also, the observed edges are fixed in the shared neighbor lists.
\end{claim}
\begin{proof}
	We know from \Cref{clm:same-obs-graph-core} that the observed graph in the core $\gobs^C$ is the same for these nodes.
	For all vertices in $C -P$, we have fixed their edges to $D$ and their adjacency lists. Since the queries are deterministically fixed, all the observed edges from these vertices will be identical for all nodes under consideration.
	For vertices in $P$, we have fixed the edges in the queried locations implying that all the observed edges from the petal vertices will be identical for all nodes under consideration.
	
	Finally, for vertices in $D$, their adjacency lists have been fixed except some edges are to petal pairs instead of the explicit vertex within the petal pair.
	However, the observed edges are fixed in the shared neighbor lists for petal pairs (specific positions of $\snl_{u,u'}$ are part of the observed edges irrespective of whether they are assigned to $u$ or $u'$).
	Therefore, $\gobs$ is the same for all nodes under consideration up to edges from $D$ to $P$ switching within petal pairs. 
\end{proof}

\begin{center}
	\begin{tcolorbox}[colback=gray!5!white, colframe=black!50, boxrule=0.5pt, arc=4pt, width=0.95\linewidth, left=6pt, right=6pt, top=6pt, bottom=6pt]
		\textbf{Fixings So Far}
		\vspace{4pt}
		
		\begin{itemize}
			\item \Cref{subsec:Coupling-the-Core}: Fixed structure of core: $\gcno$ in the no case and $\gcyes$ in the yes case.
			\item \Cref{subsec:coupling-perm}: Fixed set of star centers $S$ along with $\psi(S),\sigma_C(S)$ and labeling permutations $\pi,\pi'$.
			\item \Cref{subsec:coupling-edges-to-D}: Fixed edges to $D$ for all vertices in $C-P$ and for vertices in $P$ fixed edges to $D$ for petal pairs using $\snl$. Fixed $\ell(V-P)$ except for neighbors of $D$ in $P$ which can switch within petal pairs. Finally, fixed all queried indices in $\ell(P)$.
		\end{itemize}
	\end{tcolorbox}
\end{center}

\subsection{Fixing Final Blocks}\label{subsec:coupling-final-blocks}

We now fix what the final blocks are.
We know that the observed edges from $D$ to $P$ are fixed but the endpoints of these edges are not fixed within the petal pairs (\Cref{clm:same-obs-graph-except-petals}).
However, these observed edges are fixed in the shared neighbor list.
So for each petal pair we fix the endpoints for these observed edges ($\snlobs$) by partitioning the observed edges in the shared neighbor lists between the petals in the petal pairs.
Note that this could be any arbitrary partitioning for the observed edges $\snlobs$ for all $(u,u') \in \ppair$. 
We let $\snlobs(u)$ be the subset of $\snlobs$ assigned to $u$ with size $\snlcount_u$ and let $\snlobs(u')$ be the subset of $\snlobs$ assigned to $u'$ with size $\snlcount_{u'}$.

If this fixing violates \Cref{clm:indeg-from D} then the corresponding block is bad, and we ignore it. We can do this because the observed edges from $D$ have been fixed.
We are now ready to define the final blocks such that $\mm$ is block-compliant with respect to it.
Let $\xx_f$ be the subset of $\xx_3$ with this fixing that partitioning of observed edges in the shared neighbor lists of the petals, and define $\yy_f$ similarly.
We will ensure that $\mm$ is confined to $(\xx_f,\yy_f)$.
This final fixing has now completely fixed the observed graph:
\begin{claim}\label{clm:same-obs-graph-full}
	For any node in $\xx_f$ with permutation $\pi$ and any node in $\yy_f$ with permutation $\pi'$, the observed graph $\gobs$ is the same.
\end{claim}
\begin{proof}
	We know from \Cref{clm:same-obs-graph-except-petals} that the observed graph $\gobs$ is the same up to edges from $D$ to $P$ switching within petal pairs.
	Fixing the partitioning of the queried edges from $D$ in $\snl_{u,u'}$ for all $(u,u') \in \ppair$ fixed all the observed edges from $D$ to $P$ (this has removed all the randomness of endpoints of edges switching within petal pairs).
	Thus, all the edges that are observed by $\Alg$ are now same among all nodes in $\xx_f$ under $\pi$ and $\yy_f$ under $\pi'$ making $\gobs$ the same.
\end{proof}

\begin{center}
	\begin{tcolorbox}[colback=gray!5!white, colframe=black!50, boxrule=0.5pt, arc=4pt, width=0.95\linewidth, left=6pt, right=6pt, top=6pt, bottom=6pt]
		\textbf{Fixings So Far}
		\vspace{4pt}
		
		\begin{itemize}
			\item \Cref{subsec:Coupling-the-Core}: Fixed structure of core: $\gcno$ in the no case and $\gcyes$ in the yes case.
			\item \Cref{subsec:coupling-perm}: Fixed set of star centers $S$ along with $\psi(S),\sigma_C(S)$ and labeling permutations $\pi,\pi'$.
			\item \Cref{subsec:coupling-edges-to-D}: Fixed edges to $D$ for $C-P$ and for $P$ fixed edges to $D$ for petal pairs using $\snl$. Fixed $\ell(V-P)$ except for neighbors of $D$ in $P$ which can switch within petal pairs. Finally, fixed all queried indices in $\ell(P)$.
			\item \Cref{subsec:coupling-final-blocks}: Fixed the split of observed edges in $\snl$ ($\snlobs$) between vertices in petal pairs. This fixes the entire observed graph.
		\end{itemize}
		\vspace{4pt}
		\textbf{Randomness left:} Splitting remaining edges in $\snl$ and placing them in the adjacency lists of the petals. Note that the petals only have elements in the queried indices and all other indices are empty.
	\end{tcolorbox}
\end{center}

\begin{table}[t]
	\footnotesize
	\centering
	\renewcommand{\arraystretch}{1.2}
	\begin{tabular}{|c|>{\raggedright\arraybackslash}p{12.2cm}|}
		\hline
		\textbf{Symbol / Term} & \textbf{\makebox[\linewidth][c]{Description}} \\
		\hline
		$\xx_f, \yy_f$ & Final blocks of nodes in $\xx$ and $\yy$, obtained by taking subsets of $\xx_3$ and $\yy_3$ with the following fixing: For all $(u,u') \in \ppair$, arbitrarily split up $\snlobs$ between $u$ and $u'$. \\
		$\xx_f^{\pi}, \xx_f^{\pi'}$, $\yy_f^{\pi}, \yy_f^{\pi'}$ & Subsets of $\xx_f$ and $\yy_f$ where labeling permutations are fixed to either $\pi$ or $\pi'$. \\
		$\snlobs(u)$ & Subset of observed edges from the shared neighbor list $\snl_{u,u'}$ assigned to petal $u$. $\snlcount_u := \card{\snlobs(u)}$ \\
		\hline
	\end{tabular}
	\caption{Notation for \Cref{subsec:coupling-final-blocks}.}
\end{table}

The observed graph $\gobs$ is the same for any node in $\xx_f$ with permutation $\pi$ and any node in $\yy_f$ with permutation $\pi'$.
The same proof also shows that the observed graph $\gobs$ is the same for any node in $\xx_f$ with permutation $\pi'$ and any node in $\yy_f$ with permutation $\pi$.
We let $\xx_f^{\pi}$ represent the part of block with permutation $\pi$ and $\xx_f^{\pi'}$  represent the part of block with permutation $\pi'$.
Define $\yy_f^{\pi}$ and $\yy_f^{\pi'}$ similarly.
The only randomness left now is in partitioning the shared neighbor list between the petals and $\psi(v),\sigma_C(v)$ and $\sigma_D(v)$ for all $v \in P$.
We now show that the sizes of the blocks $\xx_f$ and $\yy_f$ are the same.
\begin{claim}
	$\card{\xx_f^{\pi}} =\card{\yy_f^{\pi}}$, $\card{\xx_f^{\pi'}} = \card{\yy_f^{\pi'}}$ and  $\card{\xx_f} = \card{\yy_f}$.
\end{claim}
\begin{proof}
	Consider labeling permutation $\pi$ and any petal pair $(u,u') \in \ppair$ where $u' \in A$ and $u \in B$. The following calculations will be the same in the yes and the no case under the same permutation because the partitioning and picking core indices does not depend on it.
	
	Under labeling permutation $\pi$, $u$ has one core index and $u'$ has $n^{\delta}$ core indices.
	Also, the $q_u$ queried positions of $u$ and the $q_{u'}$ positions of $u'$ are already filled.
	Thus, $u$ needs $d^* - 1 -q_u$ elements of $\snl_{u,u'}$ and $u'$ needs the remaining $d^* - n^{\delta} - q_{u'}$.
	We have already assigned $\snlcount_u$ neighbors to $u$ and $\snlcount_{u'}$ neighbors to $u'$ from the shared neighbor list (apart from the $q_u$ and $q_{u'}$ elements that have already been put in the adjacency lists of $u$ and $u'$).
	Thus, we need to assign $d^* - 1 -q_u - \snlcount_u$ neighbors to $u$ out of $\snlsize - \snlcount_{u,u'}$ neighbors where $\snlcount_{u,u'}:= \card{\snlobs}$.
	We need to be careful here because duplicate elements ($\snldup$) have to be split between $u$ and $u'$.
	Some of them were already placed in the queried positions of $u$ and $u'$ so the remaining ones have count $\snldupct_{u,u'}=\snldupct_u+\snldupct_{u'}$.
	Thus, under permutation $\pi$, the number of choices to pick the remaining dummy neighbors of $u$ is
	$\binom{\snlsize - \snlcount_{u,u'} - \snldupct_{u,u'}}{d^* - 1 -q_u - \snlcount_u-\snldupct_u}$.
	
	We also have to choose the core indices for $u$ and $u'$ which can be done in $d^*- q_u$ ways for $u$ and $\binom{d^*-q_{u'}}{n^{\delta}}$ ways for $u'$.
	Fixing this fixes the randomness of $\psi(u),\psi(u')$.
	The number of permutations $\sigma_C(u')$ is $(n^{\delta})!$ and $\sigma_C(u)$ is $1$.
	The number of permutations $\sigma_D(u')$ is $(d^*-n^{\delta}-q_{u'})!$ and $\sigma_D(u)$ is $(d^*-1-q_u)!$.
	This fixes the randomness of $\sigma_C,\sigma_D$ along with $\psi$.
	
	We note that all of these steps are independent of each other.
	This implies that the number of fixings for pair $(u,u')$ is
	\[
		f_{u,u'}^{\pi}:= \binom{\snlsize - \snlcount_{u,u'} - \snldupct_{u,u'}}{d^* - 1 -q_u - \snlcount_u-\snldupct_u}   \left[ \binom{d^*-q_{u'}}{n^{\delta}} (d^*- q_u) (n^{\delta})! (d^*-1-q_u)!  (d^*-n^{\delta}-q_{u'})!  \right].
	\]
	Thus, the size of $\xx_f^{\pi}$ and $\yy_f^{\pi}$ is:
	\[
		\prod_{(u,u') \in \ppair}	\binom{\snlsize - \snlcount_{u,u'} - \snldupct_{u,u'}}{d^* - 1 -q_u - \snlcount_u-\snldupct_u}  \left[ \binom{d^*-q_{u'}}{n^{\delta}} (d^*- q_u) (n^{\delta})! (d^*-1-q_u)!  (d^*-n^{\delta}-q_{u'})!  \right].
	\]
	
	We can do a similar process under the permutation $\pi'$.
	The only difference there is that $u \in A$ and $u' \in B$.
	Thus, we get that the size of $\xx_f^{\pi'}$ and $\yy_f^{\pi'}$ is:
	\[
	\prod_{(u,u') \in \ppair}	\binom{\snlsize- \snlcount_{u,u'} - \snldupct_{u,u'}}{d^* - 1 -q_{u'} - \snlcount_{u'}-\snldupct_{u'}}  \left[ \binom{d^*-q_{u}}{n^{\delta}} (d^*- q_{u'}) (n^{\delta})! (d^*-1-q_{u'})!  (d^*-n^{\delta}-q_{u})!  \right].
	\]
	These also imply that $\card{\xx_f} = \card{\yy_f}$.
\end{proof}

Recall that $\xx_f$ with $\pi$ has the same observed graph as $\yy_f$ with $\pi'$ (\Cref{clm:same-obs-graph-full}).
The same proof also gives us that $\xx_f$ with $\pi'$ has the same observed graph as $\yy_f$ with $\pi$.
Thus, we need to add edges of $\mm$ between $\xx_f^{\pi}$ and $\yy_f^{\pi'}$ (similarly for  $\xx_f^{\pi'}$ and $\yy_f^{\pi}$).
We will add an arbitrary matching between $\xx_f^{\pi}$ and $\yy_f^{\pi'}$ which is perfect on the smaller side (see \Cref{fig:dist}). This means that the size of the matching will be $\min \left( \card{\xx_f^{\pi}} , \card{\yy_f^{\pi'}} \right)$.
We will get a matching of the same size between $\xx_f^{\pi'}$ and $\yy_f^{\pi}$.
We now have to show that the size of the matching is large enough. To show that we will show that these sets are of a comparable size.
To prove this claim we first show that for any pair $(u,u') \in \ppair$ their number of possible fixings are similar under $\pi$ and $\pi'$.
\begin{claim}
	We have $\card{\yy_f^{\pi'}} = (1\pm o(1)) \card{\xx_f^{\pi}}$.
\end{claim}
\begin{proof}
	Without loss of generality let $\card{\xx_f^{\pi}} \geq \card{\yy_f^{\pi'}}$. Then we have to show $\card{\yy_f^{\pi'}} \geq (1-o(1)) \card{\xx_f^{\pi}}$.
	We will first bound $f_{u,u'}^{\pi'}/f_{u,u'}^{\pi}$ for any $(u,u') \in \ppair$ and then use that to bound $\card{\yy_f^{\pi'}}/ \card{\xx_f^{\pi}}$.
We have the following bounds:
\begin{itemize}
	\item $\card{\snl_{u,u'}} =\snlsize= 2d^* - q_u - q_{u'} - n^{\delta} - 1$.
	\item $\snlcount_u, \snlcount_{u'} \leq 4n^{\eps}$  (by \Cref{clm:indeg-from D}).
	\item $\card{\snlobs}= \snlcount_{u,u'} = \snlcount_u + \snlcount_{u'}$ since the observed edges belong to either $u$ or $u'$.
	\item $n^{1-\delta}/2 - 6 \sqrt{\log n} \cdot n^{(1-\delta)/2} -q_u -q_{u'} \leq \snldupct_{u,u'} \leq n^{1-\delta}/2 + 6 \sqrt{\log n} \cdot n^{(1-\delta)/2}$.
	\item $n^{1-\delta}/4 - 3 \sqrt{\log n} \cdot n^{(1-\delta)/2} -q_u \leq \snldupct_{u} \leq n^{1-\delta}/4 + 3 \sqrt{\log n} \cdot n^{(1-\delta)/2}$ (similar for $\snldupct_{u'}$).
	\item $\alpha= \log (n^{\eps} \log n)$ and $\beta = \log (16 n^{2\eps+\delta} \log^2 n)$ (defined in \Cref{clm:discovered-heavy-edges}).
%	\item Let $\tau = 2^i$ be the smallest power of $2$ such that $\tau \geq q_u$, for some $i \in [\alpha, \beta]$ (\Cref{clm:discovered-heavy-edges}).
	
\end{itemize}

We define the following:
\begin{align*}
	\kappa &:= \snlsize - \snlcount_{u,u'} -\snldupct_{u,u'} \leq 2d^*, \\
	\lambda_u &:= d^* - 1 - q_u - \snlcount_u -\snldupct_u , \\
	\lambda_{u'} &:= d^* - 1 - q_{u'} - \snlcount_{u'} -\snldupct_{u'}, \\
	\zeta &:= 3 \sqrt{\log n} \cdot n^{(1-\delta)/2}.
\end{align*}

	The values of $f_{u,u'}^{\pi}$ and $f_{u,u'}^{\pi'}$ are:
	\[
	f_{u,u'}^{\pi}:= \binom{\snlsize - \snlcount_{u,u'} - \snldupct_{u,u'}}{d^* - 1 -q_u - \snlcount_u-\snldupct_u}   \left[ \binom{d^*-q_{u'}}{n^{\delta}} (d^*- q_u) (n^{\delta})! (d^*-1-q_u)!  (d^*-n^{\delta}-q_{u'})!  \right].
	\]
	\[
	f_{u,u'}^{\pi'}:=\binom{\snlsize- \snlcount_{u,u'} - \snldupct_{u,u'}}{d^* - 1 -q_{u'} - \snlcount_{u'}-\snldupct_{u'}}  \left[ \binom{d^*-q_{u}}{n^{\delta}} (d^*- q_{u'}) (n^{\delta})! (d^*-1-q_{u'})!  (d^*-n^{\delta}-q_{u})!  \right].
	\]

	Our goal is to lower bound the ratio
	$r := \frac{f_{u,u'}^{\pi'}}{f_{u,u'}^{\pi}}$.
	To do this, we decompose $r$ into three parts:
	$
	r = r_1 \cdot r_2 \cdot r_3 ,
	$
	where:
	\begin{itemize}
		\item $r_1$ is the ratio of the first binomial terms:
		\[
		r_1 := \frac{\binom{\kappa}{\lambda_{u'}}}{\binom{\kappa}{\lambda_u}},
		\]
		\item $r_2$ is the ratio of the second binomial terms and the last two terms:
		\[
		r_2 :=  \frac{\binom{d^* - q_u}{n^{\delta}} }{\binom{d^* - q_{u'}}{n^{\delta}}} \cdot \frac{ (d^*-1-q_{u'})!  (d^*-n^{\delta}-q_{u})! }{(d^*-1-q_u)!  (d^*-n^{\delta}-q_{u'})!}.
		\]
		\item $r_3$ is the ratio of the remaining two terms:
		\[
		r_3 :=\frac{ (d^* - q_{u'})\cdot (n^{\delta})!}{(d^* - q_u)\cdot (n^{\delta})!}.
		\]
	\end{itemize}
	We will lower bound each of them separately to get a bound on $r$.
We first show that $\kappa$ is always between $d^* - o(d^*)$ and $d^* +2\zeta$. First consider the lower bound:
\[
\kappa = 2d^* - q_u -q_{u'}-n^{\delta} - 1 - \snlsize- \snldupct_{u,u'} \geq 2d^* -2^{\beta+1} -1 - n^\delta - 8 n^{\eps} - d^* -2\zeta \geq d^* -3\cdot 2^{\beta} -2\zeta \geq d^* -o(d^*)
\]
where the last inequality holds because $2^\beta +\zeta=o(d^*)$ (\Cref{constraint:matching-size}).
We now show the upper bound:
 \[
 \kappa = 2d^* - q_u -q_{u'}-n^{\delta} - 1 - \snlsize- \snldupct_{u,u'} \leq 2d^* - q_u -q_{u'} - (d^* -2\zeta - q_u -q_{u'}) \leq d^* +2\zeta
 \]
where the last inequality holds because $\zeta=o(d^*)$.
%To bound $r_1$ from below, we first need a lower bound on $\lambda_u$ in terms of $\kappa$. We proceed as follows:
%	\begin{align*}
%	\lambda_u &= d^* -1 -q_u  -\snlcount_{u} -\snldupct_u \\ 
%	&\geq d^*-1 -q_u - \snlcount_{u} - d^*/2 -  \zeta \\
%	&= \kappa/2 + d^*/2 -1 -q_u  -\snlcount_{u} - \zeta - \kappa/2 \\
%	&\geq \kappa/2 +d^*/2 -2\tau -\zeta -d^*/2 -\zeta \\
%	&= \kappa/2 -2\tau -2\zeta.
%	\end{align*}
	We now start by lower bounding $r_1$:
	\begin{align*}
		r_1 &= \binom{\kappa}{\lambda_{u'}} \binom{\kappa}{\lambda_{u}}^{-1} \\
		&= \frac{\lambda_{u}! \cdot (\kappa - \lambda_{u})!}{\lambda_{u'} ! \cdot (\kappa - \lambda_{u'})!} \\
		&= \frac{\lambda_{u}! \cdot (\lambda_{u'} - (n^{\delta} -1))!}{(\lambda_{u}-(n^{\delta}-1)) ! \cdot (\lambda_{u'})!} \tag{By definition of $\lambda_u,\lambda_{u'}$} \\
		&= \frac{(\lambda_u) \cdot (\lambda_u-1) \ldots (\lambda_u - n^{\delta} +2)}{(\lambda_{u'}) \cdot (\lambda_{u'}-1) \ldots (\lambda_{u'} - n^{\delta} +2)} \\
		&\geq \left( \frac{\lambda_u -n^{\delta}+2}{\lambda_{u'}} \right)^{n^\delta -1} \\
		&\geq \left( \frac{d^* -q_u -n^{\delta} - \snlcount_u -\snldupct_u }{d^* - \snldupct_{u'}} \right)^{n^\delta} \\
		&\geq \left( \frac{d^* -q_u -n^{\delta} - \snlcount_u -\snldupct_{u'} - q_{u'} }{d^* - \snldupct_{u'}} \right)^{n^\delta} \tag{$\snldupct_{u} \leq \snldupct_{u'} + q_{u'}$} \\
		&= \left(1 - \frac{ q_u +q_{u'} + n^{\delta} + 4n^{\eps} }{d^* - \snldupct_{u'}} \right)^{n^\delta} \tag{$\snlcount_u \leq 4n^{\eps}$} \\
		&\geq \exp \left( - 16 (q_u +q_{u'} + n^{\delta} + 4n^{\eps}) \cdot n^{2 \delta-1} \right) \tag{Since $1-x\geq \exp(-2x)$ for $x \in [0,1/2]$}
	\end{align*}
	where the last inequality holds because $d^* - \snldupct_{u'} \geq d^*/4$.
	We now consider $r_2$:	
	\begin{align*}
		r_2 &= \frac{\binom{d^* - q_u}{n^{\delta}} }{\binom{d^* - q_{u'}}{n^{\delta}}} \cdot \frac{ (d^*-1-q_{u'})! \cdot (d^*-n^{\delta}-q_{u})! }{(d^*-1-q_u)!  \cdot (d^*-n^{\delta}-q_{u'})!} \\
		&=\frac{(d^* - q_u)! \cdot (d^*-q_{u'} - n^{\delta})! }{(d^* - q_{u'})! \cdot (d^*-q_u - n^{\delta})!} \cdot \frac{ (d^*-1-q_{u'})! \cdot (d^*-n^{\delta}-q_{u})! }{(d^*-1-q_u)! \cdot (d^*-n^{\delta}-q_{u'})!} \\
		&=\frac{d^*-q_u}{d^*-q_{u'}}.
	\end{align*}
	Notice that $r_3$ is exactly the reciprocal of $r_2$:
	\begin{align*}
		r_3 &= \frac{ (d^* - q_{u'})\cdot (n^{\delta})!}{(d^* - q_u)\cdot (n^{\delta})!} = \frac{d^*-q_{u'}}{d^*-q_{u}}.
	\end{align*}
	Thus, we get $r_2 \cdot r_3=1$. This happens because in both cases we have one adjacency list with $q_u$ filled slots and the other with $q_{u'}$ filled slots. Therefore, the number of permutations stays the same.
	Hence, we get $r = r_1 = \exp( - 8 (q_u +q_{u'} + n^{\delta} + 4n^{\eps}) \cdot n^{2 \delta-1})$.
	Note that $q_u$ and $q_{u'}$ are the only terms that depend on the petal pair.
	The final part is to show that $r_f:=\card{\yy_f^{\pi'}} \left( \card{\xx_f^{\pi}} \right)^{-1} \geq (1-o(1))$:
	\begin{align*}
		r_f &=\card{\yy_f^{\pi'}} \left( \card{\xx_f^{\pi}} \right)^{-1} \\
		&= \prod_{(u,u') \in \ppair} f_{u,u'}^{\pi'}/f_{u,u'}^{\pi} \\
		&\geq \exp \left( - \sum_{(u,u') \in \ppair}  16 (q_u +q_{u'} + n^{\delta} + 4n^{\eps}) \cdot n^{2 \delta-1} \right) \\
		&\geq \exp \left(-16 n^{3\delta-1} -64 n^{2\delta+\eps-1}  - 16 n^{2 \delta-1} \sum_{(u,u') \in \ppair}   q_u +q_{u'}  \right)
	\end{align*}
	
	To prove this, we need to upper bound the sum of queries over petals.
	We partition $\ppair$ into different chunks based on the number of queries on the petals.
	For all $i \in [\alpha,\beta]$ chunk $i$ contains all the petal pairs $(u,u')$ with $q_u$ between $2^{i}$ and $2^{i+1}$. We also define chunk $0$ to contain all the petal pairs $(u,u')$ with $q_u$ at most $2^{\alpha}$.
	The number of petal pairs in chunk $i$ is at most $2^{-i} \cdot 8n^{2\eps+\delta} \log^2 n$ since that is an upper bound on the number of $2^i$-heavy petal vertices \Cref{clm:discovered-heavy-edges}.
	For chunk $i$ we have the following:
%	Note that to prove a lower bound, we want to upper bound the number of terms because we are multiplying numbers smaller than $1$.
%	For chunk $i$ we have $\tau=2^{i+1}$ giving us:
	\begin{align*}
		s^{(i)} &:= \sum_{(u,u') \in \text{chunk } i} q_u \leq 2^{i+1} \cdot (2^{-i} \cdot 2^{\beta-1}) =2^{\beta}.
	\end{align*}
	For chunk $0$ we bound the number of queries with $2^\alpha$ and the number of petal pairs is at most $4 n^{\eps +\delta} \log n$. 
	This is because that is the bound on the number of observed edges (\Cref{clm:no-edges-core}).
	This gives us a bound for chunk $0$:
	\begin{align*}
		s^{(0)} &:= \sum_{(u,u') \in \text{chunk } 0} q_u \leq 2^{\alpha} \cdot 4 n^{\eps+\delta} \log n.
	\end{align*}
	This gives us $\sum_{(u,u') \in \ppair} q_u \leq 2^{\alpha+2} \cdot n^{\eps+\delta} \log n + 2^{\beta} \log n$.
	Thus, we finally have:
	\begin{align*}
		r_f &\geq \exp \left(-16 n^{3\delta-1} -64 n^{2\delta+\eps-1}  - 16 n^{2 \delta-1} \sum_{(u,u') \in \ppair}   q_u +q_{u'}  \right) \\
		&\geq \exp \left(-16 n^{3\delta-1} -64 n^{2\delta+\eps-1}  - 16 n^{2 \delta-1} \cdot (2^{\beta+1} \log n + 2^{\alpha+3} \cdot n^{\eps+\delta} \log n )  \right) \\
		&= \exp \left(- \Theta(n^{3\delta-1} +  n^{2\delta+\eps-1}  + n^{3 \delta +2 \eps -1} \log^3 n  ) \right) \\
		&= (1-o(1))
	\end{align*}
	where the last inequality holds because of \Cref{constraint:matching-size}.
\end{proof}

This finally implies that the matching $\mm$ has a large size.
\begin{proof}[Proof of \Cref{lem:large-matching-coupling}]
	We divide $\xx$ and $\yy$ into blocks by fixing some part of the randomness.
	$\xx_f$ and $\yy_f$ represent blocks under some fixing.
	We add an arbitrary matching between $\xx_f^{\pi}$ and $\yy_f^{\pi'}$ of size $\min( \card{\xx_f^{\pi}} , \card{\yy_f^{\pi'}} )$.
	Similarly, we add an arbitrary matching between $\xx_f^{\pi'}$ and $\yy_f^{\pi}$ of size $\min( \card{\xx_f^{\pi'}} , \card{\yy_f^{\pi}} )$.
	These describe the edges of $\mm$.
	
	The size of $\mm$ is exactly equal to the number of matched nodes in $\xx$.
	We have ignored $(1-o(1))$ fraction of the nodes in $\xx$ because of high probability events.
	For the remaining nodes, each of them belong in some final block $\xx_f$.
	For every such block we have matched between $\card{\xx_f}$ and $(1-o(1)) \card{\xx_f}$ nodes.
	Thus, $\card{\mm} \geq (1-o(1)) \cdot (1-o(1)) \card{\xx} = (1-o(1)) \card{\xx}$.
\end{proof}

%% file: coupling_fig2.tex
\begin{tikzpicture}[xscale=0.8, yscale=1.4, every node/.style={scale=1.15}]
	
	% =========================
	% X PIPELINE (nested shading)
	% =========================
	
	% X initial (largest region shaded)
	\begin{scope}[shift={(0,0)}]
		\fill[colorSeven!15, rounded corners] (-1.6,-2.2) rectangle (1.6,2.2);
		\draw[thick, rounded corners] (-1.6,-2.2) rectangle (1.6,2.2);
	\end{scope}
	
	% Core Edges
	\begin{scope}[shift={(5.0,0)}]
		\fill[colorSeven!35, rounded corners] (-1.3,0.3) rectangle (1.3,1.9);
		\draw[thick, rounded corners] (-1.3,0.3) rectangle (1.3,1.9);
		\draw[thick, rounded corners] (-1.3,-1.9) rectangle (1.3,-0.3);
		\draw[thick, rounded corners] (-1.6,-2.2) rectangle (1.6,2.2);
	\end{scope}
	
	% Labeling (smaller top portion shaded)
	\begin{scope}[shift={(10.0,0)}]
		\fill[colorSeven!65] (-1.3,1.1) rectangle (1.3,1.9);
		\draw[thick, rounded corners] (-1.3,0.3) rectangle (1.3,1.9);
		\draw[thick, rounded corners] (-1.3,-1.9) rectangle (1.3,-0.3);
		\draw[opacity=0.6] (-1.3,1.1) -- (1.3,1.1);
		\draw[thick, rounded corners] (-1.6,-2.2) rectangle (1.6,2.2);
	\end{scope}
	
	% Adj List (smallest region shaded)
	\begin{scope}[shift={(15.0,0)}]
		\fill[colorSeven!85] (-1.3,1.5) rectangle (1.3,1.9);
		\draw[thick, rounded corners] (-1.3,0.3) rectangle (1.3,1.9);
		\draw[thick, rounded corners] (-1.3,-1.9) rectangle (1.3,-0.3);
		\draw[opacity=0.6] (-1.3,1.1) -- (1.3,1.1);
		\draw[dashed,opacity=0.6] (-1.3,1.5) -- (1.3,1.5);
		\draw[thick, rounded corners] (-1.6,-2.2) rectangle (1.6,2.2);
	\end{scope}
	
	% X arrows
%	\draw[->, thick] (1.6,0) -- (3.4,0);
\draw[->, thick]
(1.6,0) -- (3.4,0)
node[midway, above=3pt] {$\gcno$};

\draw[->, thick]
(6.6,0) -- (8.4,0)
node[midway, above=3pt] {$\pi$};

\draw[->, thick]
(11.6,0) -- (13.4,0)
node[midway, above=3pt] {$\ell_1$};

%	\draw[->, thick] (6.6,0) -- (8.4,0);
%	\draw[->, thick] (11.6,0) -- (13.4,0);
	
	\node at (0,-3) {$\xx$};
	\node at (5,-3) {Core Edges};
	\node at (10,-3) {Labels};
	\node at (15,-3) {Partial Adj.\ List};

	% =========================
	% Y PIPELINE (nested shading, unequal size)
	% =========================
	
	% Adj List (smallest region shaded)
	\begin{scope}[shift={(20.0,0)}]
		\fill[colorTwo!85] (-1.3,1.65) rectangle (1.3,2.0);
		\draw[thick, rounded corners] (-1.3,0.2) rectangle (1.3,2.0);
		\draw[thick, rounded corners] (-1.3,-2.0) rectangle (1.3,-0.2);
		\draw[opacity=0.6] (-1.3,1.2) -- (1.3,1.2);
		\draw[dashed,opacity=0.6] (-1.3,1.65) -- (1.3,1.65);
		\draw[thick, rounded corners] (-1.6,-2.2) rectangle (1.6,2.2);
	\end{scope}
	
	% Labeling
	\begin{scope}[shift={(25.0,0)}]
		\fill[colorTwo!65] (-1.3,1.2) rectangle (1.3,2.0);
		\draw[thick, rounded corners] (-1.3,0.2) rectangle (1.3,2.0);
		\draw[thick, rounded corners] (-1.3,-2.0) rectangle (1.3,-0.2);
		\draw[opacity=0.6] (-1.3,1.2) -- (1.3,1.2);
		\draw[thick, rounded corners] (-1.6,-2.2) rectangle (1.6,2.2);
	\end{scope}
	
	% Core Edges
	\begin{scope}[shift={(30.0,0)}]
		\fill[colorTwo!35, rounded corners] (-1.3,0.2) rectangle (1.3,2.0);
		\draw[thick, rounded corners] (-1.3,0.2) rectangle (1.3,2.0);
		\draw[thick, rounded corners] (-1.3,-2.0) rectangle (1.3,-0.2);
		\draw[thick, rounded corners] (-1.6,-2.2) rectangle (1.6,2.2);
	\end{scope}
	
	% Y initial
	\begin{scope}[shift={(35.0,0)}]
	\fill[colorTwo!15, rounded corners] (-1.6,-2.2) rectangle (1.6,2.2);
		\draw[thick, rounded corners] (-1.6,-2.2) rectangle (1.6,2.2);
	\end{scope}
	
	% Y arrows
	
	\draw[->, thick]
	(33.4,0) -- (31.6,0)
	node[midway, above=3pt] {$\gcyes$};
	
	\draw[->, thick]
	(28.4,0) -- (26.6,0)
	node[midway, above=3pt] {$\pi'$};
	
	\draw[->, thick]
	(23.4,0) -- (21.6,0)
	node[midway, above=3pt] {$\ell_2$};

%	\draw[->, thick] (33.4,0) -- (31.6,0);
%	\draw[->, thick] (28.4,0) -- (26.6,0);
%	\draw[->, thick] (23.4,0) -- (21.6,0);
	
	\node at (35,-3) {$\yy$};
	\node at (30,-3) {Core Edges};
	\node at (25,-3) {Labels};
	\node at (20,-3) {Partial Adj.\ List};

	% =========================
	% Common Observed Graph
	% =========================
	
	\node at (17.5,2.8) {$\gobs$};
	
	\draw[->, thick] (15,1.8) -- (16.8,2.68);
	\draw[->, thick] (20,1.8) -- (18.2,2.65);
	
\end{tikzpicture}

%% file: distribution_coupling.tex
\begin{tikzpicture}[node distance=0.3cm and 1.8cm, every node/.style={font=\small}]
	
	% Box dimensions and styling
	\def\boxwidth{1.8}
	\def\boxheight{1.2}
	\def\inset{0.4}
	\def\subarrowopacity{0.7}
	\def\subdivideropacity{0.6}
	
	% === Left Blocks ===
	\node[draw, rounded corners, thick, minimum width=\boxwidth cm, minimum height=\boxheight cm] (AL) at (0,3.0) {};
	\node[draw, rounded corners, thick, minimum width=\boxwidth cm, minimum height=\boxheight cm, below=of AL] (BL) {};
	\node[draw, rounded corners, thick, minimum width=\boxwidth cm, minimum height=\boxheight cm, below=of BL] (CL) {};
	
	% === Right Blocks ===
	\node[draw, rounded corners, thick, minimum width=\boxwidth cm, minimum height=\boxheight cm, right=of AL] (AR) {};
	\node[draw, rounded corners, thick, minimum width=\boxwidth cm, minimum height=\boxheight cm, below=of AR] (BR) {};
	\node[draw, rounded corners, thick, minimum width=\boxwidth cm, minimum height=\boxheight cm, below=of BR] (CR) {};
	
	% === External Labels ===
	\node[left=6.5pt of AL] {${\gobs}^{(1)}$};
	\node[left=6.5pt of BL] {${\gobs}^{(2)}$};
	\node[left=6.5pt of CL] {${\gobs}^{(3)}$};
	
	\node[right=6.5pt of AR] {${\gobs}^{(1)}$};
	\node[right=6.5pt of BR] {${\gobs}^{(2)}$};
	\node[right=6.5pt of CR] {${\gobs}^{(3)}$};
	
	% ----------------------
	% Sub-block division lines
	% Pair (1): AL has 3, AR has 4
	\foreach \t in {0.3333,0.6667} {
		\draw[opacity=\subdivideropacity] ($ (AL.north)!\t!(AL.south) + (-0.9,0) $) -- ++(1.8,0);
	}
	\foreach \t in {0.25,0.5,0.75} {
		\draw[opacity=\subdivideropacity] ($ (AR.north)!\t!(AR.south) + (-0.9,0) $) -- ++(1.8,0);
	}
	
	% Pair (2): BL has 4, BR has 3
	\foreach \t in {0.25,0.5,0.75} {
		\draw[opacity=\subdivideropacity] ($ (BL.north)!\t!(BL.south) + (-0.9,0) $) -- ++(1.8,0);
	}
	\foreach \t in {0.3333,0.6667} {
		\draw[opacity=\subdivideropacity] ($ (BR.north)!\t!(BR.south) + (-0.9,0) $) -- ++(1.8,0);
	}
	
	% Pair (3): CL has 3, CR has 3
	\foreach \t in {0.3333,0.6667} {
		\draw[opacity=\subdivideropacity] ($ (CL.north)!\t!(CL.south) + (-0.9,0) $) -- ++(1.8,0);
		\draw[opacity=\subdivideropacity] ($ (CR.north)!\t!(CR.south) + (-0.9,0) $) -- ++(1.8,0);
	}

	% ----------------------
	% Subblock matchings (lighter arrows)
	% Pair (1): AL (3) -> AR (4)
	
	\draw[->, semithick, color=colorTwo, opacity=\subarrowopacity]
	($ (AL.north)!0.1667!(AL.south) + (\inset,0) $) --
	($ (AR.north)!0.125!(AR.south) + (-\inset,0) $);
	
	\draw[->, semithick, color=colorTwo, opacity=\subarrowopacity]
	($ (AL.north)!0.5!(AL.south) + (\inset,0) $) --
	($ (AR.north)!0.6!(AR.south) + (-\inset,0.25) $);
	
	\draw[->, semithick, color=colorTwo, opacity=\subarrowopacity]
	($ (AL.north)!0.8333!(AL.south) + (\inset,0) $) --
	($ (AR.north)!1.03!(AR.south) + (-\inset,0.5) $);
	
	% ----------------------
	% Pair (2): BL (4) -> BR (3)
	
	\draw[->, semithick, color=colorTwo, opacity=\subarrowopacity]
	($ (BL.north)!0.025!(BL.south) + (\inset,-0.5) $) --
	($ (BR.north)!0.1667!(BR.south) + (-\inset,0) $);
	
	\draw[->, semithick, color=colorTwo, opacity=\subarrowopacity]
	($ (BL.north)!0.45!(BL.south) + (\inset,-0.25) $) --
	($ (BR.north)!0.5!(BR.south) + (-\inset,0) $);
	
	\draw[->, semithick, color=colorTwo, opacity=\subarrowopacity]
	($ (BL.north)!0.875!(BL.south) + (\inset,0) $) --
	($ (BR.north)!0.8333!(BR.south) + (-\inset,0) $);

	% Pair (3): CL (3) -> CR (3)
	\foreach \p in {0.1667, 0.5, 0.8333} {
		\draw[->, semithick, color=colorTwo, opacity=\subarrowopacity]
		($ (CL.north)!\p!(CL.south) + (\inset,0) $) --
		($ (CR.north)!\p!(CR.south) + (-\inset,0) $);
	}

	% Extra: CL (2) -> CR (1)

	\draw[->, semithick, color=red, opacity=\subarrowopacity]
	($ (BL.north)!0.125!(BL.south) + (\inset,0) $) --
($ (AR.north)!0.875!(AR.south) + (-\inset,0) $);

	% ----------------------
	% Enclosing rectangles
	\node[draw, rounded corners, fit=(AL)(BL)(CL), inner sep=8pt, label={[xshift=0cm]$\xx$}] {};
	\node[draw, rounded corners, fit=(AR)(BR)(CR), inner sep=8pt, label={[xshift=0cm]$\yy$}] {};
	
	\node at ($ (AL.north)!0.20!(AL.south) + (1.8,0.5) $) {$\mm^*$};

	\node[below=0.8cm of CL] {\small NO};
	\node[below=0.8cm of CR] {\small YES};	
	
	%%%%%%%%%%%%% Legend 	%%%%%%%%%%%%%
	
	\node[below=1.4cm of CL, xshift=2.5cm] {
		\tikz{\draw[->, color=colorTwo, line width=1pt] (0,0)--(0.8,0);} indistinguishable ($\mm$)
		\hspace{1em}
		\tikz{\draw[->, color=red, line width=1pt] (0,0)--(0.8,0);} distinguishable ($\mm'$)
	};

\end{tikzpicture}

%% file: distribution.tex
\begin{tikzpicture}[every node/.style={font=\small}, node distance=0.3cm and 1.8cm]
	
	% Box dimensions and styling
	\def\boxwidth{1.8}
	\def\boxheight{1.2}
	\def\inset{0.4}
	\def\subarrowopacity{0.5}
	\def\subdivideropacity{0.6}
	
	% === Left Blocks ===
	\node[draw, rounded corners, thick, minimum width=\boxwidth cm, minimum height=\boxheight cm] (AL) at (0,3.0) {};
	\node[draw, rounded corners, thick, minimum width=\boxwidth cm, minimum height=\boxheight cm, below=of AL] (BL) {};
	\node[draw, rounded corners, thick, minimum width=\boxwidth cm, minimum height=\boxheight cm, below=of BL] (CL) {};
	
	% === Right Blocks ===
	\node[draw, rounded corners, thick, minimum width=\boxwidth cm, minimum height=\boxheight cm, right=of AL] (AR) {};
	\node[draw, rounded corners, thick, minimum width=\boxwidth cm, minimum height=\boxheight cm, below=of AR] (BR) {};
	\node[draw, rounded corners, thick, minimum width=\boxwidth cm, minimum height=\boxheight cm, below=of BR] (CR) {};
	
	% === External Labels ===
	\node[left=6.5pt of AL] {${\gcno}^{(1)}$};
	\node[left=6.5pt of BL] {${\gcno}^{(2)}$};
	\node[left=6.5pt of CL] {${\gcno}^{(3)}$};
	
	\node[right=6.5pt of AR] {${\gcyes}^{(1)}$};
	\node[right=6.5pt of BR] {${\gcyes}^{(2)}$};
	\node[right=6.5pt of CR] {${\gcyes}^{(3)}$};
	
	% === Sub-block division lines (slightly darker)
	\foreach \block in {AL, AR, BL, BR, CL, CR} {
		\draw[opacity=\subdivideropacity] ($(\block.north)!0.25!(\block.south) + (-0.9,0)$) -- ++(1.8,0);
		\draw[opacity=\subdivideropacity] ($(\block.north)!0.50!(\block.south) + (-0.9,0)$) -- ++(1.8,0);
		\draw[opacity=\subdivideropacity] ($(\block.north)!0.75!(\block.south) + (-0.9,0)$) -- ++(1.8,0);
	}
	
	% === Subblock matchings (lighter arrows)
	\foreach \p in {0.125, 0.375, 0.625, 0.875} {
		\draw[->, semithick, color=colorOne, opacity=\subarrowopacity] 
		($ (AL.north)!\p!(AL.south) + (\inset,0) $) -- 
		($ (AR.north)!\p!(AR.south) + (-\inset,0) $);
		
		\draw[->, semithick, color=colorTwo, opacity=\subarrowopacity] 
		($ (BL.north)!\p!(BL.south) + (\inset,0) $) -- 
		($ (BR.north)!\p!(BR.south) + (-\inset,0) $);
		
		\draw[->, semithick, color=colorSeven, opacity=\subarrowopacity] 
		($ (CL.north)!\p!(CL.south) + (\inset,0) $) -- 
		($ (CR.north)!\p!(CR.south) + (-\inset,0) $);
	}
	
	% === Block-level matchings
	\draw[very thick, ->, >=latex, color=colorOne]   (AL.east) -- (AR.west);
	\draw[very thick, ->, >=latex, color=colorTwo]   (BL.east) -- (BR.west);
	\draw[very thick, ->, >=latex, color=colorSeven] (CL.east) -- (CR.west);
	
	% === Enclosing rectangles
	\node[draw, rounded corners, fit=(AL)(BL)(CL), inner sep=8pt, label={[xshift=0cm]$\xx$}] {};
	\node[draw, rounded corners, fit=(AR)(BR)(CR), inner sep=8pt, label={[xshift=0cm]$\yy$}] {};
	
	\node at ($ (AL.north)!0.20!(AL.south) + (1.8,0.3) $) {$\mm$};
	
\end{tikzpicture}

%% file: distribution2.tex
\begin{tikzpicture}[every node/.style={font=\small}, node distance=0.3cm and 1.8cm]
	
	% Box dimensions and styling
	\def\boxwidth{1.8}
	\def\boxheight{1.2}
	\def\inset{0.4}
	\def\subarrowopacity{0.8}
	\def\subdivideropacity{0.6}
	
	% === Left Blocks ===
	\node[draw, rounded corners, thick, minimum width=\boxwidth cm, minimum height=\boxheight cm] (AL) at (0,3.0) {};
	\node[draw, rounded corners, thick, minimum width=\boxwidth cm, minimum height=\boxheight cm, below=of AL] (BL) {};
	\node[draw, rounded corners, thick, minimum width=\boxwidth cm, minimum height=\boxheight cm, below=of BL] (CL) {};
	
	% === Right Blocks ===
	\node[draw, rounded corners, thick, minimum width=\boxwidth cm, minimum height=\boxheight cm, right=of AL] (AR) {};
	\node[draw, rounded corners, thick, minimum width=\boxwidth cm, minimum height=\boxheight cm, below=of AR] (BR) {};
	\node[draw, rounded corners, thick, minimum width=\boxwidth cm, minimum height=\boxheight cm, below=of BR] (CR) {};
	
	% === External Labels ===
	\node[left=6.5pt of AL] {$\xx_f^1$};
	\node[left=6.5pt of BL] {$\xx_f^2$};
	\node[left=6.5pt of CL] {$\xx_f^3$};
	\node[right=6.5pt of AR] {$\yy_f^1$};
	\node[right=6.5pt of BR] {$\yy_f^2$};
	\node[right=6.5pt of CR] {$\yy_f^3$};
	
	% === Block-level matchings (background)
	\draw[very thick, ->, >=latex, color=colorOne, opacity=0.4]   (AL.east) -- (AR.west);
	\draw[very thick, ->, >=latex, color=colorTwo, opacity=0.4]   (BL.east) -- (BR.west);
	\draw[very thick, ->, >=latex, color=colorSeven, opacity=0.4] (CL.east) -- (CR.west);
	
	% === Sub-block division lines
	\foreach \block in {AL, AR} {
		\draw[opacity=\subdivideropacity] ($(\block.north)!0.40!(\block.south) + (-0.9,0)$) -- ++(1.8,0);
	}
	\foreach \block in {BL, BR} {
		\draw[opacity=\subdivideropacity] ($(\block.north)!0.65!(\block.south) + (-0.9,0)$) -- ++(1.8,0);
	}
	\foreach \block in {CL, CR} {
		\draw[opacity=\subdivideropacity] ($(\block.north)!0.55!(\block.south) + (-0.9,0)$) -- ++(1.8,0);
	}
	
	% === Crossed Subblock Matchings (foreground)
	
	% A (40/60)
	\draw[->, semithick, color=colorOne, opacity=\subarrowopacity]
	($ (AL.north)!0.20!(AL.south) + (\inset,0) $) -- 
	($ (AR.north)!0.80!(AR.south) + (-\inset,0) $);
	\draw[->, semithick, color=colorOne, opacity=\subarrowopacity]
	($ (AL.north)!0.80!(AL.south) + (\inset,0) $) -- 
	($ (AR.north)!0.20!(AR.south) + (-\inset,0) $);
	
	% B (65/35)
	\draw[->, semithick, color=colorTwo, opacity=\subarrowopacity]
	($ (BL.north)!0.325!(BL.south) + (\inset,0) $) -- 
	($ (BR.north)!0.825!(BR.south) + (-\inset,0) $);
	\draw[->, semithick, color=colorTwo, opacity=\subarrowopacity]
	($ (BL.north)!0.825!(BL.south) + (\inset,0) $) -- 
	($ (BR.north)!0.325!(BR.south) + (-\inset,0) $);
	
	% C (55/45)
	\draw[->, semithick, color=colorSeven, opacity=\subarrowopacity]
	($ (CL.north)!0.275!(CL.south) + (\inset,0) $) -- 
	($ (CR.north)!0.775!(CR.south) + (-\inset,0) $);
	\draw[->, semithick, color=colorSeven, opacity=\subarrowopacity]
	($ (CL.north)!0.775!(CL.south) + (\inset,0) $) -- 
	($ (CR.north)!0.275!(CR.south) + (-\inset,0) $);
	
	% === Label inside top subblock of A^L
	\node at ($ (AL.north)!0.20!(AL.south) + (0,0) $) {$\pi$};
	\node at ($ (AL.north)!0.20!(AL.south) + (3.6,0) $) {$\pi$};
	\node at ($ (AL.north)!0.20!(AL.south) + (0,-0.55) $) {$\pi'$};
	\node at ($ (AL.north)!0.20!(AL.south) + (3.6,-0.55) $) {$\pi'$};
	
	\node at ($ (AL.north)!0.20!(AL.south) + (1.8,0.3) $) {$\mm$};
	
	% === Enclosing rectangles
	\node[draw, rounded corners, fit=(AL)(BL)(CL), inner sep=8pt, label={[xshift=0cm]$\xx$}] {};
	\node[draw, rounded corners, fit=(AR)(BR)(CR), inner sep=8pt, label={[xshift=0cm]$\yy$}] {};
	
\end{tikzpicture}

%% file: algorithm.tex
\section{The Algorithm}\label{sec:alg}
In this section, we give an algorithm that gives an $n^{\delta}$-approximation to the maximum matching size using $O(n^{2-2\delta} \log^2 n)$ non-adaptive queries. In particular, this gives us a $\sqrt{n}$-approximation algorithm in $O(n \log^2 n)$ queries. 
\ubthm*

Consider the following algorithm:
\begin{Algorithm}\label{alg:ub}
	An $n^{\delta}$-approximation Algorithm.
	
	\medskip
	
	\textbf{Input:} Non-Adaptive Query access to the Adjacency List of a graph $G=(V,E)$.

	\medskip
	
	\textbf{Output:} An $n^{\delta}$-approximation to the maximum matching size of $G$.
	
	\begin{enumerate}
		\item Query $q:= 80 n^{1-2\delta} \log n$ random neighbors for each vertex in $V$.
		\item If $q< 1$ then sample a random neighbor with probability $q$.
		\item Output the size of the largest matching $M$ in the returned edges. 
	\end{enumerate}
\end{Algorithm}

\paragraph{Simulating random neighbor queries.}
Before going into the analysis we note that since we are in the non-adaptive model, and we do not know vertex degrees, our algorithm is allowed to query random neighbors of a vertex in the adjacency list.
It is very easy to get rid of this assumption if we increase the number of queries by a factor of $O(\log n)$.
We guess the size of the adjacency list of a vertex in powers of $2$ and then make the random queries assuming that the degree of the vertex is exactly the guess.
Let the true size of the list be between $2^{i-1}$ and $2^i$.
When the guess is $2^i$, a random query lands in the NULL part with probability at most $1/2$.
So a constant fraction of the queries will be random queries over the true adjacency list.
Thus, we have simulated random query access with an $O(\log n)$ factor overhead.

\paragraph{Running time.}
As currently stated, the algorithm invokes a maximum matching computation,
which would exceed the claimed $O(n^{2-2\delta}\log^2 n)$ running time.
However, this is unnecessary. Instead, we compute the size of a maximal
greedy matching in the sampled subgraph. This can be done in time linear
in the number of sampled edges, and hence in
$O(n^{2-2\delta}\log^2 n)$ time overall.

Using a maximal matching instead of a maximum matching loses at most a factor of $2$ in the approximation. This constant factor in the approximation can be absorbed into the number of queries by rescaling the approximation factor by $2$ (which increases the number of queries by a factor of $4$).

We now turn to the analysis.
Consider a graph $G=(V,E)$ and let $M^*$ be a maximum matching of $G$.
The main lemma we will prove in this section is the following:
\begin{lemma}\label{lem:ub-matching-size}
	$\card{M} \geq n^{-\delta} \card{M^*}$.
\end{lemma}
If $\card{M^*} \leq n^{\delta}$ then we get an $n^{\delta}$-approximation easily by outputting a single edge of the graph. We will be able to find one edge because we query all the vertices.
Thus, for the rest of the proof, we assume $\card{M^*} \geq n^{\delta}$.
We define $\vlow$ to be the low degree vertices in $V(M^*)$
\[
\vlow := \{ v \in V(M^*) \mid \deg(v) \leq 4n^{1-\delta} \}
\]
and define $\vhigh := V(M^*) -\vlow$ to be the high degree vertices in $V(M^*)$.
We break the rest of the analysis into two cases.
\begin{itemize}
	\item \textbf{Case 1:} At least $\card{M^*}/2$ of the vertices of $M^*$ have degree at most $4n^{1-\delta}$ i.e. $\card{\vlow} \geq \card{M^*}/2$. \\
	In this case, we will find a large fraction of the edges of $M^*$. 
	\item \textbf{Case 2:} At least $\card{M}/2$ of the vertices of $M^*$ have degree strictly larger than $4n^{1-\delta}$ i.e. $\card{\vhigh} \geq \card{M^*}/2$. 
	In this case, we will find many edges because the degrees of the vertices are large.
\end{itemize}

We first consider Case 1 and show that the matching found is large enough.
\begin{claim}\label{clm:ub-case-one}
	If $\card{\vlow} \geq \card{M^*}/2$, then $\card{M} \geq \card{M^*} \cdot n^{-\delta}$ with probability at least $1-1/n$.
\end{claim}
\begin{proof}
	Consider a vertex $v \in \vlow$ and let its edge in the maximum matching be called its \emph{special edge}.
	We make $q$ queries on the vertex $v$.
%	The probability we do not find its special edge in any query is $(1-n^{\delta-1})^{q} \leq \exp(-q \cdot n^{\delta-1}) = \exp(-10 n^{-\delta}) \leq 1 - 5 n^{-\delta}$. The last inequality holds because $\exp(-2x) \leq 1-x$ for $x \in [0,1/2]$.
	The queries find the special edge of $v$ with probability $q/\deg(v) \geq 20n^{-\delta} \log n$.
	If $q<1$ then we make a query with probability $q$ and find the special edge with probability $1/\deg(v)$ getting the same bound.
	This holds true for all vertices in $\vlow$.
	
	We know that $\card{\vlow} \geq \card{M^*}/2$ so we fix any $\card{M^*}/2$ vertices in the set and define the following.
	Let $X_i$ be the random variable which is $1$ if the queries find the special edge for the $i^{th}$ vertex in $\vlow$ and $0$ otherwise, for $i \in [\card{M^*}/2]$.
	We know that $\expect{X_i} = \prob{X_i=1} \geq 20 n^{-\delta} \log n$.
	Let $X= \sum_{i=1}^{\card{M^*}/2} X_i$ be the number of special edges found.
	We have that $X \leq \card{M}$ because we are only looking at a subset of the edges found.
	The expected value of $X$ is $\expect{X} = \sum_i \expect{X_i} \geq 20 n^{-\delta} \log n \cdot \card{M^*}/2 = \mu_L$.
	Each $X_i$ is independent and in $[0,1]$ so we can apply the Chernoff bound (\Cref{prop:chernoff}) with $\eps=1/2$:
	\[
	\prob{X < \mu_L/2} \leq \exp(- \mu_L/10) \leq \exp(- n^{-\delta} \log n \cdot \card{M^*}) \leq n^{-1}.
	\]
	Thus, $\card{M} \geq X \geq \mu_L/2 = 5 n^{-\delta} \log n \cdot \card{M^*}$ with probability $1-1/n$.
\end{proof}

We now consider Case 2 and show that the matching found is large enough in this case too.
\begin{claim}\label{clm:ub-case-two}
	If $\card{\vlow} < \card{M}/2$, then $\card{M} \geq \card{M^*} \cdot n^{-\delta}$ with probability at least $1-1/\poly(n)$.
\end{claim}
\begin{proof}
	In this case, we have at least $\card{M^*}/2$ vertices in $\vhigh$.
	We now consider the following procedure to find a large matching.
	We go over the vertices of $\vhigh$ in an arbitrary but fixed order and maintain a matching in the following inductive way.
	We start with an empty matching $M_0$.
	For the $i^{th}$ vertex in the order, we construct matching $M_i$ by greedily trying to increase the size of $M_{i-1}$ by adding an edge to it from the queried edges for the $i^{th}$ vertex.
	
	If at any point $j$ we find that the size of the matching found $M_j$ is $n^{1-\delta}$ then we stop the process since $\card{M} \geq \card{M_j} = n^{1-\delta} \geq n^{-\delta} \cdot \card{M^*}$, since $M^*$ cannot have more than $n$ edges.

	We now show that if at any point $\card{M_{i-1}} < n^{1-\delta}$ then $\card{M_i} = \card{M_{i-1}}+1$ with probability at least $1/2$.
	This is true because the current number of matched vertices is $\card{V(M_{i-1})} < 2n^{1-\delta}$ and the degree of the $i^{th}$ vertex is at least $4n^{1-\delta}$.
	So even a single query is enough to find an unmatched neighbor with probability at least $1/2$.
	If $q<1$ then we make a query with probability $q$ so we find an unmatched neighbor with probability at least $q/2$.
	Let $p:=\min(1/2,q/2)$ be the probability of finding an unmatched neighbor (this covers both cases). 
	If we find an unmatched neighbor then we can greedily add that edge to the matching to increase its size by $1$.
	
	For $i \in [\card{M^*}/2]$ let $Y_i$ be the random variable which is $1$ if $\card{M_i} = \card{M_{i-1}}+1$ i.e. the $i^{th}$ vertex in the order finds an unmatched neighbor.
	We have $\expect{Y_i} = \prob{Y_i=1} \geq p$.
	Let $Y=\sum_i Y_i$ be the size of the final matching $M_{\card{M^*}/2}$.
	We have 
	\[\expect{Y} = \sum_i \expect{Y_i} \geq p \cdot \card{M^*}/2 \geq \min(\card{M^*}/4,q\cdot n^{\delta}/4) \geq \min(\card{M^*}/4,20 n^{1-\delta}\log n) = \mu_L.
	\]
	Each $Y_i$ is independent and in $[0,1]$ so we can apply the Chernoff bound (\Cref{prop:chernoff}) with $\eps=1/2$:
	\[
	\prob{Y < \mu_L/2} \leq \exp(- \mu_L/10) \leq \exp(- \card{M^*}/40) \leq \exp(-n^{\delta}/40) \ll 1/\poly(n).
	\]
	Thus, $\card{M} \geq Y > \mu_L/2 = \min(\card{M^*}/8,10n^{1-\delta}\log n) \geq n^{-\delta} \card{M^*}$ with probability $1-1/\poly(n)$.
	Note that this analysis will fail if at some step $i$ we find a matching of size $n^{1-\delta}$.
	If this happens then we can only guarantee that the size of the matching $M$ is at least $n^{1-\delta}$.
	Therefore, we get that $\card{M} \geq \min(n^{1-\delta}, n^{-\delta} \card{M^*}) \geq n^{-\delta} \card{M^*}$.
\end{proof}

\begin{proof}[Proof of \Cref{thm:ub}]
The proof of the approximation (\Cref{lem:ub-matching-size}) follows from \Cref{clm:ub-case-one,clm:ub-case-two}.

When $q\ge 1$, \Cref{alg:ub} makes $O(n^{1-2\delta}\log n)$ queries per vertex,
so the base query complexity is $O(n^{2-2\delta}\log n)$ overall.
Accounting for the additional $O(\log n)$ overhead needed to simulate random
neighbor queries (by guessing degrees), the total becomes
$O(n^{2-2\delta}\log^2 n)$.

When $q<1$, we make a query on each vertex independently with probability $q$.
For $i \in [n]$ let $Y_i$ be the random variable which is $1$ if we make a query for vertex $i$.
We have $\expect{Y_i} = \prob{Y_i=1} = q$.
Let $Y=\sum_i Y_i$ be the total number of queries made.
We have $\expect{Y} = \sum_i \expect{Y_i} = q \cdot n = 80n^{2-2\delta}\log n = \mu_H$.
Each $Y_i$ is independent and in $[0,1]$ so we can apply the Chernoff bound (\Cref{prop:chernoff}) with $\eps=1$:
\[
\prob{Y > 2\mu_H} \leq \exp(- \mu_L/4) \leq \exp(- 20n^{2-2\delta}\log n)  \ll 1/\poly(n).
\]
Thus, the number of queries is at most $2\mu_H =O(n^{2-2\delta}\log n)$ with probability $1-1/\poly(n)$.
After the same $O(\log n)$ overhead for simulating random neighbor access, the total number
of adjacency-list queries is $O(n^{2-2\delta}\log^2 n)$.

Finally, for the running time, we compute a maximal greedy matching on the
sampled subgraph instead of a maximum matching. This takes time linear in the
number of sampled edges, which is $O(n^{2-2\delta}\log^2 n)$, matching the
query bound up to constant factors.
\end{proof}

%% file: tree-query.tex
\newcommand{\degIn}[1]{\deg_{\text{in}}\left(#1\right)}
\newcommand{\degOut}[1]{\deg_{\text{out}}\left(#1\right)}
\newcommand{\assign}[2]{#1 \mapsto #2}

\section{Non-Adaptive Tree Probe Model}
In this section, we consider a different notion of non-adaptive queries that was introduced very recently in \cite{azarmehr2025lower}. In this model, access to the adjacency list is structured as follows: the algorithm selects a \emph{root vertex} $r$ and commits to a fixed \emph{query tree} rooted at $r$. It then queries the adjacency list according to this tree—first querying $\deg(r)$ neighbors of the root, then querying neighbors of those vertices, and so on, traversing the entire tree structure.
Note that this model is also non-adaptive because the tree structure is fixed without adapting it based on any observed answers. 
Consider the following formal definition.
\begin{definition}[Non-adaptive Tree Probe Model]
	In this model, the algorithm selects a root vertex $r$ and commits to a sequence of query instructions $(a_1, b_1), (a_2, b_2), \dots, (a_q, b_q)$ such that $a_i \leq i$ and $1 \leq b_i \leq \Delta$. Here, $q$ is the query complexity of the algorithm.
	
	Afterwards, the instructions are used to explore the graph as follows: $u_1 = r$ is the first discovered vertex. Then, for every step $1 \leq i \leq Q$, the newly discovered vertex $u_{i+1}$ is set to the $b_i$-th vertex in the adjacency list of $u_{a_i}$. If $u_{a_i}$ has fewer than $b_i$ neighbors, then $u_{i+1}$ is set to a predesignated null value $\perp$, and any further queries made to $u_{i+1}$ also return $\perp$. The algorithm then computes the output based on the explored subgraph.
\end{definition}
We assume that after all the queries are made, degrees of all vertices are revealed to the algorithm along with all query answers.
We prove that any randomized algorithm operating in this non-adaptive tree probe model still requires $\Omega(n^{1+\eps})$ queries to achieve an $n^{\delta}$-approximation, for any $\eps,\delta>0$ satisfying $n^{2\eps+3\delta} \log^3 n =o(n)$, matching the lower bound in \Cref{thm:lb}.
\treelbthm*

Note that the queries could form a distribution over trees, but we will fix a deterministic tree structure and prove a distributional lower bound. This will imply a lower bound for distributions over trees by Yao’s minimax principle like in \Cref{thm:lb} (see also \cite{azarmehr2025lower}).

We start with some notation. The query tree $T=(V_T,E_T)$ is a directed tree with root vertex $r$.
We will call a node in the query tree as a tree-node. Each edge corresponds to a query, so we have $\card{E_T}=q =n^{1+\eps}$. This implies that $\card{V_T} = q+1$. The in-degree of a tree-node $v$ is denoted by $\degIn{v}$, and its out-degree is denoted by $\degOut{v}$. The in-degree of all tree-nodes is $1$ (an edge to the parent denoted $p(v)$) except for the root $r$ which has an in-degree of $0$. Finally, we use the notation $\assign{x}{v}$ to indicate that tree-node $x$ is assigned vertex $v \in V$ of the input graph and the notation $\assign{x}{D}$ to indicate that tree-node $x$ is assigned to a vertex in a subset $D \subseteq V$ of vertices in the input graph.

In this section, we will prove structural results similar to \Cref{subsec:Limitations-of-Alg}. The coupling argument of \Cref{sec:coupling} will remain untouched and will work as is along with this section.
Finally, to prove the claims we need, we will set the following constraint:
\begin{constraint}\label{constraint:tree-queries}
	n^{2\eps+ 3 \delta}\log^2 n =o(n).
\end{constraint}

We now turn to the formal claims and their proofs. 
We begin by showing that, the neighbor of a vertex in $D$ belongs to $A$ with a small probability.
\begin{claim}\label{clm:tree:degree-D-to-A}
	Let $d \in D$ be an arbitrary vertex and let $v$ be a random neighbor of $d$. The probability that $v\in A$ is at most $2 n^{-\delta}$.
\end{claim}
\begin{proof}
	Since $d\in D$, we know that, with high probability, its degree is in $(1 \pm 0.2) \cdot n/2$ and also its number of edges to $A$ is in $(1 \pm 0.2) \cdot n^{1-\delta}/2$ (\Cref{clm:deg-D-verts}).
	Thus, the probability that a random neighbor of $d$ is a vertex in $A$ is at most $\frac{(1.2) \cdot n^{1-\delta}/2}{(0.8) \cdot n/2} \leq 2 n^{-\delta}$.
\end{proof}

We now show that the probability that any tree-node is a vertex in $A$ is at most $\Ot(n^{-\delta})$.
\begin{claim}\label{clm:prob-being-A}
	The probability that any tree-node is a vertex in $A$ is at most $2 n^{-\delta}$.
\end{claim}
\begin{proof}
	Consider any tree-node $t$. What is the probability that it will be a vertex in $A$? The root node is randomly picked so the probability it is a vertex in $A$ is at most $n^{1-\delta}/n = n^{-\delta}$.
	If we consider that $t$ is not a root node then this probability will depend on the parent of this tree-node.
	If the parent is in $A$ then the largest probability is in the yes case: $n^{\delta}/d^*=2n^{2\delta-1}$.
	If the parent is in $B$ then the largest probability is in the no case: $1/d^*=2n^{\delta-1}$.
	If the parent is in $D$ then the largest probability is: $2 n^{-\delta}$ (\Cref{clm:tree:degree-D-to-A}). The maximum out of all of these is $2 n^{-\delta}$ because of \Cref{constraint:tree-queries}.
%	Thus, the probability that $t \in A$ is at most $2n^{-\delta}$.
%	This is true for all tree-nodes thus, in expectation, at most $2 n^{1+\eps-\delta}$ tree-nodes are vertices in $A$.
%	Applying the Markov inequality we get that with probability $1-o(1)$, at most $2 n^{1+\eps-\delta} \log n$ tree-nodes are vertices in $A$.
\end{proof}

We now observe that, with high probability, most queries return edges to dummy vertices.
\begin{observation}\label{obs:in-core-prob}
	A query to a vertex in $A$ returns a core edge with probability at most $2n^{2\delta-1}$, while a query to a vertex in $B$ returns a core edge with probability at most $2n^{\delta-1}$.
\end{observation}
This is a result of a vertex in $A$ having $n^{\delta}$ core neighbors, a vertex in $B$ having $1$ core neighbor, and the adjacency lists being randomly permuted. 
We now show that the number of directed $2$-paths in $T$ is at most $q$.
\begin{claim}\label{clm:no-directed-two-paths}
	The number of directed $2$-paths in the query tree $T$ is at most $q$.
\end{claim}
\begin{proof}
	The number of directed $2$-paths in the query tree is at most the number of edges in the query tree $q$. Consider the edge $e$ of the $2$-path that is farther from the root. The only edge closer to the root that can form a directed $2$-path with $e$ is its parent. This implies that the number of directed $2$-paths is upper bounded by the number of edges $q$. 
\end{proof}

We now show that no directed $2$-path in the tree can have both edges as core edges.
\begin{claim}\label{clm:no-core-two-path}
	With high probability, no directed $2$-path in the query tree can have both edges as core edges.
\end{claim}
\begin{proof}
	We will prove this separately for the yes and the no case.
	We start with the no case. Consider the $2$-path with tree-nodes $x,y$ and $z$ respectively. 
	Let $x$ be a vertex in $A$. The probability that $y$ is a vertex in the core is at most $2 n^{2\delta-1}$ (\Cref{obs:in-core-prob}). If $y$ is in the core it will be a vertex in $B$ (because we are in the no case) so the probability that $z$ is also in the core is at most $2 n^{\delta-1}$ (\Cref{obs:in-core-prob}). Thus, the probability that this $2$-path is entirely in the core is at most $4n^{3\delta-2}$. The calculation is the same when $x\in B$.
	
	We now want to union bound this over all $2$-paths. The number of directed $2$-paths in the query tree is at most $q$ (\Cref{clm:no-directed-two-paths}).
	Thus, the probability that any directed $2$-path lies in the core is at most $4 n^{3\delta-2} \cdot n^{1+\eps} = 4 n^{3\delta+\eps-1} =o(1)$ (by \Cref{constraint:tree-queries}). Therefore, with probability $1-o(1)$, there are no directed $2$-paths in the query tree that lie in the core.
	
	We now consider the yes case.
	The same argument as above works for $B$-to-$B$ edges, so we focus on $A$-to-$A$ edges.
	Consider the $2$-path $x,y,z$. In this case, for the $2$-path to be in the core we need $x$ to be an $A$ vertex, we need $y$ to lie in the core, and we need $z$ to lie in the core too.
	The probability that $x$ is a vertex in $A$ is at most $2 n^{-\delta}$ (\Cref{clm:prob-being-A}).
	The probability that $y$ is a vertex in the core is at most $2 n^{2\delta-1}$ (\Cref{obs:in-core-prob}). The probability that $z$ is in the core is also at most $2 n^{2\delta-1}$.  Thus, the probability that this $2$-path is entirely in the core $A$ vertices is at most $4n^{4\delta-2} \cdot 2n^{-\delta} = 8n^{3\delta-2}$. 
	
	We now union bound this over all $2$-paths. Thus, the probability that any directed $2$-path lies in the core is at most $n^{1+\eps} \cdot 8n^{3\delta-2} =o(1)$ (by \Cref{constraint:tree-queries}). Therefore, with probability $1-o(1)$, there are no directed $2$-paths in the query tree that lie in the core.
\end{proof}

\Cref{clm:no-core-two-path} implies that every query is essentially picking a random vertex and making a neighbor query on it. Thus, the power of this stronger non-adaptive tree probe model is intuitively the same as the power of our standard non-adaptive query model.
This implies that we should be able to prove the same claims we did in the non-adaptive query model.
We now show the rest of the proof for completeness.
Since the goal is to replicate the structural results from \Cref{subsec:Limitations-of-Alg}, we omit further intuition and explanations, and focus solely on the formal arguments.

\begin{claim}
Consider an arbitrary vertex $v\in C$ in the core. The probability $v$ is sampled as a neighbor of a dummy vertex $d \in D$ in a single neighbor query is at least $1/4n$ and at most $2.5/n \leq 4/n$.	
\end{claim}
\begin{proof}
The probability that $v$ is sampled as a neighbor of $d \in D$ in a single neighbor query is at most $1/\deg(d) \leq (2/0.8n) \leq 4/n$, implying the upper bound.

For $v$ to be sampled as a neighbor of a dummy vertex $d$, $(v,d)$ should be an edge and the query should pick $v$ out of the $\deg(d)$ neighbors in the adjacency list.
The probability that $v$ has an edge to $d$ is at least $\frac{d^*-n^{\delta}}{n^{1-\delta}} \geq \frac{d^*/2}{n^{1-\delta}}=1/4$. The probability that the query picks $v$ out of the $\deg(d)$ neighbors is $1/\deg(d) \geq (2/1.2 n) \geq 1/n$. This implies that the probability $v$ is picked as the query answer is at least $1/4n$.
\end{proof}

\begin{claim}\label{clm:tree:no-edges-core}
	The probability that a core edge $(u,v)$ is observed is at most $16 n^{\eps+\delta-1}$.
	Also, the number of core edges observed is at most $4 n^{\eps+\delta} \log^{1/2} n$ in $\dyes$ and $\dno$ with high probability.
\end{claim}
\begin{proof}
We just count the number of core edges observed where their parent in the tree is a vertex in $D$. We already know that the number of core edges observed where their parent in the tree is not a vertex in $D$ is $0$ with high probability (\Cref{clm:no-core-two-path}).
	
Fix an edge $(u,v)$ in the core. What is the probability that this edge is observed? For that $u$ has to be sampled and then the edge $(u,v)$ has to be picked.
Fix a tree-edge $(x,y)$. The probability that $\assign{x}{u}$ is at most $4/n$ and the probability that $\assign{y}{v}$ is exactly $1/d^*$. Thus, the probability the tree edge $(x,y)$ is $(u,v)$ is at most $16 n^{\delta-2}$ because the edge could be observed in the other direction. Thus, the probability that this fixed core edge $(u,v)$ is observed over all query tree edges is at most $q \cdot 16 n^{\delta-2} = 16 n^{\eps+\delta-1}$.

There are $2n$ core edges in the graph. Therefore, the expected number of core edges observed is $32 n^{\eps+\delta}$.
Markov's inequality implies that with probability $1-o(1)$, the number of core edges observed is at most $4 n^{\eps+\delta} \log^{1/2} n$.
\end{proof}

\begin{claim}\label{clm:tree:no-happy-edge}
	With high probability, the incoming endpoint of every observed edge in the core is not a happy vertex, in both $\dcyes$ and $\dcno$.
\end{claim}
\begin{proof}
	Consider any observed edge $u \rightarrow v$. We will show that $v$ is not happy.
	$v$ is happy if a query on $v$ returns an edge in the core. We know that no directed $2$-path in $T$ can lie in the core. Thus, for $v$ to be happy we need a tree-node $x$ such that $\assign{x}{v}$, $\assign{p(x)}{D}$, and one of the queries on $x$ observes an edge in the core.
	 
	The probability that a query on $v$ lands in the core is at most $n^{\delta}/d^*$. 
	The number of queries on $v$ is the sum of the out degrees of all the tree-nodes $x$ where $x$ is the vertex $v$.
	Thus, the probability $v$ is happy is:
	\begin{align*}
		\prob{\text{$v$ is happy}} &\leq \sum_{x\in V_T: \assign{p(x)}{D}} \prob{\assign{x}{v}} \cdot \degOut{x} \cdot (n^{\delta}/d^*) \\
		&\leq (4/n) \cdot (n^{\delta}/d^*)  \sum_{x\in V_T}  \degOut{x} \\
		&= 8 n^{2\delta-2} \cdot n^{1+\eps}.
	\end{align*}
	
	We now union bound this probability over all $4 n^{\eps+\delta} \log n$ observed edges (\Cref{clm:tree:no-edges-core}).
	Thus, the probability that any endpoint of an observed edge is happy is at most $8 n^{2\delta+\eps-1} \cdot 4 n^{\eps+\delta} \log n =o(1)$ (\Cref{constraint:tree-queries}).
\end{proof}

\begin{claim}\label{clm:tree:no-happy-pairs}
	With high probability, no pair of vertices in the core observe an edge to the same neighbor for $\dcyes$ and $\dcno$.
\end{claim}
\begin{proof}
	Consider any observed edge $u \rightarrow v$. We will show that there is no other observed edge incoming to $v$.
	$v$ has at most $n^{\delta}$ neighbors.
	A bad event would be $w \in N(v)$ observing an edge to $v$.
	Note that the parent of $w$ in the tree has to be a vertex in $D$ because we know that no directed $2$-path in $T$ can lie in the core. Thus, to observe an edge to $v$, we need a tree-node $x$ such that $\assign{x}{N(v)}$, $\assign{p(x)}{D}$, and one of the queries on $x$ observes the edge to $v$.
	
	The probability that a query on $w \in N(v)$ observes an edge to $v$ is $1/d^*$. 
	The number of queries on $N(v)$ is the sum of the out degrees of all the tree-nodes $x$ where $x$ is assigned a vertex in $N(v)$.
	Thus, the probability we observe an edge to $v$ is:
	\begin{align*}
		\prob{\text{observe edge to $v$}} &\leq \sum_{x\in V_T: \assign{p(x)}{D}} \prob{\assign{x}{N(v)}} \cdot \degOut{x} \cdot (1/d^*) \\
		&\leq (4n^{\delta}/n) \cdot (1/d^*)  \sum_{x\in V_T}  \degOut{x} \\
		&= 8 n^{2\delta-2} \cdot n^{1+\eps}.
	\end{align*}
	
	We now union bound this probability over all $4 n^{\eps+\delta} \log n$ observed edges (\Cref{clm:tree:no-edges-core}).
	Thus, the probability that any endpoint of a observed edge is happy is at most $8 n^{2\delta+\eps-1} \cdot 4 n^{\eps+\delta} \log n =o(1)$ (\Cref{constraint:tree-queries}).
\end{proof}

\begin{lemma}\label{lem:tree:disj-union-stars}
	Let $G \sim \mathcal{D}$ and consider the query tree $T$ with $n^{1+\eps}$ adjacency list queries.
	Then, with high probability, the edges $\Alg$ finds inside the core ($\eobs^C$) forms a disjoint union of stars.
\end{lemma}
\begin{proof}
	We condition on the high probability events in \Cref{clm:tree:no-happy-edge} and \Cref{clm:tree:no-happy-pairs}.
	Consider a happy vertex $v$ and its observed edge $e = v \rightarrow u$.
	We know from \Cref{clm:tree:no-happy-edge} that the observed endpoint $u$ of an observed edge cannot be happy.
	Thus, there are no observed edges going out of $u$.
	Also, we know by \Cref{clm:tree:no-happy-pairs} that no pair of vertices in the core observe an edge to the same vertex.
	Thus, there are no other observed edges incident on $u$.
	Therefore, the edges observed in the core ($\eobs^C$) forms a disjoint union of stars with high probability.
\end{proof}

%We partition the tree-nodes of the query tree $T$ into at most $2\log n$ groups based on their degrees.
%Let $S_i$ be the set of tree nodes with degrees in the range $2^{i}$ to $2^{i+1}$.
%We now look at $\tau$-heavy vertices. 
\begin{claim}\label{clm:tree:discovered-tau-heavy}
	The number of edges observed with $\tau$-heavy petal vertices is at most $8n^{2\eps+\delta} \log^2 n/\tau$ with probability $1-1/\log^{1.5} n$.
\end{claim}
\begin{proof}
	Consider two vertices $u, v \in A$. By symmetry, the probability that $u$ is $\tau$-heavy is the same as that for $v$. Let $p_A$ denote the probability that a vertex $v \in A$ is $\tau$-heavy. Similarly, let $p_B$ denote the probability that a vertex $v \in B$ is $\tau$-heavy.
	
	We know that $p_B \geq p_A$ because vertices in $D$ are more likely to sample a vertex $u \in B$ than a vertex $v \in A$, as $u$ has a higher degree to $D$ than $v$ does.
	We know that the number of queries on all $\tau$-heavy vertices is bounded by $n^{1+\eps}$.
	Thus, $n \cdot p_B \cdot \tau \leq n^{1+\eps}$. This implies $p_B \leq n^{\eps}/\tau$.
	This also implies that $p_A \leq n^{\eps}/\tau$.
	We can also show a similar bound using Markov's inequality because the expected number of queries on a vertex is $\Theta(n^{\eps})$.
	
	There are at most $4 n^{\eps+\delta} \log^{1/2} n$ observed core edges. The probability that their endpoints are heavy is at most $p_B$. This implies that in expectation there are at most $4 n^{\eps+\delta} \log^{1/2} n \cdot n^{\eps}/\tau$ petal vertices that are $\tau$-heavy.
	Thus, the probability that there are more than $8n^{2\eps+\delta} \log^2 n/\tau$ heavy vertices is at most $1/\log^{1.5} n$.
\end{proof}

\begin{lemma}\label{clm:tree:discovered-heavy-edges}
	With high probability, for all $i \in [\alpha,\beta]$, the number of edges observed with $2^i$-heavy petal vertices is at most $2^{-i} \cdot 8n^{2\eps+\delta} \log^2 n$ for $\alpha= \log (n^{\eps} \log n)$ and $\beta = \log (16 n^{2\eps+\delta} \log^2 n)$.
\end{lemma}
\begin{proof}
	We condition on the event in \Cref{clm:tree:discovered-tau-heavy} for $\tau = 2^i$, for all $i \in [\alpha, \beta]$. 
	To account for all such values of $\tau$, we apply a union bound over the failure probabilities of all the cases. 
	Since the number of cases are at most $\beta \leq \log n$ (by \Cref{constraint:no-happy-pairs}), the total failure probability is at most $\log n \cdot (1/\log^{1.5} n) = o(1)$.	
\end{proof}

\begin{claim}\label{clm:tree:indeg-from D}
	Every vertex $v \in A\cup B$ has in-degree at most $4n^{\eps}$ from $D$ in the observed graph with high probability.
\end{claim}
\begin{proof}
	Fix a vertex $v \in C$. It has at most $d^*$ edges to $D$.
	Consider a query edge $(x,y)\in E_T$ where $\assign{x}{D}$. The probability that this edge goes to $v$ is at most $2/0.8n=2.5/n$.
	There are at most $q=n^{1+\eps}$ query edges $(x,y)\in E_T$ where $\assign{x}{D}$.
	Let $X_i$ be a random variable that is $1$ iff the $i^{th}$ query edge $(x,y)\in E_T$ where $\assign{x}{D}$, finds an edge to $v$ i.e. $\assign{y}{v}$ and $0$ otherwise.
	We know $\prob{X_i=1} \leq 2.5/n$.
	
	Observe that many $X_i$'s can correspond to the same vertex $d \in D$ so they may not be independent. We let $Y_j$ for $j \in [n^{1-\delta}]$ be a random variable that is $1$ iff the $j^{th}$ vertex in $D$ observes the edge to $v$ and $0$ otherwise. 
	Note that the $Y_j$'s form a disjoint sum of the $X_i$'s (each $X_i$ is assigned to exactly one $Y_j$).
	Also, we know that each $Y_j$ is independent because the random permutations on the adjacency lists of all vertices in $D$ are independent.
	
	Let $Y=\sum_j Y_j$ be the random variable denoting the number of observed edges from $D$ to $v$.
	We have $\expect{Y} = \sum_j \expect{Y_j} = \sum_i \expect{X_i} \leq n^{1+\eps} \cdot 2.5/n = 2.5n^{\eps}=\mu_H$.
	Since $Y_j$'s are independent and in $[0,1]$, we can apply the Chernoff bound (\Cref{prop:chernoff}):
	\[
	\prob{Y > 1.1 \mu_H} \leq \exp(-\mu_H/310) \ll 1/\poly(n).
	\] 
	Therefore, the in-degree of $v$ from $D$ is at most $4 n^{\eps}$ with probability $1-1/\poly(n)$.
	We can union bound this over vertices $v\in C$ and get that with high probability, all the vertices have in-degree at most $4n^{\eps}$ from $D$.
\end{proof}

Thus, we have proven \Cref{clm:tree:no-edges-core,lem:tree:disj-union-stars,clm:tree:discovered-heavy-edges,clm:tree:indeg-from D}, whose statements are identical to those of \Cref{clm:no-edges-core,lem:disj-union-stars,clm:discovered-heavy-edges,clm:indeg-from D}, respectively.
For the rest of the proof we condition on the high probability events in these claims along with \Cref{clm:deg-D-verts}.
The coupling argument in \Cref{sec:coupling} carries over directly, since the claims proven in this section are structurally identical to those in \Cref{subsec:Limitations-of-Alg}.
This establishes a distributional lower bound: any deterministic algorithm making $n^{1+\eps}$ queries can distinguish between $\dyes$ and $\dno$ with probability at most $1/2 + o(1)$. 
Applying Yao’s minimax principle then implies the desired randomized lower bound, completing the proof of \Cref{thm:tree:lb} (in the same way as the proof of \Cref{thm:lb}).

\begin{remark}
	Our lower bound also holds in an even more general model that we call the \emph{Non-Adaptive Forest Probe Model}. In this setting, the algorithm is allowed to initiate multiple query trees rather than being limited to a single root and tree structure. This model strictly generalizes both the standard Non-Adaptive model and the Non-Adaptive Tree Probe Model, yet the hardness result remains unchanged.
\end{remark}